\documentclass[12pt]{article}
\usepackage{natbib,amsmath,amssymb,amsfonts,amsthm,dsfont,enumitem,float}
\RequirePackage[colorlinks,citecolor=blue,urlcolor=blue,breaklinks]{hyperref}
\usepackage[normalem]{ulem}
\usepackage{epsfig,longtable,xcolor,algorithm,algpseudocode}

\makeatletter

\newcommand{\bfm}[1]{\ensuremath{\mathbf{#1}}}
\def\ba{\bfm a}     \def\bA{\bfm A}     \def\cA{{\cal  A}}     
     \def\bB{\bfm B}          
               
     \def\bD{\bfm D}     \def\cD{{\cal  D}}     
\def\be{\bfm e}     \def\bE{\bfm E}          
         \def\cF{{\cal  F}}     
     \def\bG{\bfm G}          
               
     \def\bI{\bfm I}     \def\cI{{\cal  I}}     
          \def\cJ{{\cal  J}}     
     \def\bK{\bfm K}          
     \def\bL{\bfm L}          
               
          \def\cN{{\cal  N}}     
     \def\bO{\bfm O}          
     \def\bP{\bfm P}          
     \def\bQ{\bfm Q}          
     \def\bR{\bfm R}          
\def\bs{\bfm s}     \def\bS{\bfm S}     \def\cS{{\cal  S}}     
\def\bt{\bfm t}          \def\cT{{\cal  T}}     
\def\bu{{\bfm u}}     \def\bU{\bfm U}          
\def\bv{\bfm v}     \def\bV{\bfm V}          
\def\bw{\bfm w}     \def\bW{\bfm W}          
\def\bx{\bfm x}     \def\bX{\bfm X}          
\def\by{\bfm y}     \def\bY{{\bfm Y}}          
\def\bz{\bfm z}     \def\bZ{\bfm Z}     \def\cZ{{\cal  Z}}     
\def\bzero{\bfm 0}

\newcommand{\bfsym}[1]{\ensuremath{\boldsymbol{#1}}}
\def \balpha   {\bfsym{\alpha}}       \def \bbeta    {\bfsym{\beta}}

\def \betta    {\bfsym{\eta}}

         \def \bomega   {\bfsym{\omega}}

\def \bGamma   {\bfsym{\Gamma}}       
\def \bTheta   {\bfsym{\Theta}}       
          
\def \bSigma   {\bfsym{\Sigma}}       
         
\def \bOmega   {{\bfsym{\Omega}}}

\DeclareMathOperator*{\argmin}{argmin}

\DeclareMathOperator{\var}{Var}
\DeclareMathOperator{\diag}{diag}
\newcommand{\E}{\mathbb{E}}
\DeclareMathOperator{\rank}{rank}

\DeclareMathOperator \Tr {\mathrm{tr}}
\DeclareMathOperator \col {\text{Col}}

\def \soft{\text{soft}}
\def \bone   {\bfsym{1}}
\def \RR	{\mathbb{R}}
\def \PP {\mathbb{P}}
\def \SS {\mathbb{S}}
\def \OO {\mathbb{O}}
\def \NN {\mathbb{N}}

\def \KL {\text{KL}}

\newcommand{\beq}  {\begin{equation}}
\newcommand{\eeq}  {\end{equation}}
\newcommand{\beqn} {\begin{eqnarray}}
\newcommand{\eeqn} {\end{eqnarray}}
\newcommand{\beqnn}{\begin{eqnarray*}}
\newcommand{\eeqnn}{\end{eqnarray*}}

\theoremstyle{definition}

\theoremstyle{plain}
\newtheorem{lem}{Lemma}

\newtheorem{prop}{Proposition}
\newtheorem{thm}{Theorem}

\newcommand{\lonenorm}[1]{\lVert#1\rVert_1}
\newcommand{\ltwonorm}[1]{\lVert#1\rVert_2}
\newcommand{\linfinity}[1]{\lVert#1\rVert_{\infty}}
\newcommand{\opnorm}[1]{\|#1\|_{\mathrm{op}}}
\newcommand{\fnorm}[1]{\|#1\|_{\mathrm{F}}}
\newcommand{\nnorm}[1]{\lVert#1\rVert_*}

\newcommand{\twotoinf}[1]{\lVert#1\rVert_{2\rightarrow \infty}}
\newcommand{\onetoone}[1]{\lVert#1\rVert_{1\rightarrow 1}}
\newcommand{\inftoinf}[1]{\lVert#1\rVert_{\infty \rightarrow\infty}}

\def \bell {\bfsym{\ell}}

\title{High-dimensional principal component analysis with heterogeneous missingness}
\author{Ziwei Zhu$^\ast$, Tengyao Wang$^{\ast,\dagger}$ and Richard J. Samworth$^{\ast}$ \\
 \normalsize
 $^\ast$Statistical Laboratory, University of Cambridge \\ \normalsize
 $^\dagger$Department of Statistical Science, University College London \\
\date{\today} }

\begin{document}

\maketitle

\begin{abstract}
  We study the problem of high-dimensional Principal Component Analysis (PCA) with missing observations.  In simple, homogeneous missingness settings with a noise level of constant order, we show that an existing inverse-probability weighted (IPW) estimator of the leading principal components can (nearly) attain the minimax optimal rate of convergence. 
  However, deeper investigation reveals both that, particularly in more realistic settings where the missingness mechanism is heterogeneous, the empirical performance of the IPW estimator can be unsatisfactory, and moreover that, in the noiseless case, it fails to provide exact recovery of the principal components.  Our main contribution, then, is to introduce a new method for high-dimensional PCA, called \texttt{primePCA}, that is designed to cope with situations where observations may be missing in a heterogeneous manner.  Starting from the IPW estimator, \texttt{primePCA} iteratively projects the observed entries of the data matrix onto the column space of our current estimate to impute the missing entries, and then updates our estimate by computing the leading right singular space of the imputed data matrix.  It turns out that the interaction between the heterogeneity of missingness and the low-dimensional structure is crucial in determining the feasibility of the problem.  We therefore introduce an incoherence condition on the principal components and prove that in the noiseless case, the error of \texttt{primePCA} converges to zero at a geometric rate when the signal strength is not too small.  An important feature of our theoretical guarantees is that they depend on average, as opposed to worst-case, properties of the missingness mechanism.  Our numerical studies on both simulated and real data reveal that \texttt{primePCA} exhibits very encouraging performance across a wide range of scenarios.

\end{abstract}

\section{Introduction}
One of the ironies of working with Big Data is that missing data play an ever more significant role, and often present serious difficulties for analysis.  For instance, a common approach to handling missing data is to perform a so-called \emph{complete-case analysis} \citep{LRu14}, where we restrict attention to individuals in our study with no missing attributes.  When relatively few features are recorded for each individual, one can frequently expect a sufficiently large proportion of complete cases that, under an appropriate missing at random hypothesis, a complete-case analysis may result in only a relatively small loss of efficiency.  On the other hand, in high-dimensional regimes where there are many features of interest, there is often such a small proportion of complete cases that this approach becomes infeasible.  As a very simple illustration of this phenomenon, imagine an $n \times p$ data matrix in which each entry is missing independently with probability $0.01$.  When $p = 5$, a complete-case analysis would result in around 95$\%$ of the individuals (rows) being retained, but even when we reach $p=300$, only around 5$\%$ of rows will have no missing entries.

The inadequacy of the complete-case approach in many applications has motivated numerous methodological developments in the field of missing data over the past 60 years or so, including imputation \citep{For83, Rub04}, factored likelihood \citep{And57} and Expectation-Maximisation approaches \citep{DLR77}; see, e.g., \citet{LRu14} for an introduction to the area.  Recent years have also witnessed increasing emphasis on understanding the performance of methods for dealing with missing data in a variety of high-dimensional problems, including sparse regression \citep{LMa12,BRT17}, classification \citep{Cai18}, sparse principal component analysis \citep{Elsener18} and covariance and precision matrix estimation \citep{Lou14,LT18}.


In this paper, we study the effects of missing data in one of the canonical problems of high-dimensional data analysis, namely dimension reduction via Principal Component Analysis (PCA).  This is closely related to the topic of \emph{matrix completion}, which has received a great deal of attention in the literature over the last decade or so \citep[e.g.][]{CRe09, CPl10, KMO10, MHT10, KLT11, CLM11, NWa12}.  There, the focus is typically on accurate recovery of the missing entries, subject to a low-rank assumption on the signal matrix; by contrast, our focus is on estimation of the principal eigenspaces.  Previously proposed methods for low-dimensional PCA with missing data include non-linear iterative partial least squares \citep{WoldLyttkens1969}, iterative PCA \citep{Kiers1997,JosseHusson2012} and its regularised variant \citep{JPH2009}; see \citet{DrayJosse2015} for a nice survey and comparative study.  More broadly, the \texttt{R-miss-tastic} website \texttt{https://rmisstastic.netlify.com/} provides a valuable resource on methods for handling missing data.

The importance of the problem of high-dimensional PCA with missing data derives from its manifold applications.  For instance, in many commercial settings, one may have a matrix of customers and products, with entries recording the number of purchases.  Naturally, there will typically be a high proportion of missing entries.  Nevertheless, PCA can be used to identify items that distinguish the preferences of customers particularly effectively, to make recommendations to users of products they might like and to summarise efficiently customers' preferences.  Later, we will illustrate such an application, on the Million Song Dataset, where we are able to identify particular songs that have substantial discriminatory power for users' preferences as well as other interesting characteristics of the user database.  Other potential application areas include health data, where one may seek features that best capture the variation in a population, and where the corresponding principal component scores may be used to cluster individuals into subgroups (that may, for instance, receive different treatment regimens).  

To formalise the problem we consider, suppose that the (partially observed) matrix $\bY \in \mathbb{R}^{n\times d}$ is of the form 
\begin{equation}
\label{Eq:ModelA}
\bY = \bX + \bZ, 
\end{equation}
for independent random matrices $\bX$ and $\bZ$, where $\bX$ is a low-rank matrix and $\bZ$ is a noise matrix with independent and identically distributed subgaussian entries having zero mean and unit variance. The low-rank property of $\bX$ is encoded through the assumption that it is generated via 
\begin{equation}
\label{Eq:ModelB}
	\bX = \bU\bV_K^\top,
\end{equation}
where $\bV_K\in\mathbb{R}^{d\times K}$ has orthonormal columns and $\bU$ is a random $n \times K$ matrix (with $n > K$) having independent and identically distributed rows with mean zero.

We are interested in estimating the column space of $\bV_K$, denoted by $\col(\bV_K)$, which is also the $K$-dimensional leading eigenspace of $\bSigma_{\by} := n^{-1}\mathbb{E} \bY^\top \bY$.  \citet{CKR15} considered a different but related model where $\bU$ in~\eqref{Eq:ModelB} is deterministic, and is not necessarily centred, so that $\bV_K$ is the top $K$ right singular space of $\mathbb{E}(\bY)$.  (By contrast, in our setting, $\mathbb{E}(\bY) = \bzero$, so the mean structure is uninformative for recovering $\bV_K$.)  In the context of \emph{$p$-homogeneous missingness}, where each entry of $\bY$ is observed independently with probability $p \in (0,1)$, \citet{CKR15} proposed to estimate $\col(\bV_K)$ by $\col(\widehat{\bV}_K)$, where $\widehat{\bV}_K$ is a simple estimator formed as the top $K$ eigenvectors of an inverse-probability weighted (IPW) version of the sample covariance matrix (here, the weighting is designed to achieve approximate unbiasedness).  Our first contribution, in Section~\ref{Sec:NoIncoherence}, is to provide a detailed, finite-sample analysis of this estimator in the model given by~\eqref{Eq:ModelA} and~\eqref{Eq:ModelB}, with a noise level of constant order.  The differences between the settings necessitate completely different arguments, and reveal in particular a new phenomenon in the form of a phase transition in the attainable risk bound for the $\sin\Theta$ loss function, i.e.~the Frobenius norm of the diagonal matrix of the sines of the principal angles between $\widehat{\bV}_K$ and $\bV_K$.  Moreover, we also provide a minimax lower bound in the case of estimating a single principal component, which reveals that this estimator achieves the minimax optimal rate up to a poly-logarithmic factor.

While this appears to be a very encouraging story for the IPW estimator, it turns out that it is really only the starting point for a more complete understanding of high-dimensional PCA.  For instance, in the noiseless case, the IPW estimator fails to provide exact recovery of the principal components.  Moreover, it is the norm rather than the exception in applications that the missingness mechanism is \emph{heterogeneous}, in the sense that the probability of observing entries of $\bY$ varies (often significantly) across columns.  For instance, in recommendation systems, some products will typically be more popular than others, and hence we observe more ratings in those columns.  As another example, in meta-analyses of data from several studies, it is frequently the case that some covariates are common across all studies, while others appear only in a reduced proportion of them.  In Section~\ref{sec:2.2}, we present an example to show that PCA algorithms can break down entirely for such heterogeneous missingness mechanisms when individual coordinates in $\bV_K$ may be large in absolute value.  Intuitively, if we do not observe the interaction between the $j$th and $k$th columns of $\bY$, then we cannot hope to estimate the $j$th or $k$th rows of $\bV_K$, and this will cause substantial error if these rows of $\bV_K$ contain significant signal.  This example illustrates that it is only possible to handle heterogeneous missingness in high-dimensional PCA with additional structure, and indicates that it is natural to assume \emph{incoherence} among the entries of $\bV_K$; i.e.~no single coordinate of $\bV_K$ is too large in absolute value.

Our main contribution, then, is to propose a new, iterative algorithm, called \texttt{primePCA} (short for \underline{p}rojected \underline{r}efinement for \underline{i}mputation of \underline{m}issing \underline{e}ntries in \underline{P}rincipal \underline{C}omponent \underline{A}nalysis), in Section~\ref{Sec:WithIncoherence}, to estimate $\bV_K$ under this incoherence assumption, even with heterogeneous missingness.  The initialiser for this algorithm is a modified version of the simple estimator discussed above, where the modification accounts for potential heterogeneity.  Each iteration of \texttt{primePCA} projects the observed entries of $\bY$ onto the column space of the current estimate of $\bV_K$ to impute missing entries, and then updates our estimate of $\bV_K$ by computing the leading right singular space of the imputed data matrix.  We show that in the noiseless setting, i.e., $\bZ = \bzero$, \texttt{primePCA} achieves exact recovery of the principal eigenspaces (with a geometric convergence rate) when the initial estimator is close to the truth and a sufficiently large proportion of the data are observed.  Moreover, we also provide a performance guarantee for the initial estimator, showing that it satisfies the desired requirement with high probability, conditional on any observed missingness pattern.  Code for our algorithm is available in the \texttt{R} package \texttt{primePCA} \citep{ZWS19}. 

To the best of our knowledge, \texttt{primePCA} is the first method for high-dimensional PCA that is designed to cope with settings where the missingness mechanism is heterogeneous.  Indeed, the previously mentioned works on high-dimensional PCA and other high-dimensional statistical problems  with missing data have either focused on a uniform missingness setting or have imposed a lower bound on entrywise observation probabilities, which reduces to this uniform case.  A key contribution of our work is to account explicitly for the the effect of a heterogeneous missingness mechanism, where the estimation error depends on average entrywise missingness rather than worst-case missingness; see the discussions after Theorem~\ref{thm:init_2toinf} and Proposition~\ref{prop:init_f} below. In Section~\ref{Sec:Simulations}, the empirical performance of \texttt{primePCA} is compared both with that of the initialiser, and with a popular method for matrix completion called \texttt{softImpute} \citep{MHT10, HML15}, which solves a nuclear-norm regularised optimisation problem, and which can be adapted to provide an estimate of $\bV_K$.  It turns out that in many settings, \texttt{primePCA} outperforms the \texttt{softImpute} algorithm, even when the latter is allowed access to the oracle choice of regularisation parameter for each dataset.  Our analysis of the Million Song Dataset is given in Section~\ref{Sec:RealData}.  In Section~\ref{Sec:Discussion}, we illustrate how some of the ideas in this work may be applied to other high-dimensional statistical problems involving missing data.  Proofs of our main results are deferred to Section~\ref{Sec:Proofs}; auxiliary results and their proofs are given in Section~\ref{Sec:Auxiliary}.

\subsection{Notation}
For a positive integer $T$, we write $[T] := \{1, \ldots, T\}$. For $\bv = (v_1,\ldots,v_d)^\top \in \RR^d$ and $p \in [1, \infty)$, we define $\|\bv\|_p := \bigl(\sum_{j = 1}^d |v_j|^p\bigr)^{1/p}$ and $\|\bv\|_{\infty} := \max_{j \in [d]} |v_j|$. We let $\mathcal{S}^{d-1} := \{\bu \in \mathbb{R}^d: \|\bu\|_2 = 1\}$ denote the unit Euclidean sphere in $\mathbb{R}^d$.  Given $\bu = (u_1,\ldots,u_d)^\top, \bv=(v_1,\ldots,v_d)^\top \in \RR^d$, we denote their Hamming distance by $d_{\mathrm{H}}(\bu, \bv) := \sum_{j=1}^d \mathds{1}_{\{u_j \neq v_j\}}$. Also, we write $\be_j\in\mathbb{R}^d$, $j\in[d]$, for the standard basis vector along the $j$th coordinate axis and $\mathbf{1}_{d}$ for the all-one vector in $\mathbb{R}^d$.

Given $\bu = (u_1,\ldots,u_d)^\top \in \mathbb{R}^d$, we write $\mathrm{diag}(\bu) \in \mathbb{R}^{d \times d}$ for the diagonal matrix whose $j$th diagonal entry is $u_j$. We let $\SS^{d\times d}$ denote the set of symmetric matrices in $\mathbb{R}^{d \times d}$ and let $\OO^{d_1\times d_2}$ denote the set of matrices in $\mathbb{R}^{d_1 \times d_2}$ with orthonormal columns.  For a matrix $\bA = (A_{ij}) \in \RR^{d_1 \times d_2}$, and $p,q\in[1,\infty]$, we write  $\|\bA\|_p:=\bigl(\sum_{i,j} |A_{ij}|^p\bigr)^{1/p}$ if $1\leq  p <\infty$ and $\|\bA\|_\infty:=\max_{i,j}|A_{ij}|$ for its entrywise $\ell_p$ norm, and $\|\bA\|_{p\to q} := \sup_{\|\bv\|_p = 1} \|\bA \bv\|_q$ for its $p$-to-$q$ operator norm, where $\bA$ is viewed as a representation of a linear map from $(\mathbb{R}^{d_1},\|\cdot\|_p)$ to $(\mathbb{R}^{d_2},\|\cdot\|_q)$. We provide special notation for the (Euclidean) operator norm and the Frobenius norm by writing $\|\bA\|_{\text{op}} := \|\bA\|_{2\to 2}$ and $\|\bA\|_{\text{F}}:=\|\bA\|_2$, respectively, and also write $\|\bA\|_*$ for the nuclear norm.  If $\bA \in \mathbb{S}^{d \times d}$ has the eigendecomposition $\bA = \bQ \, \mathrm{diag}(\mu_1,\ldots,\mu_d) \bQ^\top$ for some $\bQ\in\mathbb{O}^{d\times d}$ and $\mu_1\geq \cdots\geq \mu_d$, we write $\lambda_k(\bA):=\mu_k$ for its $k$th largest eigenvalue and abuse terminology slightly to refer to the leftmost $k$ columns of $\bQ$ as the top $k$ eigenvectors of $\bA$ when $\mu_k>\mu_{k+1}$. Also, we write $|\bA| := \bQ \, \mathrm{diag}(|\mu_1|,\ldots,|\mu_d|) \bQ^\top$. For any $\bA \in \RR^{n \times d}$ with singular value decomposition $\bA = \bU\diag(\mu_1, \ldots, \mu_r) \bV^\top$, where $\bU \in \OO^{n \times r}$, $\bV \in \OO^{d \times r}$ and $\mu_1 \geq \ldots \geq \mu_r > 0$, we write $\sigma_k(\bA):=\mu_k$ for its $k$th largest singular value and refer to the leftmost $k$ columns of $\bU$ (resp.\ $\bV$) as the top $k$ left (resp.\ right) singular vectors of $\bA$.  The Moore--Penrose pseudoinverse of $\bA$ is defined as $\bA^\dag := \bV\diag(\mu^{-1}_1, \ldots, \mu^{-1}_r)\bU^\top$.  If $S \subseteq [n]$, we write $\bA_S \in \mathbb{R}^{|S| \times d}$ for the matrix obtained by extracting the rows of $\bA$ that are in $S$.

For two matrices $\bA,\bB \in \SS^{d\times d}$, we write $\bA\preceq \bB$ if $\bB-\bA$ is positive semidefinite.  The Kronecker product of two matrices $\bA = (A_{ij})\in\mathbb{R}^{d_1\times d_2}$ and $\bB = (B_{ij}) \in\mathbb{R}^{d_1'\times d_2'}$ is defined as the block matrix
\[
\bA\otimes \bB := \begin{pmatrix} A_{11}\bB & \cdots & A_{1d_2}\bB\\ \vdots & \ddots & \vdots\\ A_{d_11}\bB & \cdots & A_{d_1d_2}\bB\end{pmatrix}\in\mathbb{R}^{d_1d_1'\times d_2d_2'}.
\]
When $d_1' = d_1$ and $d_2' = d_2$, the Hadamard product of $\bA$ and $\bB$, denoted $\bA \circ \bB$, is defined such that $(\bA \circ \bB)_{ij} = A_{ij} B_{ij}$ for any $i \in [d_1]$ and $j \in [d_2]$. Moreover, we say that $\bB$ is the \emph{Hadamard inverse} of $\bA$ if $\bA\circ \bB = \mathbf{1}_{d_1}\mathbf{1}_{d_2}^\top$. 
	
If $\bU, \bV\in\mathbb{O}^{d\times K}$, then the matrix of principal angles between $\col(\bU)$ and $\col(\bV)$ is given by $\Theta(\bU,\bV):= \diag\bigl(\cos^{-1}(\sigma_1),\ldots,\allowbreak \cos^{-1}(\sigma_K)\bigr)$, where $\sigma_j = \sigma_j(\bU^\top\bV)$; we let $\sin\Theta(\bU, \bV)$ be defined entrywise.  We define the loss function
\[
  L(\bU,\bV) := \|\sin\Theta(\bU, \bV)\|_{\mathrm{F}}.
\]
For a real-valued random variable $X$ and $r\in\mathbb{N}$, we define its (Orlicz) $\psi_r$-norm by 
 \[
 	\|X\|_{\psi_r} := \sup_{q \in \mathbb{N}}\, q^{-1/r} \bigl(\E|X|^q\bigr)^{1/q}. 
 \]
 For a random vector $\mathbf{x}$ taking values in $\mathbb{R}^d$ and $r \geq 1$, we define its (Orlicz) $\psi_r$-norm by 
\[
\|\mathbf{x}\|_{\psi_{r}} := \sup_{\bu \in \mathcal{S}^{d-1}} \| \bu^\top \bx \|_{\psi_r},
\]
and define a version that is invariant to invertible affine transformations by
\[
  	\|\mathbf{x}\|_{\psi^*_{r}} := \sup_{\bu\in \cS^{d -1}} \frac{\|\bu^{\top}(\bx-\mathbb{E}\bx)\|_{\psi_r}}{\var^{1/2}(\bu^\top \bx)}.
\]
We say that a $d$-dimensional random vector $\bx$ is \emph{sub-Gaussian} if $\|\bx\|_{\psi^*_2} < \infty$. For two distributions $P_1$ and $P_2$ defined on the same measurable space $(\mathcal{X},\mathcal{A})$ and such that $P_1$ is absolutely continuous with respect to $P_2$, the Kullback--Leibler divergence from $P_2$ to $P_1$ is given by
\[	
 	\KL(P_1, P_2) := \int_{\mathcal{X}} \log\frac{d P_1}{d P_2} \, dP_1. 
      \]
      Finally, for $a,b \geq 0$, we write $a \lesssim b$ if there exists a universal constant $C > 0$ such that $a \leq Cb$, and, where $a$ and $b$ may depend on an additional variable $x$, say, we write $a \lesssim_x b$ if there exists $C > 0$, depending only on $x$, such that $a \leq Cb$.

\section{The inverse-probability weighted estimator}
\label{Sec:NoIncoherence}
	 
For notational simplicity, we will write $\lambda_j := \lambda_j(\bSigma_{\bu})$ throughout the paper.  Let $\cA_{ij}$ denote the event that the $(i,j)$th entry $Y_{ij}$ of $\bY$ is observed. We define the revelation matrix $\bOmega = (\omega_{ij}) \in \RR^{n \times d}$ by $\omega_{ij} := \mathds{1}_{\cA_{ij}}$ and the partially observed data matrix 
\begin{equation}
\label{Eq:ModelC}
\bY_{\bOmega} := \bY \circ \bOmega. 
\end{equation}
In this section, we consider the simple case, where entries of the data matrix $\bY$ are observed independently and completely at random (i.e., independent of $(\bU,\bZ)$) with $p$-homogeneous missingness probability.  Thus, $\PP(\cA_{ij}) = p \in (0,1)$ for all $i \in [n], j \in [d]$, and $\cA_{ij}$ and $\cA_{i' j'}$ are independent for $(i,j) \neq (i', j')$.

For $i \in [n]$, let $\by_i^\top$ and $\bomega_i^\top$ denote the $i$th rows of $\bY$ and $\bOmega$ respectively, and define $\widetilde \by_i:=\by_i\circ \bomega_i$. Writing $\bP := \mathbb{E} \bomega_1\bomega_1^\top$ and $\bW$ for its Hadamard inverse, we have that under the $p$-homogeneous missingness mechanism, $\bP = p^2  \bigl\{\mathbf{1}_d\mathbf{1}_d^\top - (1 - p^{-1})\bI_d\bigr\}$ and $\bW = p^{-2}  \bigl\{\mathbf{1}_d\mathbf{1}_d^\top - (1 - p)\bI_d\bigr\}$. Following \citet{CKR15}, we consider the following weighted sample covariance matrix: 
\[
\bG:= \biggl(\frac{1}{n}\bY_{\bOmega}^\top \bY_{\bOmega} \biggr)\circ \bW = \biggl(\frac{1}{n}\sum_{i=1}^n \widetilde \by_i \widetilde \by_i^\top\biggr)\circ \bW. 
\]
The reason for including the weight $\bW$ is to ensure that $\E(\bG|\bY) = n^{-1}\bY^\top\bY$, so that $\bG$ is an unbiased estimator of $\bSigma_{\by}$. Related ideas appear in the work of \citet{CZh18} on high-dimensional covariance matrix estimation with missing data.  There, the authors propose a `generalised sample covariance matrix', where the covariance between any two dimensions $j$ and $k$ is estimated using only the observations for which both dimensions~$j$ and~$k$ were observed.  In practice, $p$ is typically unknown and needs to be estimated. It is thus natural to consider the following plug-in estimator $\widehat \bG$: 
\begin{equation}
  \label{Eq:wgm2}
	\widehat \bG = \biggl(\frac{1}{n}\bY_{\bOmega}^\top \bY_{\bOmega} \biggr)\circ \widehat \bW, 
\end{equation}
where $\widehat \bW = \widehat p^{-2}  \bigl\{\mathbf{1}_d\mathbf{1}_d^\top - (1 - \widehat p)\bI_d\bigr\}$ and $\widehat p := (nd)^{-1} \lonenorm{\bOmega}$ denotes the proportion of observed entries in $\bY$.  We let $\widehat\bV_K$ denote the top $K$ eigenvectors of $\widehat \bG$. 

\subsection{Theory for homogeneous missingness}
\label{Eq:HomogeneousTheory}

In order to describe our theoretical performance guarantee for $\widehat\bV_K$, we first list our conditions on the underlying data generating mechanism.  We assume that 
$(\bY_{\bOmega},\bOmega)$ is generated according to~\eqref{Eq:ModelA}, \eqref{Eq:ModelB} and~\eqref{Eq:ModelC}, where:
\begin{enumerate}[label={(A\arabic*)},noitemsep]
\item $\bU$, $\bZ$ and $\bOmega$ are independent;
\item $\bU$ has independent and identically distributed rows $(\bu_i:i\in[n])$ with $\mathbb{E}\bu_1 = 0$ and $\|\bu_1\|_{\psi_2^*}\leq  \tau$;
\item $\bZ = (z_{ij})_{i\in[n],j\in[d]}$ has independent and identically distributed entries with $\mathbb{E}z_{11} = 0$, $\var z_{11} = 1$ and $\|z_{11}\|_{\psi^*_2}\leq  \tau$;
\item $\|y^2_{1j}\|_{\psi_1} \leq  M$ for all $j \in [d]$;
\item $\bOmega$ has independent $\text{Bern}(p)$ entries.
\end{enumerate}
In many places in this work, we will think intuitively of $\tau$ and $M$ as constants.  In particular, if $\bU$ has multivariate normal rows and $\bZ$ has normal entries, then we can simply take $\tau = 1$.  For $M$, under the same normality assumptions, we have $\|y^2_{1j}\|_{\psi_1} = \var (y_{1j})$, so this intuition amounts to thinking of the variance of each component of our data as being of constant order.

The theorem below gives bounds on the expected proximity between $\col(\widehat \bV_K)$ and $\col(\bV_K)$.
\begin{thm}
\label{thm:1}
Assume (A1)--(A5) and that $n, d\geq 2$, $dp \geq 1$.  Write $R := \lambda_1 + 1$.  Then there exists a universal constant $C>0$ such that 
\begin{equation}
	\label{eq:thm1}
		\E L(\widehat\bV_K, \bV_K) \leq  \frac{CK^{1/2}}{\lambda_K p}\biggl\{\biggl(\frac{Md(R\tau^2 p + M \log d)\log^2 d}{n}\biggr)^{1 / 2} + \frac{Md\log^2 d \log n}{n}\biggr\}.
\end{equation}
In particular, if $n\geq d\log^2d \log^2n / (\lambda_1 p + \log d)$, then there exists $C_{M, \tau} > 0$, depending only on $M$ and $\tau$, such that  
\begin{equation}
	\label{eq:thm1_minimax}			
	\E L(\widehat\bV_K, \bV_K) \leq  \frac{C_{M, \tau}}{\lambda_K p}\biggl(\frac{Kd(\lambda_1 p+ \log d)\log^2 d}{n} \biggr)^{1 / 2}. 
\end{equation}
\end{thm}

When $M, \tau$ are regarded as constants, Theorem \ref{thm:2} below shows that \eqref{eq:thm1_minimax} is the minimax rate up to logarithmic factors when $K = 1$. Note that the condition $n\geq d\log^2 d\log^2 n/(\lambda_1 p+ \log d)$ is reasonable given the scaling requirement for consistency of the empirical eigenvectors  \citep{SSZ16, WFa17, JLu12}. Indeed, Theorem~5.1 in \cite{SSZ16} shows that when $\lambda_1 \gg 1$, the top eigenvector of the sample covariance matrix estimator is consistent if and only if $d / (n\lambda_1) \rightarrow 0$. If we regard $np$ as the effective sample size in our missing data PCA problem, then it is a sensible analogy to assume $d / (np\lambda_1) \rightarrow 0$ here, which implies that the condition $n\geq d\log^2d \log^2n / (\lambda_1 p + \log d)$ holds for large $n$, up to poly-logarithmic factors. 

As mentioned in the introduction, \citet{CKR15} considered the different but related problem of singular space estimation in a model in which $\bY = \bTheta + \bZ$, where $\bTheta$ is a matrix of the form $\bU \bV_K^\top$ for a \emph{deterministic} matrix $\bU$, whose rows are not necessarily centred.  In this setting $\bV_K$ is the leading $K$-dimensional right singular space of $\bTheta$, and the same estimator $\widehat\bV_K$ can be applied.  An important distinction is that, when the rows of $\bU$ are not centred and the entries of $\bTheta$ are of comparable magnitude, $\|\bTheta\|_{\mathrm{F}}$ is of order $\sqrt{nd}$, so when $K$ is regarded as a constant, it is natural to think of the singular values of $\bTheta$ as also being of order $\sqrt{nd}$.  Indeed, this is assumed in \citet{CKR15}.  On the other hand, in our model, where the rows of $\bU$ have mean zero, assuming that the eigenvalues are of order $\sqrt{nd}$ would amount to an extremely strong requirement, essentially restricting attention to very highly spiked covariance matrices.  Removing this condition requires completely different arguments.  
Moreover, \eqref{eq:thm1_minimax} reveals an interesting phase transition phenomenon that has not been observed previously in the literature. Specifically, if the signal strength is large enough that $\lambda_1  \geq p^{-1}\log d$, then we should regard $np$ as the effective sample size, as might intuitively be expected. On the other hand, if $\lambda_1 < p^{-1}\log d$, then the estimation problem is considerably more difficult and the effective sample size is of order $np^2$.  In fact, by inspecting the proof of Theorem~\ref{thm:1}, we see that in the high signal case, it is the difficulty of estimating the diagonal entries of $\bSigma_\by$ that drives the rate, while when the signal strength is low, the bottleneck is the challenge of estimating the off-diagonal entries.  By comparing \eqref{eq:thm1_minimax} with the minimax lower bound result in Theorem~\ref{thm:2} below, we see that this phase transition phenomenon is an inherent feature of this estimation problem rather than an artefact of the proof techniques we used to derived the upper bound.  

In order to state our minimax lower bound, we let $\mathcal{P}_{n,d}(\lambda_1,p)$ denote the class of distributions of pairs $(\bY_\bOmega,\bOmega)$ satisfying (A1), (A2), (A3) and~(A5) with $K=1$.  Since we are now working with vectors instead of matrices, we write $\bv$ in place of $\bV_1$. 
\begin{thm}
	\label{thm:2}
	There exists a universal constant $c>0$ such that  
	\[
		\inf_{\widehat\bv}\sup_{P\in\mathcal{P}_{n,d}(\lambda_1,p)} \E_P L(\widehat\bv, \bv) \geq c\min\biggl\{\frac{1}{\lambda_1 p}\biggl({\frac{d (\lambda_1 p+1)}{n}}\biggr)^{1 / 2}, 1\biggr\},
	\]
    where the infimum is taken over all estimators $\widehat{\bv} = \widehat{\bv}(\bY_{\bOmega},\bOmega)$ of $\bv$. 
\end{thm}
Theorem~\ref{thm:2} reveals that $\widehat \bV_1$ in Theorem \ref{thm:1} achieves the minimax optimal rate of estimation up to a logarithmic factor when $M$ and $\tau$ are regarded as constants and $K = 1$. 
	
\subsection{General observation mechanism}
\label{sec:2.2}
Although Section~\ref{Eq:HomogeneousTheory} may appear to indicate that the problem of high-dimensional PCA with missing entries is essentially solved, the aim of this subsection is to show that the situation changes dramatically once the data can be missing heterogeneously.  

To this end, consider the following example.
Suppose that $\bomega$ is equal to $(1,0,1,\ldots,1)^\top$ or $(0,1,1,\ldots,1)^\top$ with equal probability, so that 
\[
	\bP = \mathbb{E}\bomega\bomega^\top = \begin{pmatrix}
		1/2 & 0 & 1/2 & \hdots & 1/2 \\
		0 & 1/2 & 1/2 & \hdots & 1/2 \\
		1/2 & 1/2 & 1 & \hdots & 1 \\
		\vdots & \vdots & \vdots & \ddots & \vdots \\
		1/2 & 1/2 & 1& \hdots & 1
	\end{pmatrix} \in \mathbb{R}^{d \times d}. 
\]
In other words, for each $i \in [n]$, we observe precisely one of the first two entries of $\by_i$, together with all of the remaining $(d - 2)$ entries. Let $\bSigma = \bI_d + \balpha\balpha^{\top}$, where $\balpha = (2^{-1/2}, 2^{-1/2}, 0, \ldots, 0)^{\top}\in\mathbb{R}^d$, and $\bSigma' = \bI_d + \balpha'(\balpha')^{\top}$, where $\balpha' = (2^{-1/2}, -2^{-1/2}, 0, \ldots, 0)^\top \in \RR^d$. Suppose that $\by \sim N_d(\bzero, \bSigma)$ and let $\widetilde \by :=\by\circ \bomega$, and similarly assume that $\by' \sim N_d(\bzero, \bSigma')$ and set $\widetilde \by' :=\by'\circ \bomega$.  Then $(\widetilde \by,\bomega)$ and $(\widetilde \by',\bomega)$ are identically distributed. However, the leading eigenvectors of $\bSigma$ and $\bSigma'$ are respectively $\balpha$ and $\balpha'$, which are orthogonal!
Thus, it is impossible to simultaneously estimate consistently the leading eigenvector of both $\bSigma$ and $\bSigma'$ from our observations. We note that it is the disproportionate weight of the first two coordinates in the leading eigenvector, combined with the failure to observe simultaneously the first two entries in the data, that makes the estimation problem intractable in this example.

The understanding derived from this example motivates us to assume that $\bV_K$ satisfies the incoherence condition $\|\bV_K\|_\infty\leq  \mu/\sqrt{d}$ for some $\mu>0$.  The intuition here is that the maximally incoherent case is where each column of $\bV_K$ is a unit vector proportional to a vector whose entries are either $1$ or $-1$, in which case $\|\bV_K\|_\infty = 1/\sqrt{d}$. Our condition asks for the columns of $\bV_K$ to have an incoherence of the same order as this maximally incoherent case. Similar conditions have been invoked in the literature on matrix completion \citep[e.g.,][]{CPl10,KMO10}, but for a different reason. There, the purpose is to ensure that the true right singular space is not too closely aligned with the standard basis, which allows the missing entries of the matrix to be inferred from relatively few observations. In our case, the incoherence condition ensures that significant estimation error in a few components of the leading eigenvectors does not affect the overall statistical performance too much. Thus, it rules out examples such as the one described above, where heavy corruption in only a few entries spoils any chance of consistent estimation.  

\section{Our new algorithm, \texttt{primePCA}}
\label{Sec:WithIncoherence}



We are now in a position to introduce and analyse our iterative algorithm \texttt{primePCA} to estimate principal eigenspaces of the covariance matrix $\bSigma_{\by}$.  The basic idea is to iteratively refine a current (input) iterate $\widehat \bV_{K}^{(\mathrm{in})}$ by first imputing the missing entries of the data matrix~$\bY_{\bOmega}$ using the current estimate of~$\bV_K$, and then applying a singular value decomposition (SVD) to the completed data matrix.  
More precisely, for $i \in [n]$, we let $\cJ_i$ denote the indices for which the corresponding entry of $\by_i$ is observed, and regress the observed data $\widetilde{\by}_{i, \cJ_i} = {\by}_{i, \cJ_i}$ on $(\widehat \bV_K^{(\mathrm{in})})_{\cJ_i}$ to obtain an estimate $\widehat \bu_i$ of the $i$th row of $\bU$.  This is natural in view of the data generating mechanism $\by_i = \bV_K\bu_i + \bz_i$.  We then use $\widehat \by_{i,\mathcal{J}_i^{\mathrm{c}}} := (\widehat \bV_K^{(\mathrm{in})} \widehat \bu_i)_{\cJ_i^{\mathrm{c}}}$ to impute the missing values $\by_{i,\cJ_i^{\mathrm{c}}}$, retain the original observed entries as $\widehat \by_{i,\mathcal{J}_i} := \by_{i,\mathcal{J}_i}$, and set our next (output) iterate $\widehat \bV_{K}^{(\mathrm{out})}$ to be the top $K$ right singular vectors of the imputed matrix $\widehat \bY := (\widehat\by_{1}, \ldots, \widehat\by_{n})^{\top}$.  To motivate this final choice, observe that when $\bZ = \bzero$, we have $\rank(\bY) = K$; we therefore have the SVD $\bY = \bL\bGamma\bR^\top$, where $\bL \in \OO^{n \times K}, \bR \in \OO^{d \times K}$ and $\bGamma \in \RR^{K \times K}$ is diagonal with positive diagonal entries.  This means that $\bR = \bV_K\bU^\top \bL\bGamma^{-1}$, so the column spaces of $\bR$ and $\bV_K$ coincide.  For convenience, pseudocode of a single iteration of refinement in this algorithm is given in Algorithm~\ref{algo:1}.  

\begin{algorithm}[hbtp]
	\renewcommand{\algorithmicrequire}{\textbf{\underline{Input}:}}
	\renewcommand{\algorithmicensure}{\textbf{\underline{Output}:}}
	\caption{\texttt{refine}$(K,\widehat \bV_{K}^{(\mathrm{in})},\bOmega,\bY_{\bOmega})$, a single step of refinement of current iterate $\widehat \bV_{K}^{(\mathrm{in})}$}
	\vspace{.2cm}		
         \algorithmicrequire{ $K \in [d]$, $\widehat \bV_{K}^{(\mathrm{in})} \in \OO^{d \times K}$, $\bOmega \in \{0, 1\}^{n \times d}$ \text{ with } $\min_{i} \|\bomega_i\|_1 \geq 1$, $\bY_{\bOmega} \in \RR^{n \times d}$} \\
         \algorithmicensure{ $\widehat \bV^{(\mathrm{out})}_{K} \in \OO^{d \times K}$}
	\vspace{.2cm}
        \label{algo:1}
        \begin{algorithmic}[1]
        		\For{$i$ in $[n]$}
			\State{$\cJ_i \leftarrow \{j \in [d]: \omega_{ij} = 1\}$}
    			\State{$\widehat \bu_i \leftarrow (\widehat \bV_{K}^{(\mathrm{in})})_{\cJ_i}^{\dagger} \widetilde{\by}_{i, \cJ_i}$} 
    			\State{$\widehat \by_{i, \cJ_i^{\mathrm{c}}} \leftarrow \widehat \bV_{K}^{(\mathrm{in})} \widehat\bu_{i, \cJ_i^{\mathrm{c}}}$} 
    			\State{$\widehat \by_{i, \cJ_i} \leftarrow \by_{i, \cJ_i}$}
		\EndFor
		
		\State{$\widehat \bY \leftarrow (\widehat\by_{1}, \ldots, \widehat\by_{n})^{\top}$}
		\State{$\widehat \bV^{(\mathrm{out})}_{K} \leftarrow$ top $K$ right singular vectors of $\widehat \bY$}

        \end{algorithmic}	
\end{algorithm}
We now seek to provide formal justification for Algorithm~\ref{algo:1}.  For any $\bV^{(1)}, \bV^{(2)} \in \OO^{d \times K}$, we let $\bW_1 \bD_{\bV^{(1)}, \bV^{(2)}} \bW_2^\top$ be an SVD of $(\bV^{(2)})^\top\bV^{(1)}$ and let $\bW_{\bV^{(1)}, \bV^{(2)}} := \bW_1\bW_2^\top$. The two-to-infinity distance between $\bV^{(1)}$ and $\bV^{(2)}$ is then defined to be 
\[
	\cT(\bV^{(1)}, \bV^{(2)}) := \twotoinf{\bV^{(1)} - \bV^{(2)} \bW_{\bV^{(1)}, \bV^{(2)}}}. 
\]
Note that the definition of $\cT(\bV^{(1)}, \bV^{(2)})$ does not depend on our choice of SVD and that $\cT(\bV^{(1)}, \bV^{(2)}) = \cT(\bV^{(1)}\bO_1, \bV^{(2)}\bO_2)$ for any $\bO_1, \bO_2 \in \OO^{K \times K}$.  The following proposition considers the noiseless setting $\bZ = \bzero$, and provides conditions under which, for any estimator $\widehat \bV_K^{(\mathrm{in})}$ that is close to $\bV_K$, a single iteration of refinement in Algorithm~\ref{algo:1} contracts the two-to-infinity distance between their column spaces.  In a slight abuse of notation, we write $\bOmega^{\mathrm{c}} := \mathbf{1}_d\mathbf{1}_d^\top - \bOmega$.
	
\begin{prop}	
	\label{thm:4}
	Let $\widehat\bV_K^{(\mathrm{out})} := \emph{\texttt{refine}}\bigl(K,\widehat \bV_{K}^{(\mathrm{in})},\bOmega,\bY_{\bOmega}\bigr)$ as in Algorithm~\ref{algo:1} and let $\Delta := \cT(\widehat \bV_K^{(\mathrm{in})}, \bV_K)$.
        We assume that $\min_{i \in [n]} \|\bomega_i\|_1 > K$, that $\min_{i \in [n]} \frac{d^{1 / 2}\sigma_{K}( (\widehat\bV_{K}^{(\mathrm{in})})_{\cJ_i})}{|\cJ_i|^{1/2}} \geq 1/\sigma_* > 0$, and write the SVD of $\bY$ as $\bL\bGamma\bR^\top$.  Suppose that $\bZ = \bzero$, and that both $\twotoinf{\bL} \leq \mu_1 ({K / n})^{1 / 2}$ and $\twotoinf{\bR} \leq \mu_2 ({K / d})^{1 / 2}$ hold
        for some $\mu_1,\mu_2 \geq 1$.  Then there exist $c_1,C > 0$, depending only on $\mu_1,\mu_2$ and $\sigma_*$, such that whenever 
\begin{enumerate}[label={(\roman*)},noitemsep]
		\item $\Delta \leq \frac{c_1\sigma_K(\bGamma)}{K ^ 2\sigma_1(\bGamma) \sqrt{d}}$,
                \item $\rho := \frac{CK^{2} \sigma_1(\bGamma) \|{\bOmega^{\mathrm{c}}}\|_{1 \to 1}}{\sigma_K(\bGamma) n} < 1$,
                \end{enumerate}
                we have that
                
	\[
		\cT(\widehat\bV_K^{(\mathrm{out})}, \bV_K) \leq \rho \Delta. 
	\]
      \end{prop}
      In practice, in cases where either of the two conditions on $\min_{i \in [n]} \|\bomega_i\|_1$ or $\sigma_*$ is not satisfied, we first perform a screening step that restricts attention to a set of row indices for which the data contain sufficient information to estimate the $K$ principal components.  This screening step is explicitly accounted for in Algorithm~\ref{algo:2} below, as well as in the theory that justifies it.

\begin{algorithm}[hbtp]
\caption{\texttt{primePCA}, an iterative algorithm for estimating $\bV_K$ given initialiser $\widehat \bV^{(0)}_{K}$}
\renewcommand{\algorithmicrequire}{\textbf{\underline{Input}:}}
\renewcommand{\algorithmicensure}{\textbf{\underline{Output}:}}
\vspace{.2cm}		
         \algorithmicrequire{$\; K \in [d], \widehat{\bV}^{(0)}_{K} \in \mathbb{O}^{d \times K}, \bOmega \in \{0, 1\}^{n \times d}, \bY_{\bOmega} \in \RR^{n \times d}$, $n_{\textnormal{iter}} \in \mathbb{N}$, $\sigma_* \in (0,\infty)$, $\kappa^* \in [0, \infty)$}\\
         \algorithmicensure{$\; \widehat \bV_{K} \in \RR^{d \times K}$}
\vspace{.2cm}
\label{algo:2}
\begin{algorithmic}[1]
  \For{$i$ in $[n]$}
          \State{$\cJ_i \leftarrow \{j \in [d]: \omega_{ij} = 1\}$}
          \EndFor
                        \For{$t$ in $[n_{\textnormal{iter}}]$}
                        \State $\cI^{(t-1)} \leftarrow \bigl\{i: \lonenorm{\bomega_i} > K, \sigma_{K}( (\widehat\bV_{K}^{(t-1)})_{\cJ_i}) \geq \frac{|\cJ_i|^{1/2}}{d^{1 / 2}\sigma_*}\bigr\}$
		\State $\widehat \bV^{(t)}_{K} \leftarrow \texttt{refine}(K, \widehat\bV^{(t - 1)}_{K},\bOmega_{\mathcal{I}^{(t-1)}}, (\bY_{\bOmega})_{\mathcal{I}^{(t-1)}})$ \# \texttt{refine} is defined in Algorithm \ref{algo:1}. 
		\If{$L(\widehat \bV^{(t)}_{K}, \widehat\bV^{(t - 1)}_{K}) < \kappa^*$} \textbf{break}
                \EndIf
	\EndFor	\\
	\Return $\widehat\bV_{K} = \widehat\bV^{(t)}_{K}$					
\end{algorithmic}
\end{algorithm}

Algorithm~\ref{algo:2} provides pseudocode for the iterative \texttt{primePCA} algorithm, given an initial estimator $\widehat{\bV}_K^{(0)}$.  The iterations continue until either we hit the convergence threshold $\kappa^*$ or the maximum iteration number $n_{\text{iter}}$.  Theorem~\ref{thm:primepca} below guarantees that, in the noiseless setting of Proposition~\ref{thm:4}, the \texttt{primePCA} estimator converges to $\bV_K$ at a geometric rate.
\begin{thm}
	\label{thm:primepca}
For $t \in [n_{\mathrm{iter}}]$, let $\widehat\bV_K^{(t)}$ be the $t^{\mathrm{th}}$ iterate of Algorithm~\ref{algo:2} with input $K$, $\widehat \bV^{(0)}_{K}$, $\bOmega \in \{0, 1\}^{n \times d}$, $\bY_{\bOmega} \in \RR^{n \times d}$, $n_{\textnormal{iter}} \in \mathbb{N}$, $\sigma_* \in (0,\infty)$ and $\kappa^*= 0$.	Write $\Delta := \cT(\widehat\bV^{(0)}_K, \bV_K)$ and let
	\[
        		\cI := \biggl\{i: \lonenorm{\bomega_i} > K, \sigma_{K}( (\bV_{K})_{\cJ_i}) \ge \frac{|\cJ_i|^{1/2}}{d^{1 / 2}\sigma_*}\biggr\},
	\]
where $\mathcal{J}_i := \{j:\omega_{ij}=1\}$.  Let $\bY_\cI = \bL\bGamma\bR^\top$ be an SVD of $\bY_{\cI}$. Suppose that both $\twotoinf{\bL} \leq \mu_1 ({K / |\cI|})^{1 / 2}$ and $\twotoinf{\bR} \leq \mu_2 ({K / d})^{1 / 2}$ hold, 
        for some $\mu_1,\mu_2 \geq 1$.  Let
        \[
          \mathcal{Z} := \biggl\{\frac{\sigma_{K}\bigl( (\bV_{K})_{\cJ_i}\bigr)d^{1/2}}{|\cJ_i|^{1/2}} : i \in [n], \lonenorm{\bomega_i} > K\biggr\},
        \]
and assume that $\epsilon := \min_{z \in \cZ} |z - \sigma_*^{-1}| > 0$.  
	Then there exist $c_1,C > 0$, depending only on $\mu_1,\mu_2$, $\sigma_*$ and $\epsilon$, such that whenever
\begin{enumerate}[label={(\roman*)},noitemsep]
		\item $\Delta \leq \frac{c_1\sigma_K(\bY_{\cI})}{K ^ 2\sigma_1(\bY_{\cI})\sqrt{d}}$,
                \item $\rho := \frac{CK^{2} \sigma_1(\bY_{\cI}) \|{\bOmega_{\cI}^{\mathrm{c}}}\|_{1 \to 1}}{\sigma_K(\bY_{\cI}) |\cI|} < 1$,
                \end{enumerate}
	we have that for every $t \in [n_{\mathrm{iter}}]$, 
	\[
		\cT(\widehat\bV_K^{(t)}, \bV_K) \leq \rho^{t}\Delta. 
    \]
\end{thm}

\subsection{Initialisation}
\label{Sec:Initialisation}
Theorem~\ref{thm:primepca} provides a general guarantee on the performance of \texttt{primePCA}, but relies on finding an initial estimator $\widehat\bV^{(0)}_K$ that is sufficiently close to the truth $\bV_K$.  The aim of this subsection, then, is to propose a simple initialiser and show that it satisfies the requirement of Theorem~\ref{thm:primepca} with high probability, conditional on the missingness pattern.

Consider the following modified weighted sample covariance matrix
\[
\widetilde \bG := \frac{1}{n}\sum_{i = 1}^n \widetilde \by_i \widetilde \by_i^\top \circ \widetilde\bW, 
\]
where for any $j, k \in [d]$, 
\[
\widetilde\bW_{jk} := 
\begin{cases}
\frac{n}{\sum_{i = 1}^n \omega_{ij}\omega_{ik}}  & \text{if $\sum_{i = 1}^n \omega_{ij}\omega_{ik} > 0$,}\\
0, & \text{otherwise}. 
\end{cases}
\]
Here, the matrix $\widetilde\bW$ replaces $\widehat\bW$ in~\eqref{Eq:wgm2} because we no longer wish to assume homogeneous missingness.  
We take as our initial estimator of $\bV_K$ the matrix of top $K$ eigenvectors of $\widetilde \bG$, denoted~$\widetilde \bV_K$.  
Theorem~\ref{thm:init_2toinf} below studies the performance of this initialiser, in terms of its two-to-infinity norm error, as required for application in Theorem~\ref{thm:primepca}.  We write $\mathbb{P}^{\bOmega}$ and $\mathbb{E}^{\bOmega}$ for probabilities and expectations conditional on $\bOmega$. 

\begin{thm}
\label{thm:init_2toinf}
Assume (A1)--(A4) and that $n, d\geq 2$.  Suppose further that $\|\bV_K\|_{\infty} \leq \mu/ \sqrt{d}$, that $\sum_{i = 1}^n \omega_{ij}\omega_{ik} > 0$ for all $j,k$ and let $R := \lambda_1 + 1$.  Then there exist $c_{M,\tau,\mu}, C_{M, \tau, \mu} > 0$, depending only on $M,\tau$ and $\mu$, such that for every $\xi > 2$, if
\begin{equation}
\label{eq:lambda_k}
\lambda_K >  c_{M,\tau,\mu}\biggl\{ \biggl(\frac{\max\bigl(\lonenorm{\widetilde \bW}, R \|\widetilde \bW\|_{1 \to 1}\bigr) \xi \log d}{n}\biggr)^{1 / 2} + \frac{\xi \fnorm{\widetilde \bW} \log^2 d}{n}\biggr\},
\end{equation}
then 
\begin{align*}
\mathbb{P}^{\bOmega}\biggl\{ \cT(\widetilde \bV_K, \bV_K) \!\ge \!\frac{C_{M, \tau, \mu} K ^ {3 / 2} R^{1 / 2}}{\lambda_K d^{1 / 2}}&\biggl(1 \!+\! \frac{d^{1 / 2}}{K\lambda_K}\biggr) \biggl( \frac{ \xi^{1 / 2} \inftoinf{\widetilde \bW}^{1 / 2} \log^{\!1 / 2}\!d }{n^{1 / 2}}\! + \!\frac{\xi \twotoinf{\widetilde \bW} \log d}{n}\biggr) \biggr\}  \\
&\le 2(e^{K\log 5}+K+4)d^{- (\xi - 1)} + 2d^{- (\xi - 2)}. 
\end{align*}
As a consequence, writing
\[
\cA := \biggl\{\frac{\sigma_K(\bY_{\cI})}{\sigma_1(\bY_{\cI})} > \frac{C_{M, \tau, \mu} K ^ {7 / 2} R^{1 / 2}}{c_1\lambda_K}\biggl(1\! + \!\frac{d^{1 / 2}}{K\lambda_K}\biggr)  \biggl( \frac{ \xi^{1 / 2} \inftoinf{\widetilde \bW}^{1 / 2} \log^{\!1 / 2}\!d }{n^{1 / 2}}\! + \!\frac{\xi \twotoinf{\widetilde \bW} \log d}{n}\biggr)\biggr\}, 
\]
where $c_1$ is as in Theorem~\ref{thm:primepca}, we have that
\[
\begin{aligned}
\PP^{\bOmega}\biggl( \cT (\widetilde \bV_K, \bV_K)  > \frac{c_1\sigma_K(\bY_{\cI})}{K^2 \sigma_1(\bY_{\cI})d^{1/2}}\biggr) \le 2(e^{K\log 5}+K+4)d^{- (\xi - 1)} + 2d^{- (\xi - 2)} + \mathbb{P}^{\bOmega}(\mathcal{A}^{\mathrm{c}}).
\end{aligned}
\]
      
\end{thm}
The first part of Theorem~\ref{thm:init_2toinf} provides a general probabilistic bound for $\cT(\widetilde \bV_K, \bV_K)$, after conditioning on the missingness pattern. 
This allows us, in the second part, to provide a guarantee on the probability with which $\widetilde \bV_K$ is a good enough initialiser for Theorem~\ref{thm:primepca} to apply.  
For intuition regarding $\PP^{\bOmega}(\cA^{\mathrm{c}})$, consider the $p$-homogenous missingness setting.  In that case, by Lemma~\ref{Lemma:LotsofNorms}, typical realisations of $\widetilde{\bW}$ have $\inftoinf{\widetilde \bW} = O(d/p^2)$ and $\twotoinf{\widetilde \bW} = O(d^{1/2}/p^2)$ when $np^2 \gg \log d$, so in the spiked model where $\lambda_1$ and $\lambda_K$ are both of order $d$, we expect $\PP^{\bOmega}(\cA^{\mathrm{c}})$ to be small.

One of the attractions of our analysis is the fact that we are able to provide bounds that only depend on entrywise missingness probabilities in an average sense, as opposed to worst-case missingness probabilities.  The refinements conferred by such bounds are particularly important when the missingness mechanism is heterogeneous, as typically encountered in practice.  The averaging of missingness probabilities can be partially seen in Theorem~\ref{thm:init_2toinf}, since $\inftoinf{\widetilde \bW}$ and $\twotoinf{\widetilde \bW}$ depend only on the $\ell_1$ and $\ell_2$ norms of each row of $\widetilde \bW$, but is even more evident in the proposition below, which gives a probabilistic bound on the original $\sin \Theta$ distance between $\widetilde \bV_K$ and $\bV_K$.

\begin{prop}
  \label{prop:init_f}
Assume the same conditions as in Theorem~\ref{thm:init_2toinf}.  Then there exists a universal constant $C > 0$ such that for any $\xi > 1$, if 
\begin{equation}
	\label{eq:lambda_k_f}
	\lambda_K > C\biggl\{ \biggl(\frac{M\tau^2 R\|\widetilde \bW\|_{1\rightarrow 1}\xi\log d}{n}\biggr)^{1 / 2} + \frac{M\opnorm{\widetilde \bW}  \xi\log^2 d}{n}\biggr\}, 
\end{equation}
then
\begin{align*}
	\PP^{\bOmega}\biggl\{ L(\widetilde \bV_K, \bV_K) \ge \frac{2^{9/2}e\tau\mu}{\lambda_K}\biggl(\frac{KMR}{d}\biggr)^{\hspace{-.15cm}1 / 2}\biggl( \frac{\xi^{1/2} \|\widetilde \bW\|_1^{1/2}\log^{1/2} d  }{n^{1/2}} &+  \frac{\xi\fnorm{\widetilde\bW}\log d }{n}\biggr)\biggr\} \\
    & \qquad \leq (2K + 4)d^{-(\xi - 1)}. 
\end{align*}
\end{prop}
In this bound, then, we see that $L(\widetilde \bV_K, \bV_K)$ only depends on $\widetilde\bW$ through the entrywise $\ell_1$ and $\ell_2$ norms of the whole matrix. Lemma~\ref{Lemma:LotsofNorms} provides probabilistic control of these norms under the $p$-homogeneous missingness mechanism. In general, if the rows of $\bOmega$ are independent and identically distributed, but different covariates are missing with different probabilities, then entries of $\widetilde{\bW}$ will concentrate around the reciprocals of the simultaneous observation probabilities of pairs of covariates. As such, for a typical realisation of $\bOmega$, our bound in Proposition~\ref{prop:init_f} depends only on the harmonic averages of these simultaneous observation probabilities and their squares. Such an averaging effect ensures that our method is effective in a much wider range of heterogeneous settings than previously allowed in the literature.



\section{Simulation studies}
\label{Sec:Simulations}
In this section, we assess the empirical performance of \texttt{primePCA} as proposed in Algorithm~\ref{algo:2}, with initialiser $\widetilde \bV_K$ from Section~\ref{Sec:Initialisation}, and denote the output of this algorithm by $\widehat\bV_K^{\text{prime}}$.  We generate observations according to the model described in~\eqref{Eq:ModelA},~\eqref{Eq:ModelB} and~\eqref{Eq:ModelC} where the rows of the matrix $\bU$ are independent $N_d(\bzero,\bSigma_{\bu})$ random vectors, for some $\bSigma_{\bu} \succeq \bzero$.  We further generate the observation indicator matrix $\bOmega$, independently of $\bU$ and $\bZ$, and investigate the following four missingness mechanisms that represent different levels of heterogeneity:
\begin{enumerate}[noitemsep, label={(H\arabic*)}]
	\item Homogeneous: $\PP(\omega_{ij} = 1) = 0.05$ for all $i \in [n], j \in [d]$; 
	\item Mildly heterogeneous: $\PP(\omega_{ij} = 1) = P_iQ_j$ for $i \in [n], j \in [d]$, where $P_1,\ldots,P_n \stackrel{\mathrm{iid}}{\sim} U[0, 0.2]$ and $Q_1,\ldots,Q_d \stackrel{\mathrm{iid}}{\sim} U[0.05, 0.95]$ independently; 
	\item Highly heterogeneous columns: $\PP(\omega_{ij} = 1) = 0.19$ for $i \in [n]$ and all odd $j \in [d]$ and $\PP(\omega_{ij} = 1) = 0.01$ for $i \in [n]$ and all even $j \in [d]$. 
	\item Highly heterogeneous rows: $\PP(\omega_{ij} = 1) = 0.18$ for $j \in [d]$ and all odd $i \in [n]$ and $\PP(\omega_{ij} = 1) = 0.02$ for $j \in [d]$ and all even $i \in [n]$. 
\end{enumerate}
In Sections~\ref{Sec:Noiseless},~\ref{Sec:Noisy} and~\ref{Sec:Sensitivity} below, we investigate \texttt{primePCA} in noiseless, noisy and misspecified settings respectively.  In all cases, the average statistical error was estimated from 100 Monte Carlo repetitions of the experiment.  For comparison, we also studied the \texttt{softImpute} algorithm \citep{MHT10, HML15}, which is considered to be state-of-the-art for matrix completion \citep{CLC18}. This algorithm imputes the missing entries of $\bY$ by solving the following nuclear-norm-regularised optimisation problem: 
\[	
	 	\widehat \bY^{\soft} :=  \argmin_{\bX \in \RR^{n \times d}} \biggl\{\frac{1}{2}\fnorm{\bY_{\bOmega} - \bX_{\bOmega}}^2 + \lambda \nnorm{\bX}\biggr\}, 
\]
where $\lambda > 0$ is to be chosen by the practitioner. The \texttt{softImpute} estimator of $\bV_K$ is then given by the top $K$ right singular space $\widehat\bV_K^{\soft}$ of $\widehat\bY^{\soft}$. 

Figure~\ref{fig:noiseless_tau} presents Monte Carlo estimates of $\mathbb{E}L(\widehat\bV_K^{\text{prime}},\bV_K)$ for different choices of $\sigma_*$ in two different settings.  The first uses the noiseless set-up of Section~\ref{Sec:Noiseless}, together with missingness mechanism~(H1); the second uses the noisy setting of Section~\ref{Sec:Noisy} with parameter $\nu=20$ and missingness mechanism~(H2).  We see that the error barely changes when $\sigma_*$ varies within $[2, 10]$; very similar plots were obtained for different data generation and missingness mechanisms, though we omit these for brevity.  For definiteness, we therefore fixed $\sigma_* = 3$ throughout our simulation study. 
\begin{figure}[H]
	\label{fig:noiseless_tau}
	\centering
	\begin{tabular}{cc}
		\includegraphics[clip, trim=0 .5cm 0 1cm, width=6cm]{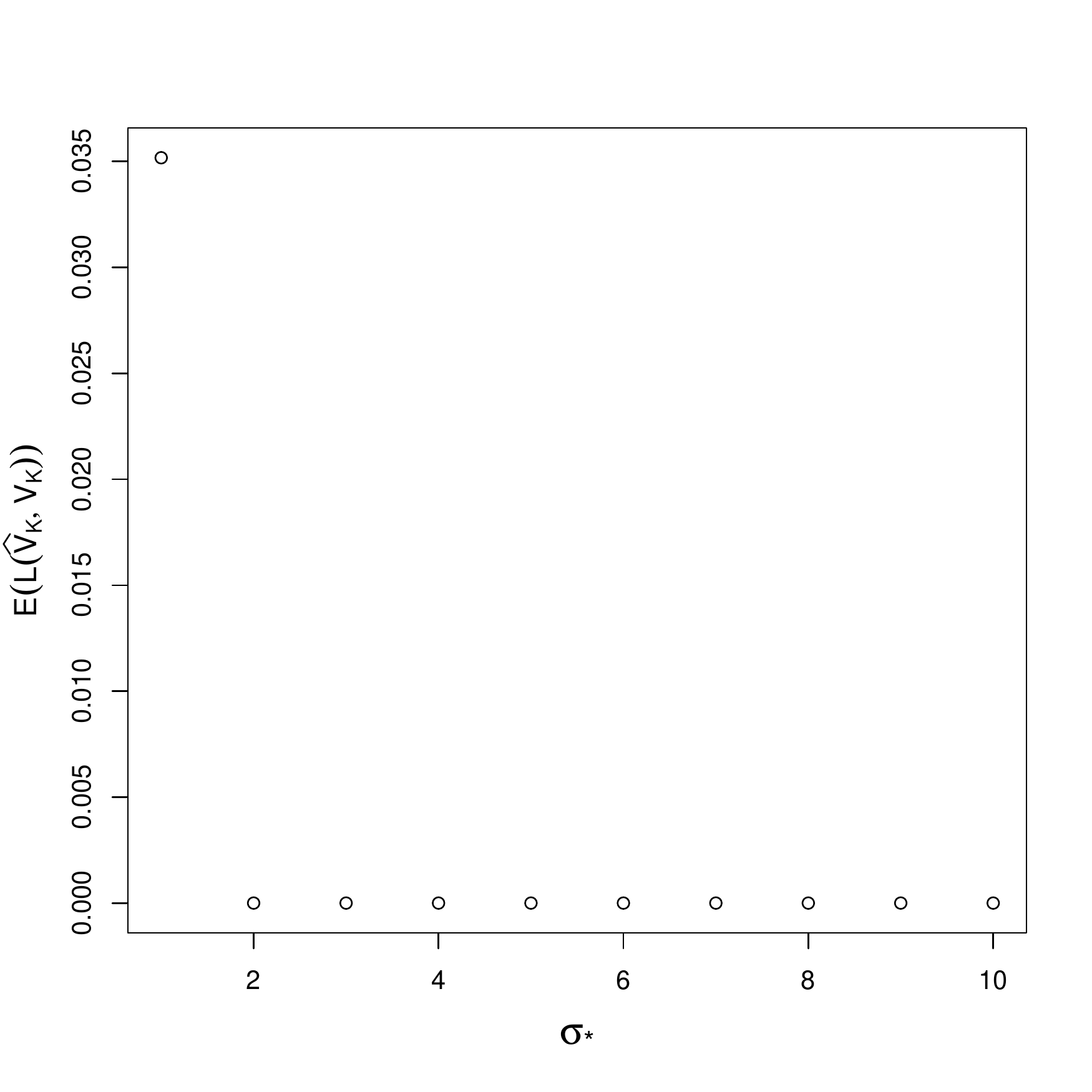} & \includegraphics[clip, trim=0 .5cm 0 1cm, width=6cm]{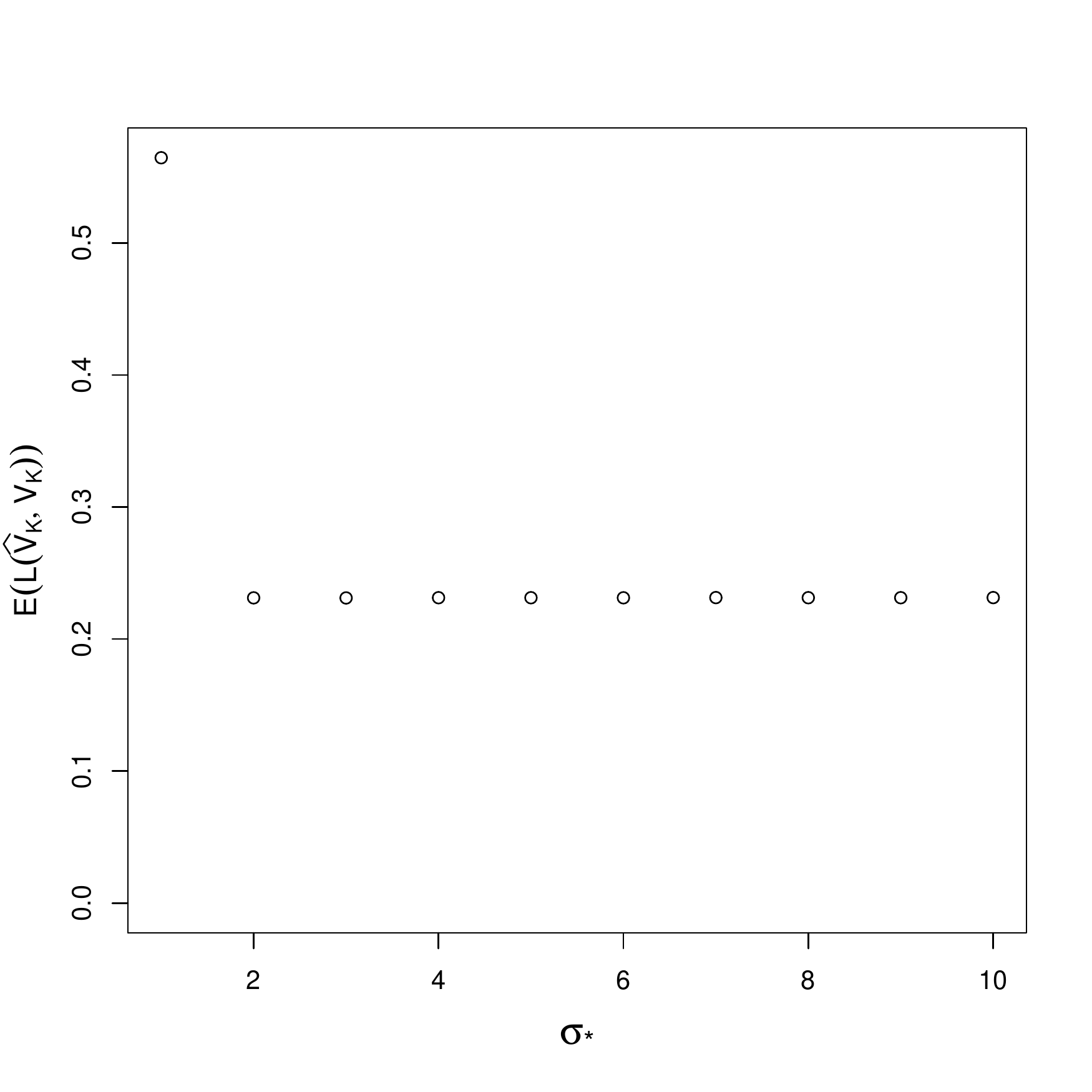} 
	\end{tabular}
	\caption{Estimates of $\mathbb{E}L(\widehat\bV_K^{\text{prime}},\bV_K)$ for various choices of $\sigma_*$ under (H1) in the noiseless setting of Section~\ref{Sec:Noiseless} (left) and (H2) in the noisy setting of Section~\ref{Sec:Noisy} with $\nu=20$ (right).}
      \end{figure}

\subsection{Noiseless case}
\label{Sec:Noiseless}

In the noiseless setting, we set $\bZ = \bzero$, and also fix $n =  2000$, $d = 500$, $K=2$ and $\bSigma_{\bu} = 100\bI_{2}$.  We set
\[
	\bV_K = \sqrt{\frac{1}{500}} \biggl(
	\begin{array}{cc}
		\bone_{250} & \bone_{250} \\
		\bone_{250} & -\bone_{250} 
	\end{array}
	\bigg) \in \mathbb{R}^{500 \times 2}.
\]


In Figure \ref{fig:1}, we present the logarithm of the estimated average loss of \texttt{primePCA} and \texttt{softImpute} under (H1), (H2), (H3) and (H4).  We set the range of $y$-axis to be the same for each method to facilitate straightforward comparison.  We see that the statistical error of \texttt{primePCA} decreases geometrically as the number of iterations increases, which confirms the conclusion of Theorem~\ref{thm:primepca} in this noiseless setting.  Moreover, after a moderate number of iterations, its performance is a substantial improvement on that of the \texttt{softImpute} algorithm, even if this latter algorithm is given access to an oracle choice of the regularisation parameter~$\lambda$.  The high statistical error of \texttt{softImpute} in these settings can be partly explained by the default value of the tuning parameter \texttt{thresh} in the \texttt{softImpute} package in~$\texttt{R}$, namely $10^{-5}$, which corresponds to the red curve in the right-hand panels of Figure~\ref{fig:1}.  By reducing the values of \texttt{thresh} to $10^{-7}$ and $10^{-9}$, corresponding to the green and blue curves in Figure~\ref{fig:1} respectively, we were able to improve the performance of \texttt{softImpute} to some extent, though the statistical error is sensitive to the choice of the regularisation parameter~$\lambda$.  Moreover, even with the optimal choice of $\lambda$, it is not competitive with \texttt{primePCA} (which is also considerably faster to compute, even with $2000$ iterations).

\begin{figure}
	\centering
    	\begin{tabular}{cc}
    	\texttt{primePCA} & \texttt{softImpute} \\
    	\includegraphics[clip, trim=0 .5cm 0 1cm, width=5.0cm]{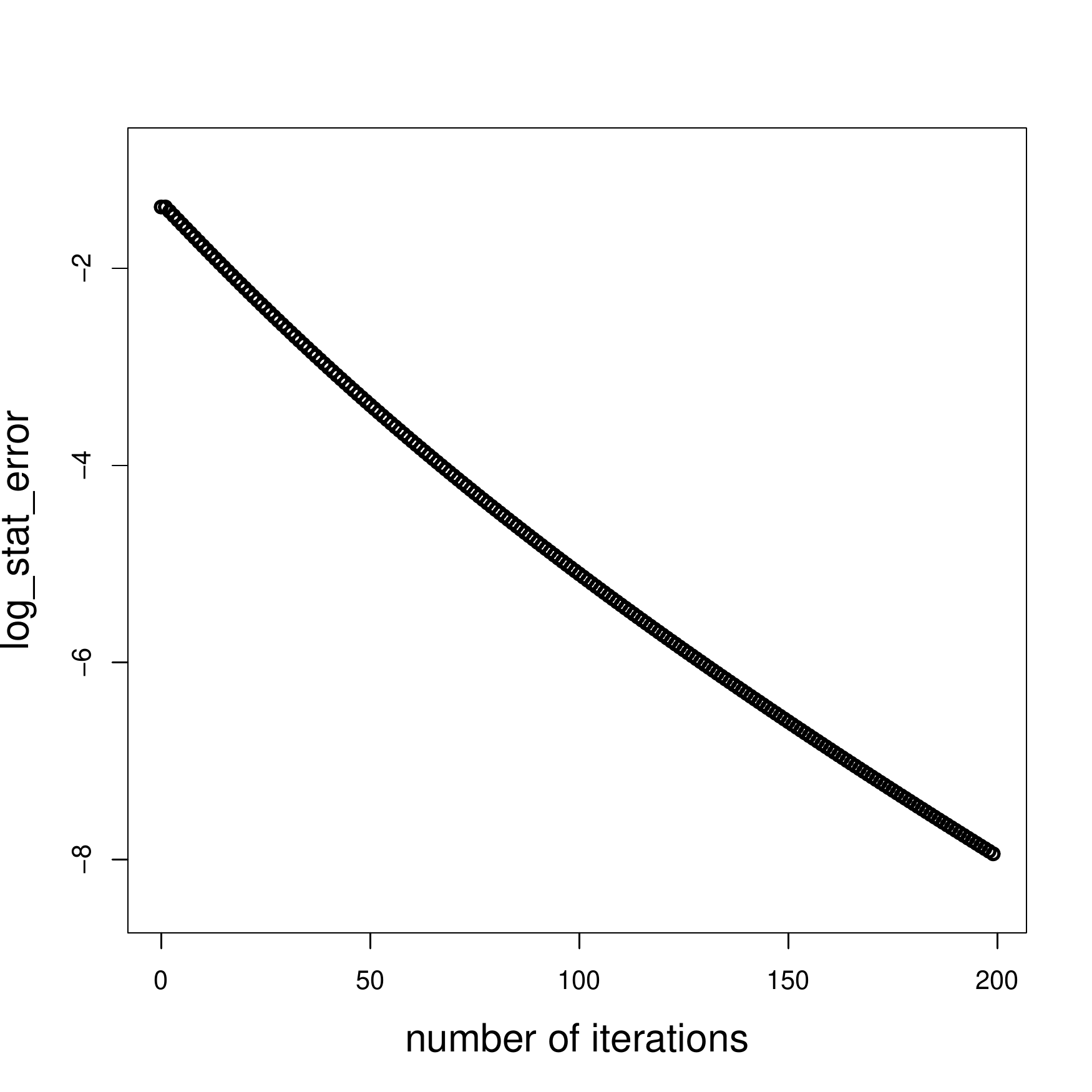} & \includegraphics[clip, trim=0 0.5cm 0 1cm, width=5.0cm]{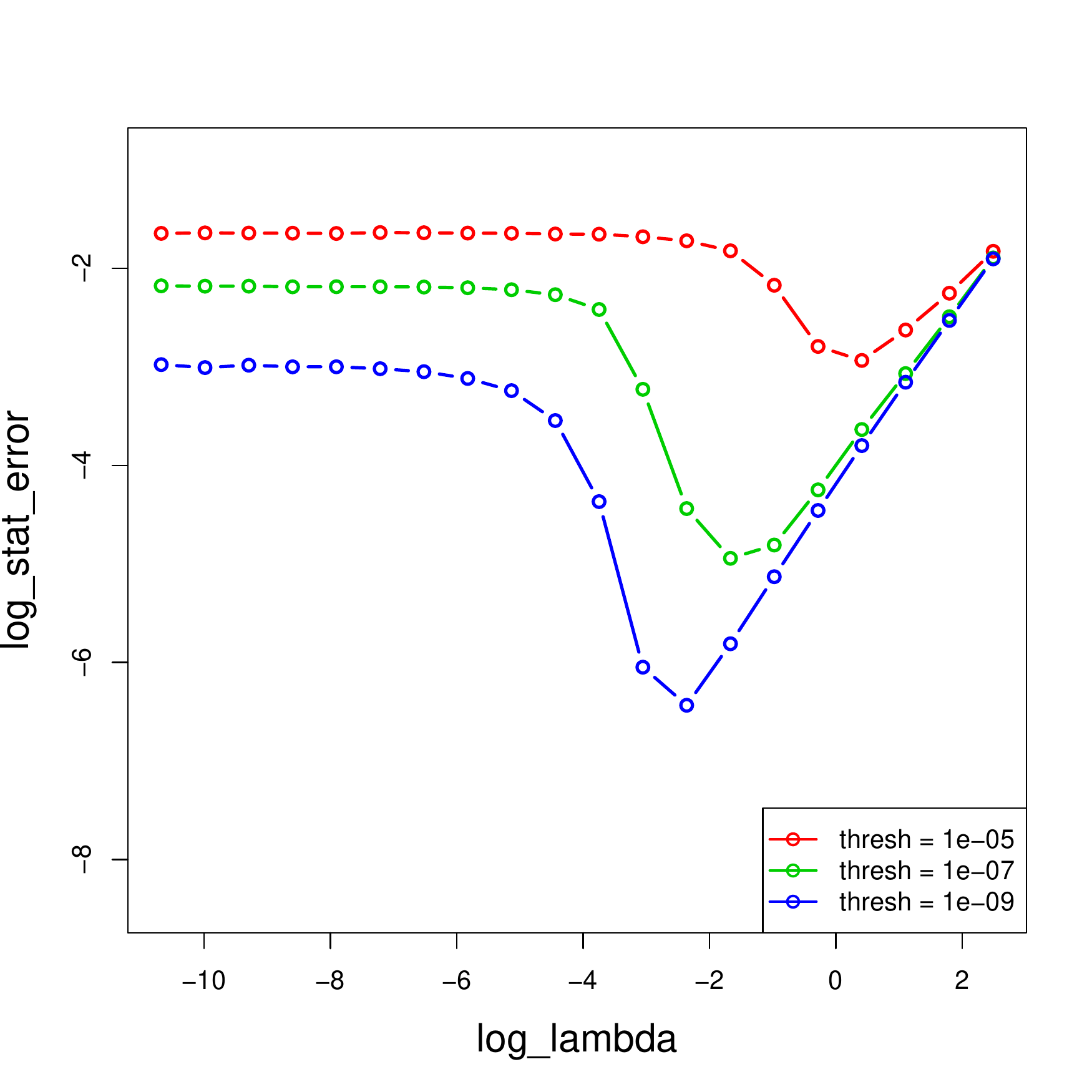} \\
    	\includegraphics[clip, trim=0 .5cm 0 1cm, width=5.0cm]{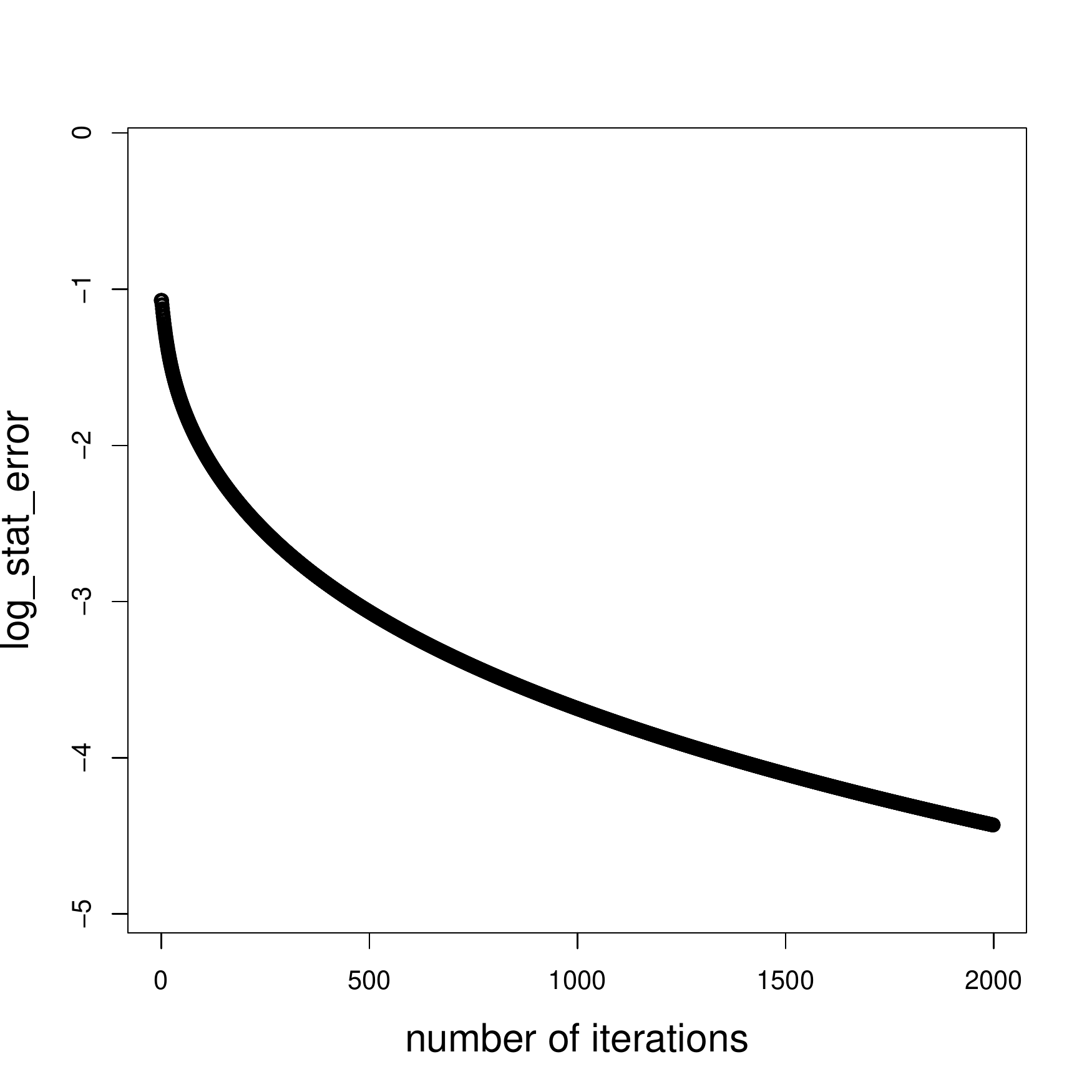} & \includegraphics[clip, trim=0 0.5cm 0 1cm, width=5.0cm]{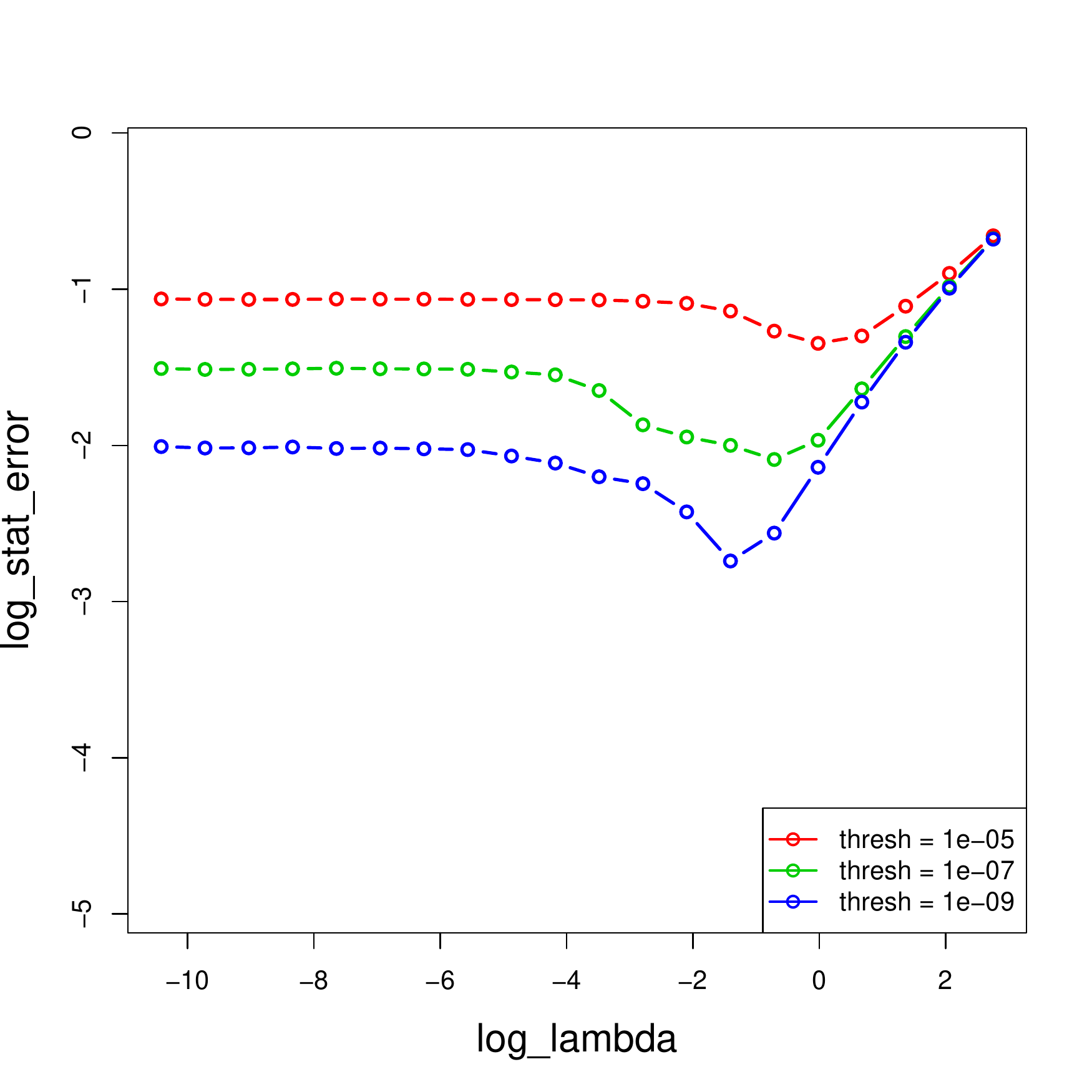} \\	
    	\includegraphics[clip, trim=0 .5cm 0 1cm, width=5.0cm]{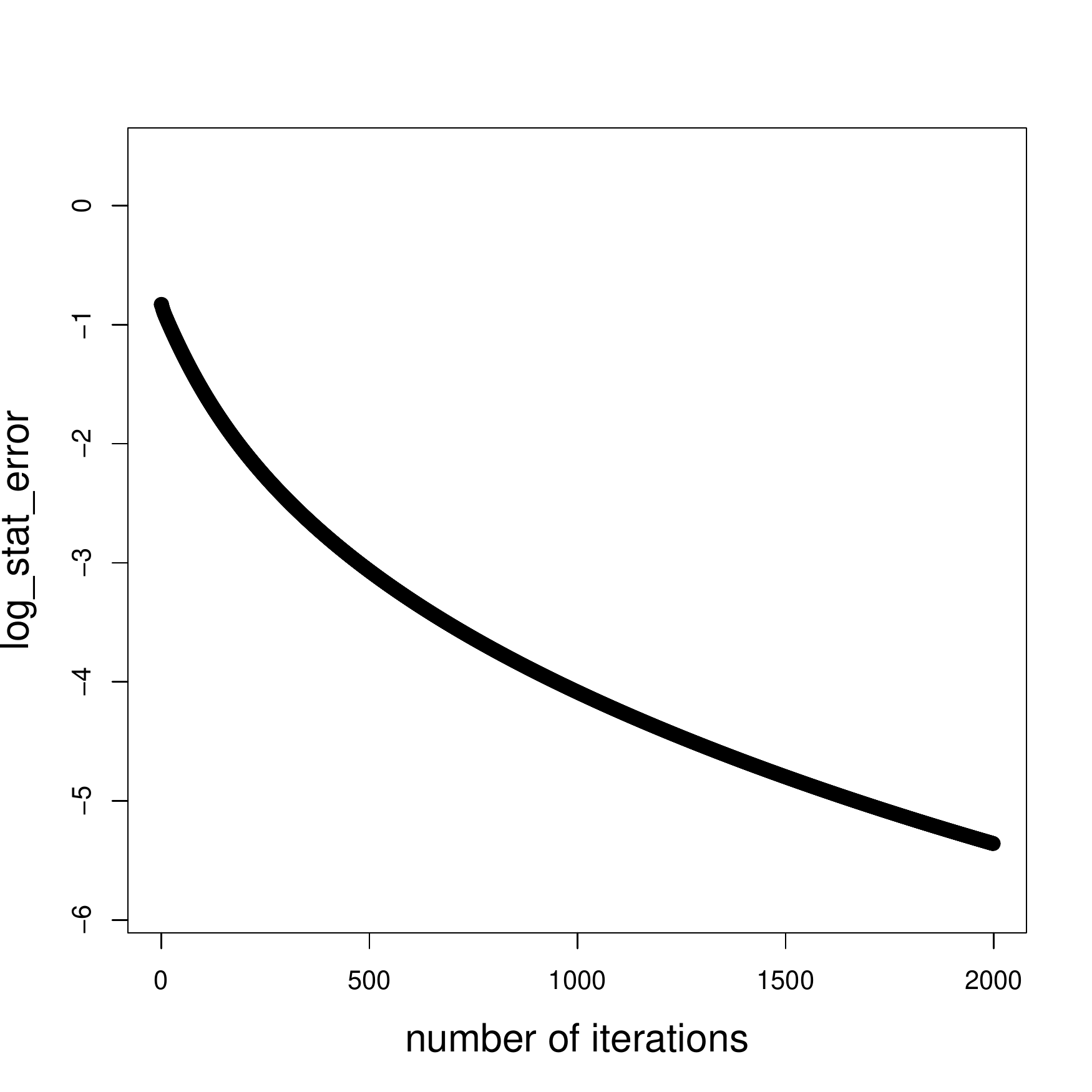} & \includegraphics[clip, trim=0 0.5cm 0 1cm, width=5.0cm]{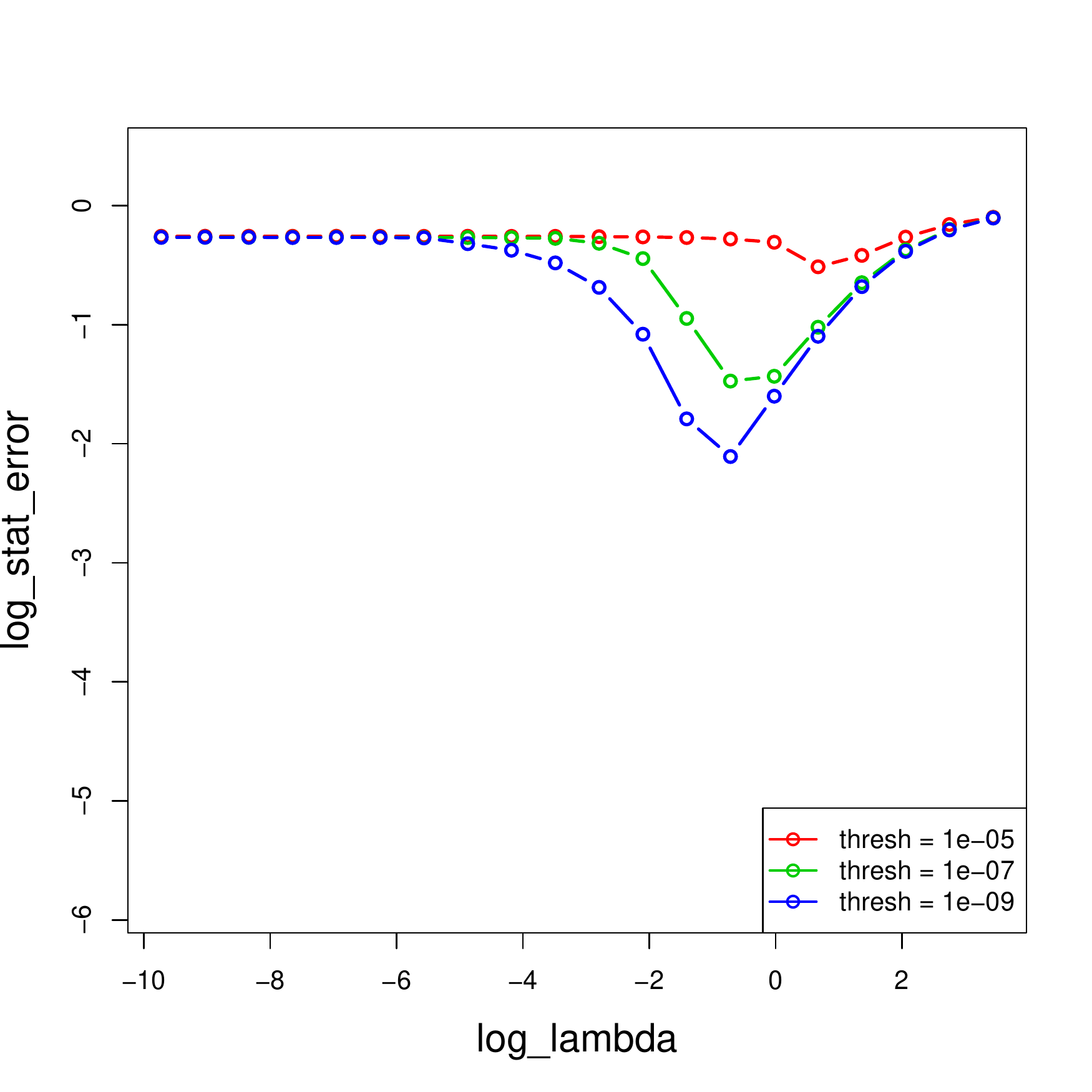}\\
    	\includegraphics[clip, trim=0 .5cm 0 1cm, width=5.0cm]{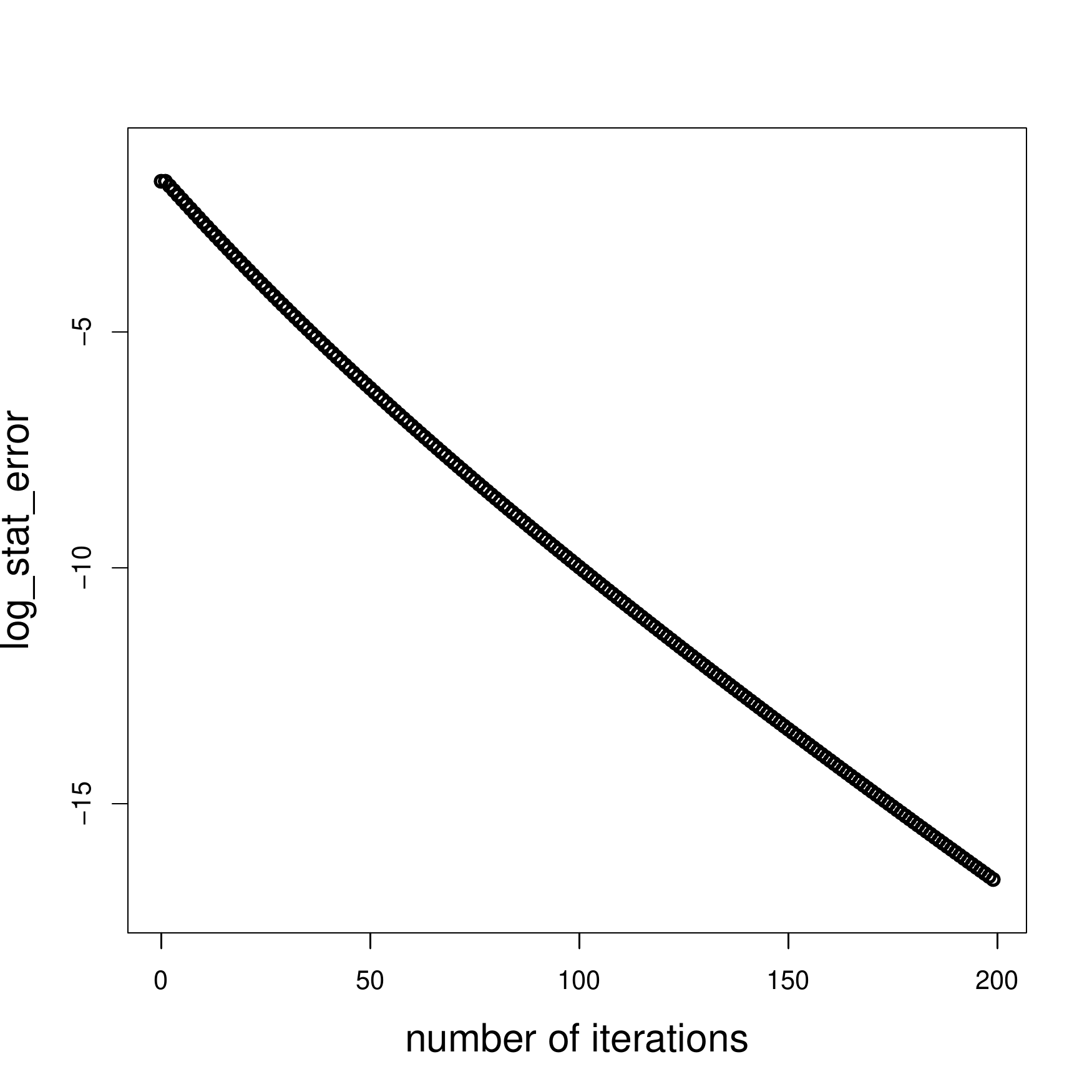} & \includegraphics[clip, trim=0 0.5cm 0 1cm, width=5.0cm]{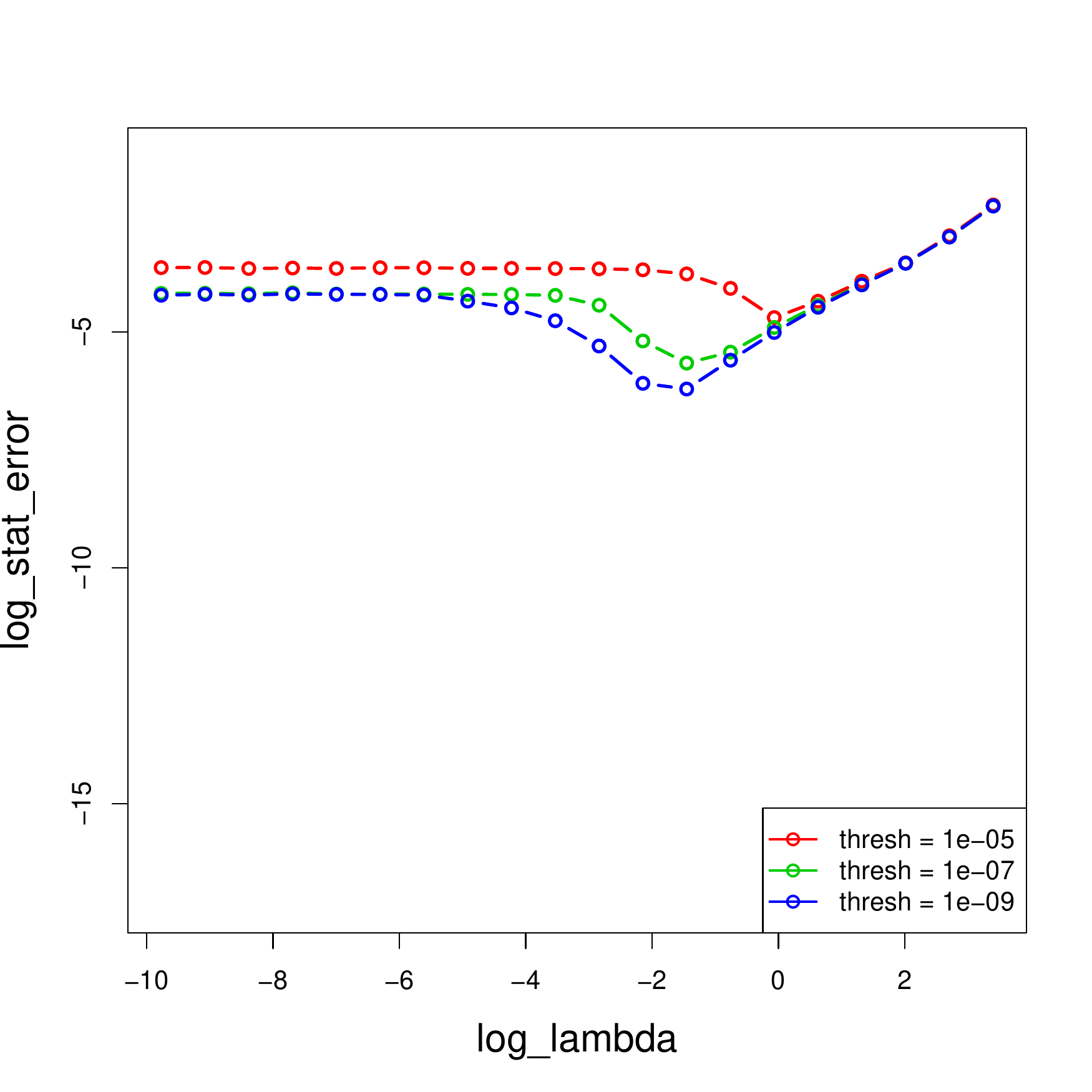}	
    \end{tabular}
    \caption{\label{fig:1}Logarithms of the average Frobenius norm $\sin \Theta$ error of \texttt{primePCA} and \texttt{softImpute} under various heterogeneity levels of missingness in absence of noise. The three rows of plots above, from the top to bottom, correspond to (H1), (H2), (H3) and (H4).}
\end{figure}


\subsection{Noisy case}
\label{Sec:Noisy}

Here, we generate the rows of $\bZ$ as independent $N_d(\bzero,\bI_d)$ random vectors, independent of all other data.  We maintain the same choices of $n$, $d$, $K$ and $\bV_K$ as in Section~\ref{Sec:Noiseless}, set $\bSigma_{\bu} = \nu^2\bI_{2}$ and vary $\nu > 0$ to achieve different signal-to-noise ratios.  In particular, defining $\mathrm{SNR} := \Tr \mathrm{Cov}(\bx_1)/\Tr \mathrm{Cov}(\bz_1)$, the choices $\nu=20, 40, 60$ correspond to the low, medium and high signal-to-noise ratios $\mathrm{SNR} = 1.6, 6.4, 14.4$, respectively.  For an additional comparision, we also considered a variant of the \texttt{softImpute} algorithm called \texttt{hardImpute} \citep{MHT10}, which retains only a fixed number of top singular values in each iteration of matrix imputation; this can be achieved by setting the argument~\texttt{$\lambda$} in the \texttt{softImpute} function to be $0$.

We remark that in general, the choice of $\lambda$ for \texttt{softImpute} is more challenging than in many regularised $M$-estimation contexts, because in our setting we have no response variable, so cross-validation techniques are less readily available.  For our comparisons, therefore, we gave the \texttt{softImpute} algorithm a particularly strong form of oracle choice of $\lambda$, namely where $\lambda$ was chosen for each individual repetition of the experiment, so as to minimise the loss function.  Naturally, such a choice is not available to the practitioner.  Moreover, in order to ensure the range of $\lambda$ was wide enough to include the best \texttt{softImpute} solution, we set the argument $\texttt{rank.max}$ in that algorithm to be $20$.



In Table~\ref{Table:1}, we report the statistical error of \texttt{primePCA} after $2000$ iterations of refinement, together with the corresponding statistical errors of our initial estimator \texttt{primePCA\_init} and those of \texttt{softImpute(oracle)} and \texttt{hardImpute}.  Remarkably, \texttt{primePCA} exhibits stronger performance than these other methods across each of the signal-to-noise ratio regimes and different missingness mechanisms.  We also remark that \texttt{hardImpute} is inaccurate and unstable, because it might converge to the local optimum that is far from the truth. 

\vspace{.1cm}

\begin{table}
        \caption{\label{Table:1}Average losses (with standard errors in brackets) under (H1), (H2), (H3) and (H4).}
    \centering
    \begin{tabular}{rlll}
    	\hline\hline
    	& $\nu = 20$ & $\nu = 40$ & $\nu = 60$ \\
    	\hline
    	(H1) \hfill \texttt{hardImpute} & $0.444_{(0.001)}$& $0.251_{(0.001)}$ & $0.186_{(0.0005)}$\\
    	\texttt{~~~~softImpute(oracle)} & $0.186_{(0.0004)}$& $0.095_{(0.0002)}$ & $0.064_{(0.0002)}$ \\
    	\texttt{primePCA\_init} & $0.306_{(0.001)}$& $0.266_{(0.001)}$ & $0.259_{(0.001)}$\\
    	\texttt{primePCA} & $0.171_{(0.0004)}$& $0.084_{(0.0002)}$ & $0.056_{(0.0001)}$ \\
    	\hline 
    	(H2) \hfill \texttt{hardImpute} & $0.473_{(0.001)}$& $0.291_{(0.001)}$ & $0.236_{(0.001)}$\\
    	\texttt{softImpute(oracle)} & $0.308_{(0.001)}$& $0.185_{(0.001)}$ & $0.141_{(0.001)}$ \\
    	\texttt{primePCA\_init} & $0.399_{(0.002)}$& $0.357_{(0.001)}$ & $0.349_{(0.001)}$\\
    	\texttt{primePCA} & $0.232_{(0.001)}$& $0.115_{(0.001)}$ & $0.077_{(0.0005)}$ \\
	\hline 
    	(H3) \hfill \texttt{hardImpute} & $0.479_{(0.001)}$& $0.385_{(0.001)}$ & $0.427_{(0.001)}$\\
    	\texttt{softImpute(oracle)} & $0.374_{(0.001)}$& $0.222_{(0.001)}$ & $0.170_{(0.001)}$ \\
    	\texttt{primePCA\_init} & $0.486_{(0.001)}$& $0.449_{(0.001)}$ & $0.442_{(0.001)}$\\
    	\texttt{primePCA} & $0.290_{(0.001)}$& $0.145_{(0.001)}$ & $0.097_{(0.0004)}$ \\
	\hline
    	(H4) \hfill \texttt{hardImpute} & $0.174_{(0.0005)}$& $0.089_{(0.0003)}$ & $0.062_{(0.0003)}$\\
    	\texttt{softImpute(oracle)} & $0.121_{(0.0002)}$& $0.062_{(0.0001)}$ & $0.042_{(0.0001)}$ \\
    	\texttt{primePCA\_init} & $0.203_{(0.001)}$& $0.175_{(0.0005)}$ & $0.169_{(0.0004)}$\\
    	\texttt{primePCA} & $0.116_{(0.0003)}$& $0.058_{(0.0002)}$ & $0.038_{(0.0001)}$ \\
    	\hline \hline
    \end{tabular}
\end{table}

\subsection{Near low-rank case}
\label{Sec:Sensitivity}

Here, we set $n = 2000$, $d = 500$, $K=10$, $\bSigma_{\bu} = \diag(2^{10}, 2^9, \ldots, 2)$, and 
fixed~$\bV_K$ once for all experiments to be the top $K$ eigenvectors of one realisation\footnote{In \texttt{R}, we set the random seed to be $2019$ before generating $\bV_K$.} of the sample covariance matrix of $n$ independent $N_d(\bzero, \bI_d)$ random vectors.  Here $\|\bV_K\|_{\infty} / d^{1 / 2}< 3.63$, and we again generated the rows of $\bZ$ as independent $N_d(\bzero,\bI_d)$ random vectors. Table~\ref{Table:2} reports the average loss of estimating the top $\widehat K$ eigenvectors of $\bSigma_{\by}$, where $\widehat K$ varies from $1$ to $5$.  Interestingly, even in this misspecified setting, \texttt{primePCA} is competitive with the oracle version of \texttt{softImpute}.
\begin{table}
    \caption{\label{Table:2}Average losses (with standard errors in brackets) in the setting of Section~\ref{Sec:Sensitivity} under (H1), (H2), (H3) and (H4).}
    \centering
    \begin{tabular}{rlllll}
    	\hline\hline
    	& $\widehat K = 1$ & $\widehat K = 2$ & $\widehat K = 3$ & $\widehat K = 4$ & $\widehat K = 5$\\
    	\hline
    	(H1) \hfill\texttt{hardImpute} & $0.308_{(0.002)}$ & $0.507_{(0.002)}$ & $0.764_{(0.004)}$ & $1.199_{(0.006)}$ & $1.524_{(0.004)}$ \\
    	\texttt{~~~softImpute(oracle)} & $0.107_{(0.001)}$ & $0.182_{(0.001)}$ & $0.275_{(0.001)}$ & $0.401_{(0.001)}$ & $0.596_{(0.001)}$ \\
    	 \texttt{primePCA\_init} & $0.203_{(0.001)}$& $0.345_{(0.001)}$ & $0.554_{(0.003)}$ & $1.074_{(0.007)}$ & $1.427_{(0.006)}$ \\
    	 \texttt{primePCA} & $0.141_{(0.001)}$& $0.200_{(0.001)}$ & $0.269_{(0.001)}$ & $0.374_{(0.001)}$ & $0.580_{(0.001)}$\\
    	\hline 
    	(H2) \hfill\texttt{hardImpute} & $0.298_{(0.002)}$ & $0.466_{(0.002)}$ & $0.696_{(0.003)}$ & $1.124_{(0.006)}$ & $1.452_{(0.004)}$ \\
    	\texttt{softImpute(oracle)} & $0.188_{(0.001)}$ &  $0.283_{(0.001)}$ & $0.410_{(0.001)}$ & $0.562_{(0.001)}$& $0.751_{(0.001)}$ \\
    	\texttt{primePCA\_init} & $0.285_{(0.001)}$& $0.443_{(0.004)}$ & $0.757_{(0.013)}$ & $1.201_{(0.004)}$ & $1.533_{(0.003)}$\\
    	\texttt{primePCA} & $0.190_{(0.002)}$& $0.267_{(0.002)}$ & $0.368_{(0.003)}$ & $0.543_{(0.008)}$ & $0.797_{(0.009)}$ \\
	\hline 
    	(H3) \hfill\texttt{hardImpute} & $0.302_{(0.001)}$ & $0.482_{(0.002)}$ & $0.695_{(0.002)}$ & $1.004_{(0.006)}$ & $1.373_{(0.004)}$ \\
    	\texttt{softImpute(oracle)} & $0.206_{(0.001)}$& $0.338_{(0.001)}$ & $0.492_{(0.001)}$ & $0.664_{(0.002)}$ & $0.878_{(0.002)}$\\
    	\texttt{primePCA\_init} & $0.341_{(0.001)}$& $0.528_{(0.019)}$ & $1.097_{(0.008)}$ & $1.306_{(0.008)}$ & $1.597_{(0.004)}$ \\
    	\texttt{primePCA} & $0.222_{(0.001)}$& $0.330_{(0.002)}$ & $0.452_{(0.003)}$ & $0.641_{(0.008)}$ & $0.919_{(0.007)}$ \\    	
	\hline
    	(H4) \hfill\texttt{hardImpute} & $0.090_{(0.001)}$ & $0.148_{(0.001)}$ & $0.226_{(0.001)}$ & $0.346_{0.002}$ & $0.589_{(0.007)}$\\
        \texttt{softImpute(oracle)} & $0.071_{(0.001)}$ & $0.112_{(0.001)}$ & $0.164_{(0.001)}$ & $0.233_{(0.001)}$ & $0.332_{(0.001)}$ \\
        \texttt{primePCA\_init} & $0.139_{(0.001)}$& $0.220_{(0.001)}$ & $0.325_{(0.001)}$ & $0.475_{(0.002)}$ & $0.805_{(0.012)}$\\
        \texttt{primePCA} & $0.098_{(0.001)}$& $0.135_{(0.001)}$ & $0.176_{(0.001)}$ & $0.236_{(0.001)}$ & $0.328_{(0.001)}$\\
    \hline \hline
    \end{tabular}
\end{table}

\section{Real data analysis: Million Song Dataset}
\label{Sec:RealData}

We apply \texttt{primePCA} to a subset of the Million Song Dataset\footnote{\texttt{https://www.kaggle.com/c/msdchallenge/data}} to analyse music preferences.  The original data can be expressed as a matrix with $110{,}000$ users (rows) and $163{,}206$ songs (columns), with entries representing the number of times a song was played by a particular user.  The proportion of non-missing entries in the matrix is $0.008\%$.  Since the matrix is very sparse, and since most songs have very few listeners, we enhance the signal-to-noise ratio by restricting our attention to songs that have at least 100 listeners ($1{,}777$ songs in total). This improves the proportion of non-missing entries to $0.23\%$.  Further summary information about the filtered data is provided below:
\begin{enumerate}
	\item Quantiles of non-missing matrix entry values:
          \begin{center}
        		\begin{tabular}{cccccccccccc}
                    	\hline\hline
                       $0\%$ & $10\%$ & $20\%$ & $30\%$ & $40\%$ & $50\%$ & $60\%$ & $70\%$ & $80\%$ & $90\%$ & $100\%$ \\ \hline
                      $1$ & $1$ & $1$ & $1$ & $1$ & $1$ & $2$ & $3$ & $5$ & $8$ & $500$ \\
                      \hline\hline
    		  \\ 
                    	\hline\hline
                       $90\%$ & $91\%$ & $92\%$ & $93\%$ & $94\%$ & $95\%$ & $96\%$ & $97\%$ & $98\%$ & $99\%$ & $100\%$ \\ \hline
                      $8$ & $9$ & $9$ & $10$ & $11$ & $13$ & $15$ & $18$ & $23$ & $33$ & $500$ \\
                      \hline\hline
    		\end{tabular}
	\end{center}
	
	\item Quantiles of the number of listeners for each song:  
          \begin{center}
        		\begin{tabular}{cccccccccccc}
                    	\hline\hline
                       $0\%$ & $10\%$ & $20\%$ & $30\%$ & $40\%$ & $50\%$ & $60\%$ & $70\%$ & $80\%$ & $90\%$ & $100\%$ \\
                      \hline
                      $100$ & $108$ & $117$ & $126$ & $139$ & $154$ & $178$ & $214$ & $272.8$ & $455.6$ & $5043$ \\
                      \hline\hline
    		\end{tabular}
        \end{center}
	\item Quantiles of the total play counts of each user: 
          \begin{center}
        		\begin{tabular}{cccccccccccc}
                    	\hline\hline
                       $0\%$ & $10\%$ & $20\%$ & $30\%$ & $40\%$ & $50\%$ & $60\%$ & $70\%$ & $80\%$ & $90\%$ & $100\%$ \\
                      \hline
                      $0$ & $0$ & $1$ & $3$ & $4$ & $6$ & $9$ & $14$ & $21$ & $38$ & $1114$ \\
                      \hline\hline
                        \\
                    	\hline\hline
                       $90\%$ & $91\%$ & $92\%$ & $93\%$ & $94\%$ & $95\%$ & $96\%$ & $97\%$ & $98\%$ & $99\%$ & $100\%$ \\ \hline
                      $38$ & $41$ & $44$ & $48$ & $54$ & $60$ & $68$ & $79$ & $97$ & $132$ & $1114$ \\
                      \hline\hline
    		\end{tabular}
	\end{center}
	
\end{enumerate}
 Moreover, from the first numbered point above, we see that the distribution of play counts has an extremely heavy tail. To guard against excessive influence from the outliers, we discretise the play counts into five interest levels as follows: 
\begin{center}
    \begin{tabular}{cccccc} 
	\hline\hline
    	Play count & 1 & 2 -- 3 & 4 -- 6 & 7 -- 10 & $\ge 11$ \\ \hline
    	Level of interest & 1 & 2 & 3 & 4 & 5 \\
	\hline\hline
	\end{tabular}
\end{center}

\begin{figure}[htbp]
	\centering
	\includegraphics[width=0.5\textwidth, clip, trim=0 .5cm 0 1cm]{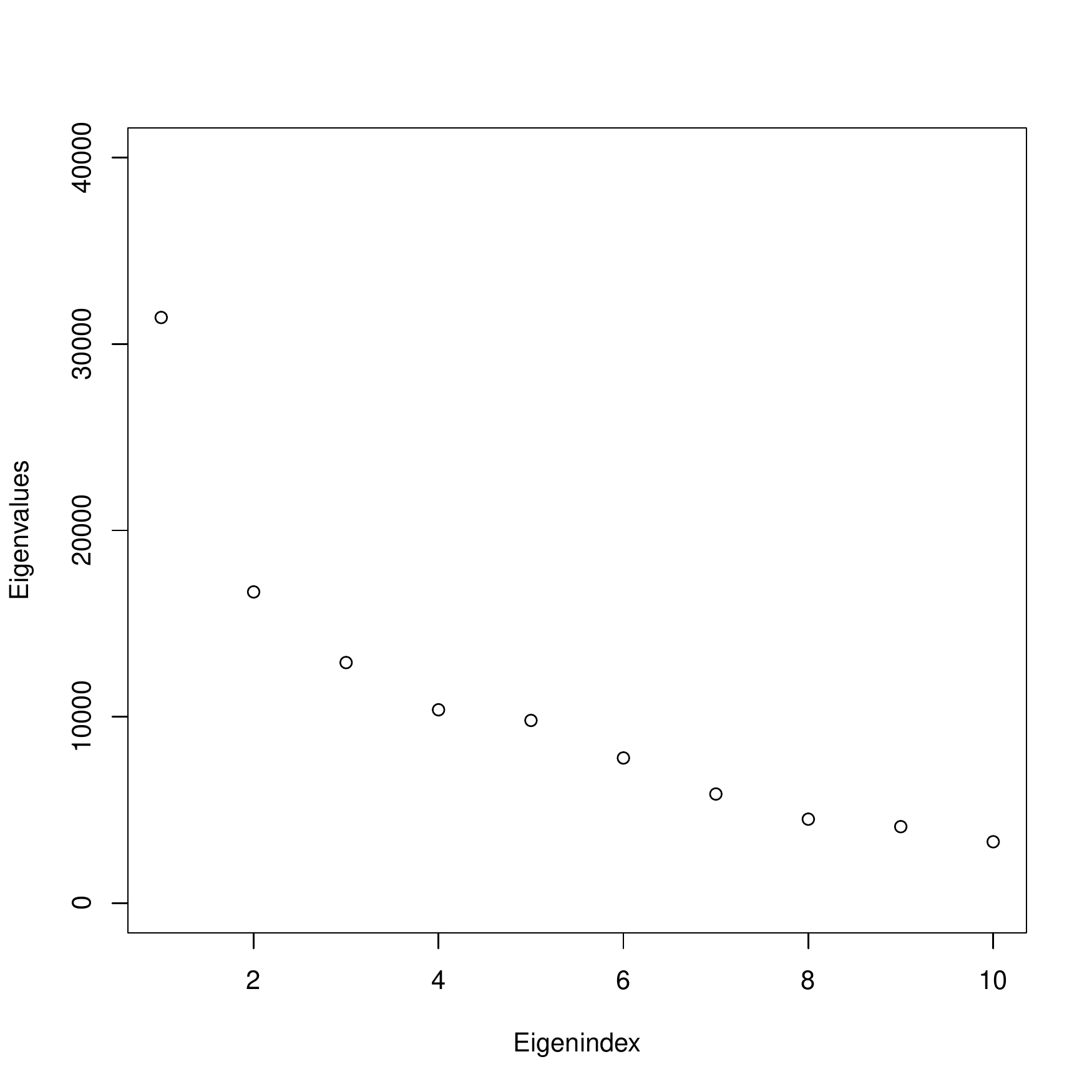}
	\caption{\label{fig:eigenvalue}Leading eigenvalues of $\widehat{\bSigma}_\by$.}
\end{figure}
We are now in a position to analyse the data using \texttt{primePCA}.  For $i=1,\ldots,n=110{,}000$ and $j = 1,\ldots,d = 1{,}777$, let $Y_{ij} \in \{1, \ldots, 5\}$ denote the level of interest of user $i$ in song~$j$, let $\widehat K = 10$ and let $\cI = \{i: \lonenorm{\omega_i} > \widehat K\}$.  Our initial goal is to assess the top $\widehat K$ eigenvalues of $\bSigma_{\by}$ to see if there is low-rank signal in $\bY = (Y_{ij})$.  To this end, we first apply Algorithm~\ref{algo:2} to obtain $\bV_{\widehat K}^{\texttt{prime}}$; next, for each $i \in \cI$, we run Steps 2--5 of Algorithm~\ref{algo:1} to obtain the estimated principal score $\widehat \bu_i$, so that we can approximate $\by_i$ by $\widehat \by_i = \widehat \bV_{\widehat K}^{\texttt{prime}}\widehat \bu_i$.  This allows us to estimate $\bSigma_{\by}$ by $ \widehat \bSigma_{\by} = n^{-1} \sum_{i \in \cI} \widehat\by_i \widehat \by_i^\top $.  Figure~\ref{fig:eigenvalue} displays the top $\widehat K$ eigenvalues of $\widehat \bSigma_{\by}$, which exhibit a fairly rapid decay, thereby providing evidence for the existence of low-rank signal in $\bY$.


In the left panel of Figure~\ref{fig:msd_primepca}, we present the estimate $\widehat \bV^{\text{prime}}_2$ of the top two eigenvectors of the covariance matrix $\bSigma_{\by}$, with colours indicating the genre of the song.  The outliers in the $x$-axis of this plot are particularly interesting: they reveal songs that polarise opinion among users (see Table~\ref{tab:outlier_songs}) and that best capture variation in individuals' preferences for types of music measured by the first principal component.  It is notable that Rock songs are overrepresented among the outliers (see Table~\ref{tab:genre_contrast}), relative to, say, Country songs.  Users who express a preference for particular songs are also more likely to enjoy songs that are nearby in the plot.  Such information is therefore potentially commercially valuable, both as an efficient means of gauging users' preferences, and for providing recommendations.  

The right panel of Figure~\ref{fig:msd_primepca} presents the principal scores $\{\widehat \bu_i\}_{i = 1}^n$ of the users, with frequent users (whose total song plays are in the top 10\% of all users) in red and occasional users in blue.  This plot reveals, for instance, that the second principal component is well aligned with general interest in the website.  Returning to the left plot, we can now interpret a positive $y$-coordinate for a particular song (which is the case for the large majority of songs) as being associated with an overall interest in the music provided by the site.

\begin{figure}[H]
	\centering
	\begin{tabular}{cc}
		\includegraphics[width=0.47\textwidth]{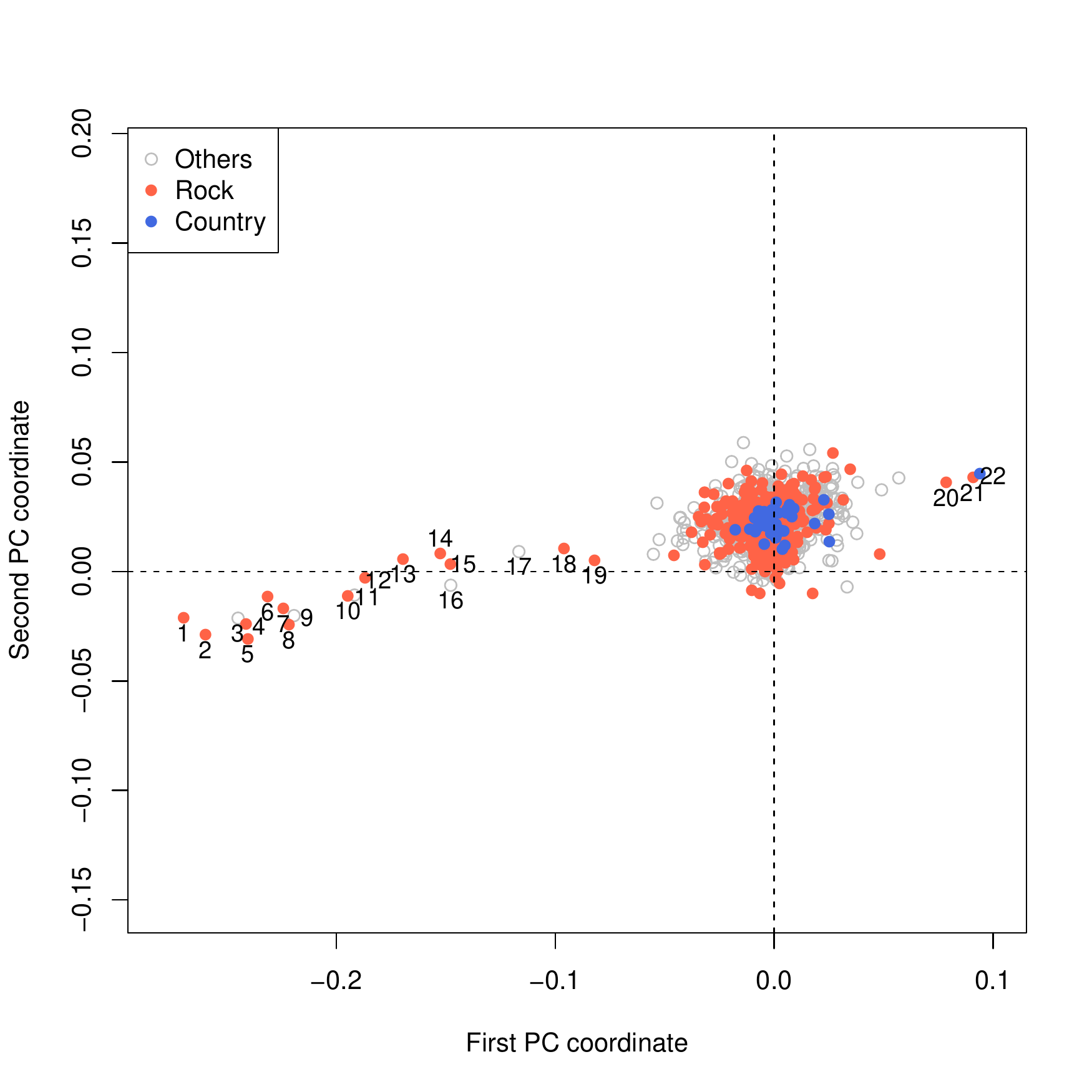} & \includegraphics[width=0.47\textwidth]{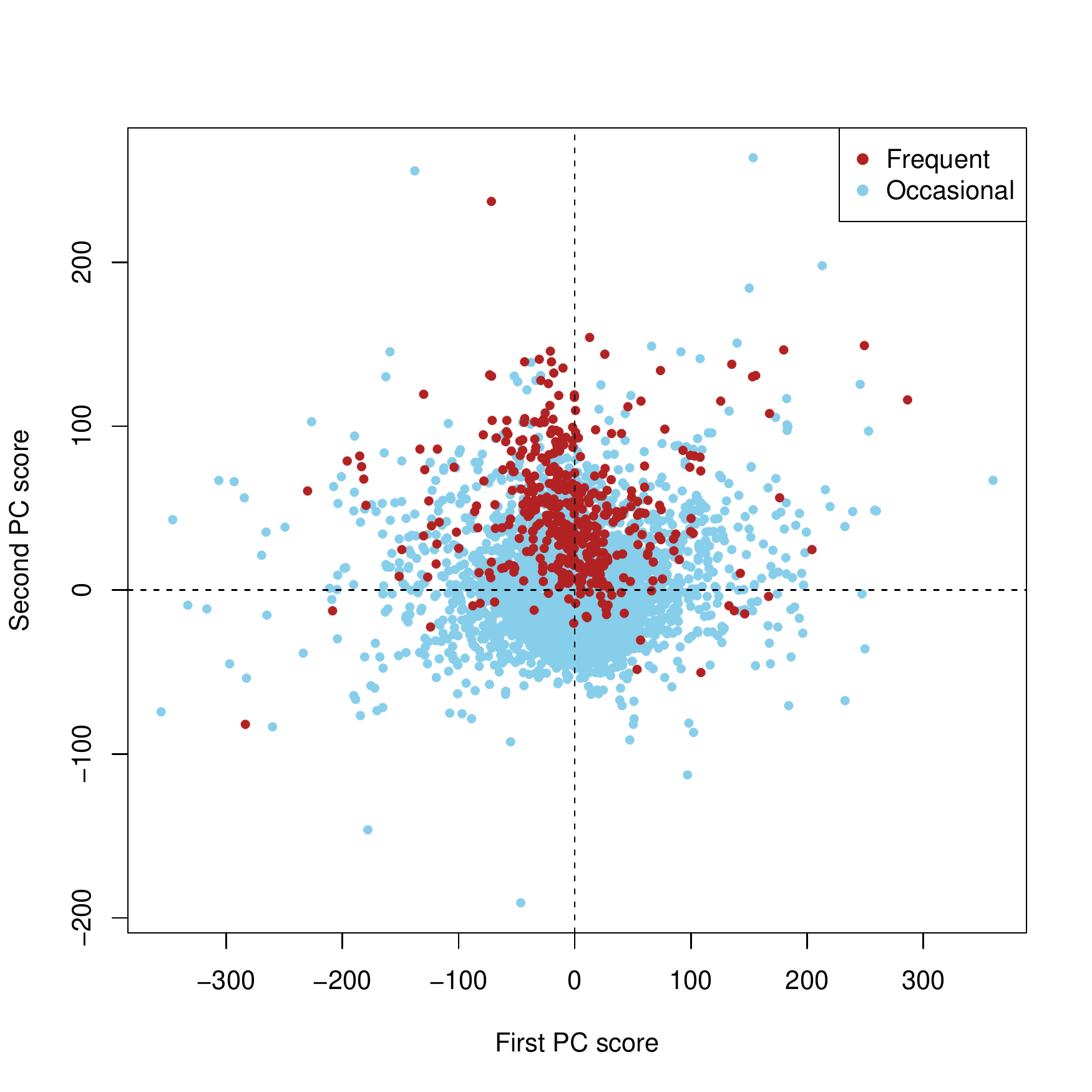}\\
	\end{tabular}
	\caption{Plots of the first two principal components $\widehat \bV^{\text{prime}}_2$ (left) and the associated scores $\{\widehat \bu_i\}_{i = 1}^n$ (right).}
	\label{fig:msd_primepca}
\end{figure}

\begin{table}
		\caption{\label{tab:outlier_songs}Titles, artists and genres of the $22$ outlier songs in Figure~\ref{fig:msd_primepca}.}
	\centering 
	\begin{tabular}{cccc}
	\hline\hline
	 ID & Title & Artist & Genre  \\ \hline
	1 & Your Hand In Mine & Explosions In The Sky & Rock \\                 
 	2 & All These Things That I've Done & The Killers & Rock \\             
	3 & Lady Marmalade & Christina Aguilera / Lil' Kim/  & Pop \\
	& & Mya / Pink & \\
	4 & Here It Goes Again & Ok Go & Rock \\
 	5 & I Hate Pretending (Album Version) & Secret Machines & Rock \\       
 	6 & No Rain & Blind Melon & Rock \\
 	7 & Comatose (Comes Alive Version) & Skillet & Rock \\
 	8 & Life In Technicolor & Coldplay & Rock \\                            
 	9 & New Soul & Yael Naïm & Pop \\
	10 & Blurry & Puddle Of Mudd & Rock \\                                   
	11 & Give It Back & Polly Paulusma & Pop \\                              
	12 & Walking On The Moon & The Police & Rock \\
	13 & Face Down (Album Version) & The Red Jumpsuit Apparatus & Rock \\    
	14 & Savior & Rise Against & Rock \\
	15 & Swing Swing & The All-American Rejects & Rock \\
	16 & Without Me & Eminem & Rap \\
	17 & Almaz & Randy Crawford & Pop \\                                     
	18 & Hotel California & Eagles & Rock \\                                 
	19 & Hey There Delilah & Plain White T's & Rock \\                       
	20 & Revelry & Kings Of Leon & Rock \\                                   
	21 & Undo & Björk & Rock \\
	22 & You're The One & Dwight Yoakam & Country \\ 	\hline\hline             	
	\end{tabular}
\end{table}

\begin{table}
		\caption{\label{tab:genre_contrast}Genre distribution of the outliers (songs whose corresponding coordinate in the estimated leading principal component is of magnitude larger than $0.07$).}
	\centering
	\begin{tabular}{cccccccccccc}
		\hline\hline
		& Rock & Pop & Electronic & Rap & Country & RnB & Latin & Others \\ \hline
		Population & $48.92\%$ & $18.53\%$ & $9.12\%$ & $7.15\%$ &$4.33\%$ &$2.35\%$ &$2.26\%$ &$7.34\%$ \\ 
		Outliers & $72.73\%$ & $18.18\%$ & $0\%$ & $4.54\%$ &$4.54\%$ &$0\%$ &$0\%$ &$0\%$ \\ \hline\hline
	\end{tabular}
\end{table}

\section{Discussion}
\label{Sec:Discussion}
Heterogeneous missingness is ubiquitous in contemporary, large-scale data sets, yet we currently understand very little about how existing procedures perform or should be adapted to cope with the challenges this presents.  Here we attempt to extract the lessons learned from this study of high-dimensional PCA, in order to see how related ideas may be relevant in other statistical problems where one wishes to recover low-dimensional structure with data corrupted in a heterogeneous manner.

A key insight, as gleaned from Section~\ref{sec:2.2}, is that the way in which the heterogeneity interacts with the underlying structure of interest is crucial.  In the worst case, the missingness may be constructed to conceal precisely the structure one seeks to uncover, thereby rendering the problem infeasible by any method.  The only hope, then, in terms of providing theoretical guarantees, is to rule out such an adversarial interaction.  This was achieved via our incoherence condition in Section~\ref{Sec:WithIncoherence}, and we look forward to seeing how the relevant interactions between structure and heterogeneity can be controlled in other statistical problems such as those mentioned in the introduction.  For instance, in sparse linear regression, one would anticipate that missingness of covariates with strong signal would be much more harmful than corresponding missingness for noise variables.


Our study also contributes to the broader understanding of the uses and limitations of spectral methods for estimating hidden low-dimensional structures in high-dimensional problems.  We have seen that the IPW estimator is both methodologically simple and achieves near-minimax optimality when the noise level is of constant order.  Similar results have been obtained for spectral clustering for network community detection in stochastic block models \citep{RCY11} and in low-rank-plus-sparse matrix estimation problems \citep{FLM13}.  On the other hand, the IPW estimator fails to provide exact recovery of the principal components in the noiseless setting.  In these other aforementioned problems, it has also been observed that refinement of an initial spectral estimator can enhance performance, particularly in high signal-to-noise ratio regimes \citep{GMZZ2016,ZCW18}, as we were able to show for our \texttt{primePCA} algorithm.  This suggests that such a refinement has the potential to confer a sharper dependence of the statistical error rate on the signal-to-noise ratio compared with a vanilla spectral algorithm, and understanding this phenomenon in greater detail provides another interesting avenue for future research.

\section{Proofs of main results}

 \label{Sec:Proofs}
We define two linear maps $\mathcal{D}, \mathcal{F}: \mathbb{R}^{d\times d}\to \mathbb{R}^{d\times d}$, such that for any $\bA = (A_{ij}) \in\mathbb{R}^{d\times d}$, we have $[\mathcal{D}(\bA)]_{ij} := A_{ij}\mathds{1}_{\{i=j\}}$ and $\mathcal{F}(\bA) := \bA - \mathcal{D}(\bA)$. In other words, $\mathcal{D}(\bA)$ and $\mathcal{F}(\bA)$ correspond to the diagonal and off-diagonal parts of $\bA$ respectively.

\begin{proof}[Proof of Theorem \ref{thm:1}]
Since $\by_i = \bV_K \bu_i + \bz_i$, we have that
\begin{equation}
\label{Eq:Psi2vec}
\|\by_i\|_{\psi_2} \leq  \|\bV_K\bu_i\|_{\psi_2} + \|\bz_i\|_{\psi_2} =\|\bu_i\|_{\psi_2} + \|\bz_i\|_{\psi_2}\leq   (\lambda^{1/2}+1)\tau.
\end{equation}
Moreover, since $\max_{j \in [d]} \|y_{1j}\|_{\psi_2}\leq  M^{1/2}$ by Lemma~\ref{lem:psi1_to_psi2}, it follows from \citet[Lemma~2.2.2]{vanderVaartWellner1996} that there exist a universal constant $C>0$ such that\footnote{In \citet{vanderVaartWellner1996}, the $\psi_2$-norm of a random variable is defined slightly differently as $\|X\|_{\psi_{2}}:= \inf\{a: \mathbb{E} e^{(X/a)^2} \leq  2\}$. It can be shown \citep[Lemma~5.5]{Ver10} that these two norms are equivalent.}
\begin{equation}
\label{Eq:Psi2max}
\bigl\|\|\by_i\|_\infty\bigr\|_{\psi_2} \leq  \{CM\log d\}^{1/2}.
\end{equation}
Recall that $\widetilde\by_i^\top = (\widetilde y_{i1},\ldots,\widetilde y_{id})$ denotes the $i$th row of $\bY_\bOmega$. Define $\bA_i := \mathcal{F}(\widetilde \by_i \widetilde \by_i^\top)$ and $\bB_i := \mathcal{D}(\widetilde \by_i \widetilde \by_i^\top)$. We have the following decomposition of $\widehat \bG$:
\begin{align*}
    	\widehat \bG  & = \frac{1}{n}\sum_{i=1}^n \biggl(\frac{1}{\widehat p^2}\bA_i - \frac{1}{p^2} \E \bA_i\biggr) + \frac{1}{n}\sum_{i=1}^n \biggl(\frac{1}{\widehat p} \bB_i - \frac{1}{p} \E \bB_i\biggr) + \bSigma_{\by} \\
	& = \frac{1}{n\widehat p^2} \sum_{i = 1}^n (\bA_i - \E \bA_i ) + \frac{1}{n\widehat p} \sum_{i = 1}^n (\bB_i - \E \bB_i) + \biggl(\frac{1}{\widehat p^2} - \frac{1}{p^2}\biggr) \E \bA_1 + \biggl(\frac{1}{\widehat p} - \frac{1}{p}\biggr) \E \bB_1 + \bSigma_{\by} \\
	& = \frac{1}{n\widehat p^2} \sum_{i = 1}^n (\bA_i - \E \bA_i) + \frac{1}{n\widehat p} \sum_{i = 1}^n (\bB_i - \E \bB_i) + \biggl(\frac{p^2}{\widehat p^2} - 1\biggr) \cF(\bSigma_{\by}) + \biggl( \frac{p}{\widehat p} - 1\biggr) \cD(\bSigma_{\by}) + \bSigma_{\by} \\
	& = \frac{1}{n\widehat p^2} \sum_{i = 1}^n (\bA_i - \E \bA_i) + \frac{1}{n\widehat p} \sum_{i = 1}^n (\bB_i - \E \bB_i) + \biggl( \frac{p}{\widehat p} - \frac{p^2}{\widehat p^2}\biggr) \cD(\bSigma_{\by}) + \frac{p^2}{\widehat p^2} \bSigma_{\by}. 
\end{align*}

We regard $\widehat \bG$ as a perturbed version of $(p^2 / \widehat p^2)\bSigma_{\by}$. Applying \citet[Theorem~2]{YWS14}, we have
\begin{align}
	\label{eq:davis_kahan}
          L&(\widehat \bV_K,  \bV_K) \nonumber\\
          &\leq \frac{2K^{1/2}\widehat p^2}{ p^2  \lambda_K} \biggl \| \frac{1}{n\widehat p^2} \sum_{i = 1}^n (\bA_i - \E \bA_i) + \frac{1}{n\widehat p} \sum_{i = 1}^n (\bB_i - \E \bB_i) + \biggl( \frac{p}{\widehat p} - \frac{p^2}{\widehat p^2}\biggr) \cD(\bSigma_{\by})\biggr \|_{\text{op}} \nonumber\\
	& \le \frac{2K^{1/2}}{\lambda_K} \biggl (\biggl \| \frac{1}{np^2} \sum_{i = 1}^n (\bA_i - \E \bA_i) \biggr\|_{\text{op}}+ \biggl \|\frac{\widehat p}{np^2} \sum_{i = 1}^n (\bB_i - \E \bB_i) \biggr\|_{\text{op}}+ \biggl \|\biggl(\frac{\widehat p}{p} - 1 \biggr) \cD(\bSigma_{\by})\biggr\|_{\text{op}} \biggr). 
\end{align}
We will control the expectation of the three terms on the right-hand side of~\eqref{eq:davis_kahan} separately. Define $\widehat p_i := d^{-1}\sum_{j=1}^d \omega_{ij}$. For notational simplicity, we write $\mathbb{P}'$ and $\mathbb{E}'$ respectively for the probability and expectation conditional on $(\widehat p_1,\ldots,\widehat p_n)$. Also, let $\widehat p_i^{(2)} := \E'(\omega_{i1}\omega_{i2})$ and $\widehat p_i^{(3)} := \mathbb{E}'(\omega_{i1}\omega_{i2}\omega_{i3})$ (if $d = 2$, then $\widehat p_i^{(3)} :=0$).  For the first term, we apply a symmetrisation argument. Let $\{\bA_i^*\}_{i=1}^n$ denote copies of $\{\bA_i\}_{i=1}^n$ that are independent of $\{\bu_i,\bz_i,\bomega_i\}_{i = 1}^n$, let $\{\epsilon_i\}_{i = 1}^n$ be independent Rademacher random variables that are independent of $\{\bu_i,\bz_i,\bomega_i, \bA^*_i\}_{i = 1}^n$ and write $\mathbb{E}^*$ for expectation conditional on $\{\bu_i,\bz_i,\bomega_i\}_{i = 1}^n$.  Then by Jensen's inequality,
\begin{align}
 	\label{eq:symmetrization}
          \E \biggl \| \frac{1}{np^2} \sum_{i = 1}^n (\bA_i - \E \bA_i) \biggr\|_{\text{op}} & = \E \biggl \| \frac{1}{np^2} \sum_{i = 1}^n (\bA_i - \mathbb{E}^* \bA_i^*) \biggr\|_{\text{op}} \leq \E \biggl \| \frac{1}{np^2} \sum_{i = 1}^n (\bA_i - \bA_i^*) \biggr\|_{\text{op}} \nonumber\\
          &= \E \biggl\| \frac{1}{np^2} \sum_{i = 1}^n \epsilon_i(\bA_i - \bA_i^*)\biggr\|_{\mathrm{op}} \leq 2 \E \biggl \| \frac{1}{np^2} \sum_{i = 1}^n \epsilon_i\bA_i\biggr\|_{\text{op}}.
\end{align}
 Since $\bA_i = \widetilde \by_i \widetilde \by_i^\top - \mathcal{D}(\widetilde \by_i \widetilde \by_i^\top)$, we have that
\[
	\E'\{(\bA_i^2)_{jk} \mid \by_i\} = \begin{cases}
    \E'\{ \widetilde y_{ij}^2  \|\widetilde\by_i\|_2^2 - \widetilde y_{ij}^4 \mid \by_i\}  = \widehat p_i^{(2)} y^2_{ij}\sum_{{t} \neq j} y^2_{it}, & \text{if $j=k$,}\vspace{0.2cm}\\
     \sum_{t\notin \{j, k\}} \E'\{ \widetilde y_{ij}\widetilde y_{ik} \widetilde y^2_{it}\mid \by_i\} = \widehat p_i^{(3)}y_{ij}y_{ik} \sum_{t\notin \{j, k\}}  y^2_{it}, & \text{if $j\neq k$}.\end{cases}
\]
Writing $\by_{i,-t} := \by_i - y_{it}\be_t$, we then have
\begin{align*}
\mathbb{E}'(\bA_i^2 \mid \by_i) &= \widehat p_i^{(3)}\sum_{t=1}^d y_{it}^2\by_{i,-t}\by_{i,-t}^\top + (\widehat p_i^{(2)} - \widehat p_i^{(3)})\mathcal{D}\biggl(\sum_{t=1}^d y_{it}^2\by_{i,-t}\by_{i,-t}^\top\biggr) \\
&\preceq \widehat p_i^{(3)}\|\by_i\|_\infty^2 \sum_{t=1}^d\by_{i,-t}\by_{i,-t}^\top + (\widehat p_i^{(2)} - \widehat p_i^{(3)})\|\by_i\|_\infty^2 \mathcal{D}\biggl(\sum_{t=1}^d\by_{i,-t}\by_{i,-t}^\top\biggr). 
\end{align*}
Notice that 
\[
	\sum_{t = 1}^d \by_{i, -t}\by_{i, -t}^\top  = \sum_{t = 1}^d \bigl( \by_i\by_i^\top - y_{it}\be_{t}\by^\top_i - y_{it}\by_i\be_t^\top + y^2_{it}\be_t\be_t^\top \bigr) = (d - 2)\by_i\by_i^\top + \mathcal{D}(\by_i\by_i^\top). 
\]
Therefore, 
\[
	\begin{aligned}
		\mathbb{E}'(\bA_i^2 \mid \by_i) & \preceq \|\by_i\|_\infty^2\bigl\{\widehat p_i^{(3)}(d - 2)\by_i\by_i^\top + \bigl((d - 1)\widehat p_i^{(2)} - (d-2)\widehat p_i^{(3)}\bigr)\mathcal{D}(\by_i\by_i^\top)\bigr\} \\
		& \preceq  d\|\by_i\|_\infty^2\bigl\{\widehat p_i^{(3)}\by_i\by_i^\top + \widehat p_i^{(2)}\mathcal{D}(\by_i\by_i^\top)\bigr\}. 
	\end{aligned}
\]
Now, observe that $\|\bA_i\|_{\text{op}} \leq  d\widehat p_i\|\by_i\|_\infty^2$, so for $q\geq 2$,
\[
\mathbb{E}' \bigl(\bA_i^q\bigr) \preceq \mathbb{E}'\bigl\{(d\widehat p_i\|\by_i\|_\infty^2)^{q-2}\mathbb{E}'(\bA_i^2\mid \by_i)\bigr\}\preceq d^{q-1}\widehat p_i^{q-2}\mathbb{E}'\bigl[\|\by_i\|_\infty^{2q-2}\bigl\{\widehat p_i^{(3)}\by_i\by_i^\top + \widehat p_i^{(2)}\mathcal{D}(\by_i\by_i^\top)\bigr\}\bigr]. 
\]
By the Cauchy--Schwarz inequality, we therefore have that
\begin{align*}
\|\mathbb{E}' \bigl(\epsilon^q_i\bA^q_i \bigr)\|_{\text{op}} & \leq  d^{q-1}\widehat p_i^{q-2} \widehat p_i^{(3)} \Bigl[\mathbb{E}\bigl(\|\by_i\|_\infty^{4q-4} \bigr)  \sup_{\bv\in \mathcal{S}^{d-1}}\mathbb{E}\bigl\{(\bv^\top \by_i)^4\bigr\} \Bigr]^{1/2} + d^{q-1}\widehat p_i^{q-2}\widehat p_i^{(2)} \mathbb{E} \|\by_i\|_\infty^{2q}\\
&\leq  d^{q-1}\widehat p_i^{q-2}\Bigl\{\widehat p_i^{(3)} (4q-4)^{q-1}(CM\log d)^{q-1} 8R\tau^2 +  \widehat p_i^{(2)} (2q)^q (CM\log d)^q\Bigr\}\\
&\leq  \frac{q!}{2}\bigl\{32eCMR\tau^2\widehat p_i^{(3)} d\log d + e^2\widehat p_i^{(2)}d(CM\log d)^2\bigr\}\bigl(4eCM\widehat p_i d\log d\bigr)^{q-2} \\
& \leq \frac{q!}{2}C'Md \log d\bigl\{R\tau^2\widehat p_i^{(3)} + \widehat p_i^{(2)}M\log d\bigr\}\bigl(4eCM\widehat p_i d\log d\bigr)^{q-2}, 
\end{align*}
where $C'>0$ is a universal constant, the second inequality uses~\eqref{Eq:Psi2vec} and~\eqref{Eq:Psi2max} and the penultimate bound uses Stirling's inequality. 

%

Let $\rho:=4eCMd(\max_i \widehat p_i) \log d$ and $\sigma^2 := C'Mn^{-1}d\log d\sum_{i=1}^n \bigl\{R\tau^2\widehat p_i^{(3)}  + \widehat p_i^{(2)}M\log d\bigr\}$.  Then by \citet[Theorem~6.2]{Tro12}, we obtain that
\[
\mathbb{P}'\biggl(\biggl\|\frac{1}{n}\sum_{i=1}^n \epsilon_i\bA_i \biggr\|_\text{op}\geq t\biggr) \leq  2d \exp\biggl(\frac{-nt^2/2}{\sigma^2+\rho t}\biggr).
\]
Consequently, for $t_0:= 2\sigma n^{-1/2}\log^{1/2}d + 4\rho n^{-1}\log d$, we have 
\[
\mathbb{E}'\biggl\|\frac{1}{n}\sum_{i=1}^n \epsilon_i\bA_i \biggr\|_\text{op} \leq  t_0 + \int_{t_0}^\infty 2d\bigl\{e^{-nt^2/(4\sigma^2)} + e^{-nt/(4\rho)}\bigr\} \,dt \leq  4t_0.
\]
Given \eqref{eq:symmetrization}, integrating the left-hand side of the above inequality over $(\widehat p_i)_{i=1}^n$ yields
\begin{align}
   	\label{Eq:BoundA}
\E \biggl \| \frac{1}{np^2} \sum_{i = 1}^n (\bA_i - \E \bA_i) \biggr\|_{\text{op}} & \lesssim \frac{(\mathbb{E}\sigma^2)^{1/2}\log^{1/2}d}{n^{1/2}p^2} + \frac{\mathbb{E}\rho\log d}{np^2} \nonumber\\
& \lesssim \sqrt\frac{Md\{R\tau^2 p + M \log d\}\log^2 d}{np^2} + \frac{Md \log^2 d \log n}{np}, 
\end{align}
where the first inequality uses Jensen's inequality and 
the second inequality uses Lemma~\ref{lem:bernsteinoulli_moments}. 


For the second sum on the right-hand side of~\eqref{eq:davis_kahan}, we have by \citet[Lemma~2.2.2]{vanderVaartWellner1996} again that
\begin{align*}
\biggl\|\biggl\| \frac{1}{n}\sum_{i=1}^n(\bB_i-\mathbb{E}\bB_i)\biggr\|_{\text{op}} \biggr\|_{\psi_1} &= \biggl\|\max_{j\in[d]} \biggl|\frac{1}{n}\sum_{i=1}^n (\widetilde y_{ij}^2 - \E \widetilde y_{ij}^2)\biggr|\biggr\|_{\psi_1} \\
&\lesssim \frac{\log d}{n}\,\biggl\|\sum_{i=1}^n (\widetilde y_{i1}^2 - \E \widetilde y_{i1}^2)\biggr\|_{\psi_1} \lesssim \frac{M\log d}{\sqrt{n}},
\end{align*}
where the final inequality uses Lemma~\ref{lem:psi1bernstein} and the fact that $\|\widetilde y_{i1}^2 - \mathbb{E}\widetilde y_{i1}^2\|_{\psi_1} \leq \|\widetilde y_{i1}^2\|_{\psi_1} +\mathbb{E}\widetilde y_{i1}^2 \leq 2M$.
Now by the Cauchy--Schwarz inequality, 
\begin{align}
	\label{Eq:BoundB}
	\E \biggl\| \frac{\widehat p}{np^2} \sum_{i = 1}^n (\bB_i - \E \bB_i)\biggr\|_{\text{op}} & \le \biggl\{\E \biggl(\frac{\widehat p^2}{p^4}\biggr) \E \biggl(\biggl\|\frac{1}{n}\sum_{i = 1}^n (\bB_i - \E \bB_i) \biggr\|^2_{\text{op}}\biggr)\biggr\}^{1 / 2} \nonumber \\
	& \lesssim \biggl\{\biggl(\frac{1}{p^2} + \frac{1}{ndp^3}\biggr)\frac{M^2\log^2 d}{n}\biggr\}^{1 / 2} \lesssim \frac{M\log d}{p\sqrt{n}}, 
\end{align}
which is dominated by the bound in \eqref{Eq:BoundA}.

Finally, for the third term on the right-hand side of \eqref{eq:davis_kahan}, we have by the Cauchy--Schwarz inequality again that 
\begin{equation}
	\label{Eq:BoundC}
	\E \biggl\| \biggl(\frac{\widehat p}{p} - 1\biggr) \cD(\bSigma_{\by})\biggr\|_{\text{op}} \lesssim \frac{M}{\sqrt{ndp}}, 
\end{equation}
which is also dominated by the bound in \eqref{Eq:BoundA}. Substituting~\eqref{Eq:BoundA}, \eqref{Eq:BoundB} and \eqref{Eq:BoundC} into~\eqref{eq:davis_kahan} establishes~\eqref{eq:thm1}. 
If we regard $M$ and $\tau$ as constants and if $n\geq d\log^2d \log^2n / (\lambda_1 p + \log d)$, then the second term in the bracket of the right-hand side of~\eqref{eq:thm1} is dominated by the first term, and claim~\eqref{eq:thm1_minimax} follows immediately.
\end{proof}

\begin{proof}[Proof of Theorem \ref{thm:2}]				
Without loss of generality, we may assume that $d \geq 50$ and that $d$ is even, and write $d = 2h$ for some $h\in\mathbb{N}$. By the Gilbert--Varshamov lemma \citep[see, e.g.][Lemma~4.7]{Mas07}, there exist $W\subseteq \{0,1\}^h$ such that $\log |W|\geq h/16$ and $d_{\text{H}}(\bw,\bw')\geq h/4$ for any distinct pair of vectors $\bw,\bw'\in W$. Let $\gamma\in [0,\pi/2]$ be a real number to be specified later. To each $\bw\in W$, we can associate a distribution $P_{\bw}\in\mathcal{P}_{n,d}(\lambda_1,p)$ such that $\bU$ is a random vector ($n\times 1$ random matrix) with independent $N(0,\lambda_1)$ entries,  $\bZ$ is an $n\times d$ random matrix with independent $N(0,1)$ entries, and
\[
\bV_1 = \bV_{1,\bw} :=  \frac{1}{\sqrt{h}}\biggl\{\bw\otimes \begin{pmatrix}\cos \gamma\\ \sin\gamma\end{pmatrix} + (\mathbf{1}_h - \bw)\otimes \begin{pmatrix}\cos \gamma\\ -\sin\gamma\end{pmatrix} \biggr\}\in\mathcal{S}^{d-1}.
\]

Fixing distinct $\bw, \bw'\in W$, we write $\bv = (v_j)_{j\in [d]} := \bV_{1,\bw}$ and $\bv'=(v'_j)_{j\in [d]} := \bV_{1,\bw'}$ and let $Q_{\bw}$ and $Q_{\bw'}$ denote respectively the marginal distribution of $(\widetilde \by_1, \bomega_1)$ under $P_{\bw}$ and $P_{\bw'}$. Define $S:= \{j\in[d]: \omega_{1j} = 1\}$ and also set $\bar{\bv}_S := (v_j\mathds{1}_{\{j \in S\}})_{j \in [d]} \in \mathbb{R}^d$ and $\bar{\bv}_S' := (v_j'\mathds{1}_{\{j \in S\}})_{j \in [d]} \in \mathbb{R}^d$.  Then
\begin{align}
\text{KL}(P_{\bw}, P_{\bw'}) &= \text{KL}(Q_{\bw}^{\otimes n}, Q_{\bw'}^{\otimes n}) = n\text{KL}(Q_{\bw}, Q_{\bw'}) = n\mathbb{E}_{Q_{\bw}}\biggl\{\mathbb{E}_{Q_{\bw}}\biggl(\log\frac{dQ_{\bw}}{dQ_{\bw'}} \biggm| \bomega_1\biggr)\biggr\} \nonumber\\
&= n\mathbb{E}\,\text{KL}\bigl(N_d(\bzero,  \bI_d + \lambda_1 \bar{\bv}_S \bar{\bv}_S^\top), N_d(\bzero,  \bI_d + \lambda_1 \bar{\bv}'_S \bar{\bv}_S^{'\top})\bigr),\label{Eq:KLstep1}
\end{align}
where the final expectation is over the marginal distribution of $S$ under $P_{\bw}$. We partition $S=S_0\sqcup S_{1+}\sqcup S_{1-}$, where $S_0:=\{j\in S: \text{$j$ is odd}\}$, $S_{1+}:=\{j\in S: \text{$j$ is even and $v_j = \bar{v}'_j$}\}$ and $S_{1-}:=\{j\in S: \text{$j$ is even and $v_j \neq v'_j$}\}$. Since by construction we always have $\|\bar{\bv}_S\|_2=\|\bar{\bv}_S'\|_2$, we can apply Lemma~\ref{lem:7} to obtain
\begin{align*}
\text{KL}\bigl(N(\bzero,  \bI_d &+ \lambda_1 \bar{\bv}_S \bar{\bv}_S^\top), N(\bzero,  \bI_d + \lambda_1 \bar{\bv}'_S (\bar{\bv}'_S)^\top)\bigr) =\frac{\lambda_1^2(\|\bar{\bv}_S\|_2^4 - \langle\bar{\bv}_S, \bar{\bv}'_S\rangle^2)}{2(1+\lambda_1\|\bar{\bv}_S\|_2^2)}\\
&\leq  \frac{\lambda_1^2\langle \bar{\bv}_S, \bar{\bv}_S+\bar{\bv}'_S\rangle \langle \bar{\bv}_S, \bar{\bv}_S-\bar{\bv}'_S\rangle}{2\max\{1,\lambda_1\|\bar{\bv}_S\|_2^2\}} = \frac{\lambda_1^2\bigl(\sum_{j\in S_0\cup S_{1+}} 2v_j^2\bigr)\bigl(\sum_{j\in S_{1-}} 2v_j^2\bigr)}{2\max\{1,\lambda_1\sum_{j\in S}v_j^2\}}\\
&\leq  \min\biggl\{\frac{2\lambda_1^2}{h^2}\bigl(|S_0\times S_{1-}|\sin^2\gamma \cos^2\gamma+ |S_{1+}\times S_{1-}| \sin^4 \gamma\bigr), \, \frac{2\lambda_1|S_{1-}|\sin^2\gamma}{h}\biggr\}.
\end{align*}
Substituting the above bound into~\eqref{Eq:KLstep1}, we have
\begin{equation}
\label{Eq:Fano1}
\KL(P_{\bw}, P_{\bw'})\leq  2n\lambda_1 p \min\{1, \lambda_1 p\}\sin^2\gamma.
\end{equation}
On the other hand, since $d_{\text{H}}(\bw,\bw')\geq h/4$, we also have
\begin{equation}
\label{Eq:Fano2}
\sin^2 \Theta(\bv, \bv') = 1-(\bv^\top \bv')^2 = 1 - \biggl(1 - \frac{2 d_{\text{H}}(\bw,\bw')\sin^2\gamma}{h}\biggr)^2 \geq \frac{1}{2}\sin^2\gamma.
\end{equation}
By~\eqref{Eq:Fano1}, \eqref{Eq:Fano2} and Fano's inequality \citep[][Lemma~3]{Yub97}, 
\begin{align*}
\inf_{\widehat\bv}\sup_{P\in\mathcal{P}_{n,d,1}(\lambda_1,p)} \E_P L(\widehat\bv, \bv) &\geq \inf_{\widehat\bv }\max_{\bw \in W} \mathbb{E}_{P_{\bw}} L(\widehat\bv, \bv) \\
&\geq \frac{1}{2\sqrt{2}}\sin \gamma \biggl(1-\frac{\log 2+ 2n\lambda_1 p \min\{1,\lambda_1 p\}\sin^2\gamma}{\log |W|}\biggr). 
\end{align*}
We now choose $\gamma\in[0,\pi/2]$ such that $\sin^2\gamma = \min\bigl\{\frac{\log |W|}{8n\lambda_1 p \min\{1,\lambda_1 p\}}, 1\bigr\}$. Since $d\geq 50$, we obtain $\log|W|\geq d/32\geq 2\log 2$. Therefore, 
\[
\inf_{\widehat\bv}\sup_{P\in\mathcal{P}_{n,d,1}(\lambda_1,p)} \E_P L(\widehat\bv, \bv) \geq \frac{1}{8\sqrt{2}}\sin\gamma \geq \min\biggl\{\frac{1}{200\lambda_1}\sqrt\frac{d \max(1,\lambda_1 p)}{n p^2}, \frac{1}{8\sqrt{2}}\biggr\},
\]
as desired.
\end{proof}

\begin{proof}[Proof of Proposition~\ref{thm:4}]
For notational simplicity, we write $\widehat \bV_K := \widehat \bV_K^{(\mathrm{in})}$ and $\widehat \bV_{S,K} :=  (\widehat \bV_K)_S$ for any $S\subseteq [d]$. For $i \in \mathcal{I}$, let $\bell_i^\top \in \mathbb{R}^K$ denote the $i$th row of $\bL$.  For any $i \in \mathcal{I}$, we have $\widehat\by_{i, \cJ_i} = \by_{i, \cJ_i}$ and 
\begin{align*}
	\widehat \by_{i, \cJ_i^{\mathrm{c}}} - \by_{i, \cJ_i^{\mathrm{c}}}& = \widehat\bV_{\cJ_i^{\mathrm{c}}, K} (\widehat\bV_{\cJ_i, K}^{\top} \widehat\bV_{\cJ_i, K})^{-1} \widehat\bV_{\cJ_i, K}^{\top} \by_{i, \cJ_i} - \by_{i, \cJ_i^{\mathrm{c}}}\\
	& = \widehat \bV_{\cJ_i^{\mathrm{c}}, K}(\widehat\bV_{\cJ_i, K}^{\top} \widehat\bV_{\cJ_i, K})^{-1}\widehat\bV_{\cJ_i, K}^{\top}\bR_{\cJ_i}\bW_{\widehat \bV_K, \bR}\bW^{-1}_{\widehat \bV_K, \bR}\bGamma\bell_i - \bR_{\cJ^{\mathrm{c}}_i}\bGamma\bell_i\\
	& = \widehat \bV_{\cJ^{\mathrm{c}}_i, K}(\widehat\bV_{\cJ_i, K}^{\top} \widehat\bV_{\cJ_i, K})^{-1}\widehat \bV_{\cJ_i, K}^\top(\bR_{\cJ_i} \bW_{\widehat \bV_K, \bR} - \widehat \bV_{\cJ_i, K})\bW^{-1}_{\widehat \bV_K, \bR}\bGamma\bell_i \\
	&\hspace{6cm} + (\widehat \bV_{\cJ_i^{\mathrm{c}}, K} - \bR_{\cJ_i^{\mathrm{c}}}\bW_{\widehat \bV_K, \bR})\bW^{-1}_{\widehat \bV_K, \bR}\bGamma\bell_i.
\end{align*}
Thus 
	\begin{align*}
	\linfinity{\widehat \by_{i, \cJ_i^{\mathrm{c}}} - \by_{i, \cJ_i^{\mathrm{c}}}} & \leq \sigma_*\sqrt{d}\twotoinf{\widehat \bV_{\cJ_i^{\mathrm{c}}, K}} \twotoinf{\bR_{\cJ_i} \bW_{\widehat \bV_K, \bR} - \widehat \bV_{\cJ_i, K}} \ltwonorm{\bGamma\bell_i}\\
	&\hspace{6cm} + \twotoinf{\widehat \bV_{\cJ^{\mathrm{c}}_i, K}- \bR_{\cJ^{\mathrm{c}}_i} \bW_{\widehat \bV_K, \bR}} \ltwonorm{\bGamma\bell_i}  \\
	& \leq \Delta \ltwonorm{\bGamma\bell_i} \bigl(1 + \sigma_*\sqrt{d}\twotoinf{\widehat \bV_{K}}\bigr) \\
	& \leq \Delta \sigma_1(\bGamma) \mu_1\biggl({\frac{K}{n}}\biggr)^{1 / 2}\bigl\{1 + \sigma_*\bigl(\mu_2\sqrt{K} + \Delta\sqrt d\bigr)\bigr\} \\
                &\leq \frac{C'}{n^{1/2}}\Delta \sigma_1(\bGamma) \mu_1\mu_2K =: m, 
	\end{align*}	
say, where $C' > 0$ depends only on $\sigma_*$ and $c_1$.  Note that the inequality above holds for all $i \in \mathcal{I}$.  Writing $\bE := \widehat \bY - \bY$ for convenience, we have found that $\|\bE\|_\infty \leq m$.  Let $\bL_{\perp} \in \OO^{n \times (n-K)}, \bR_{\perp} \in \OO^{d \times (d - K)}$ be the orthogonal complements of $\bL \in \OO^{n \times K}$ and $\bR \in \OO^{d \times K}$ respectively, so that $(\bL,\bL_{\perp}) \in \mathbb{O}^{n \times n}$ and $(\bR,\bR_{\perp}) \in \mathbb{O}^{d \times d}$.  We wish to apply \citet[Theorem~1]{CZh18}.  To this end, note that
\[
\opnorm{\bL^\top \bE\bR} = \sup_{\bs, \bt \in \cS^{K - 1}} (\bL\bs)^\top \bE (\bR\bt) \le \twotoinf{\bL}\twotoinf{\bR}\|\bE\|_1 \le \frac{K \mu_1\mu_2m\lonenorm{\bOmega^{\mathrm{c}}}}{\sqrt{nd}}. 
\]
Hence, writing $\alpha := \sigma_K(\bGamma + \bL^\top \bE \bR) $, we have by Weyl's inequality that
\[
	\sigma_K(\bGamma) - \frac{K \mu_1\mu_2m\lonenorm{\bOmega^{\mathrm{c}}}}{\sqrt{nd}} \le \alpha \le \sigma_K(\bGamma) + \frac{K \mu_1\mu_2m\lonenorm{\bOmega^{\mathrm{c}}}}{\sqrt{nd}}. 
\]
	Now, writing $\beta := \opnorm{\bL_{\perp}^\top \widehat \bY \bR_{\perp}} = \opnorm{\bL^\top_{\perp}\bE\bR_{\perp}}$, we have 
\[
	\beta \leq \opnorm{\bE} \leq \fnorm{\bE} \leq m\sqrt{\lonenorm{\bOmega^{\mathrm{c}}}}. 
\]
In addition, by Cauchy--Schwarz and Jensen's inequality,
\begin{align*}
	\opnorm{\bL^\top \bE} & = \sup_{\substack{\bs \in \cS^{K - 1}\\\bt \in \cS^{d - 1}}} (\bL\bs)^\top \bE\bt \le \twotoinf{\bL} \sup_{\bt \in \cS^{K - 1}}\lonenorm{\bE\bt}  \\
	& \le \mu_1(Kn)^{1 / 2} \frac{1}{n} \sum_{i = 1}^{n} m\sqrt{\lonenorm{\bomega^{\mathrm{c}}_i}} \le \mu_1m (K \lonenorm{\bOmega^{\mathrm{c}}})^{1 / 2}.
\end{align*}
Similarly, 
\[
	\opnorm{\bE\bR} \le \mu_2m (K \lonenorm{\bOmega^{\mathrm{c}}})^{1 / 2}.
\]
Hence there exists $c_1 > 0$, depending only on $\mu_1,\mu_2$ and $\sigma_*$, such that whenever $\Delta \leq \frac{c_1\sigma_K(\bGamma)}{K ^ 2\sigma_1(\bGamma) \sqrt{d}}$, we have   
\[
	\alpha^2 - \beta^2 - \min(\opnorm{\bL^\top \bE}^2, \opnorm{\bE \bR}^2) \ge \frac{\sigma^2_K(\bGamma)}{2}~~~\text{and}~~~\alpha, \beta \le 2\sigma_K(\bGamma). 
\]
 Now let $\widehat \bY_{\mathcal{I}} = \widehat \bL \widehat \bGamma \widehat \bR^\top$ be an SVD of $\widehat \bY$. We can now apply \citet[Theorem~1]{CZh18} to deduce that for such $c_2$,
\begin{align*}
	\opnorm{\sin\Theta(\widehat\bR, \bR)} &\leq \frac{\alpha \opnorm{\bL^\top \bE} + \beta \opnorm{\bE\bR}}{\alpha^2 - \beta^2 - \min(\opnorm{\bL^\top \bE}^2, \opnorm{\bE \bR}^2)} \leq \frac{4m(\mu_1 + \mu_2)(K\lonenorm{\bOmega^{\mathrm{c}}})^{1 / 2}}{\sigma_K(\bGamma)}\\
	& \leq 	\frac{4C'K^{3 / 2} \sigma_1(\bGamma)(\mu_1 + \mu_2)\mu_1\mu_2}{\sigma_K(\bGamma)}\biggl(\frac{\lonenorm{\bOmega^{\mathrm{c}}}}{n}\biggr)^{1 / 2}\Delta =: \kappa \Delta, 
\end{align*}
say. Similarly, 
\[
	\opnorm{\sin\Theta(\widehat \bL, \bL)} \leq \frac{\alpha \opnorm{\bE\bR} + \beta \opnorm{\bL^\top \bE}}{\alpha^2 - \beta^2 - \min(\opnorm{\bL^\top \bE}^2, \opnorm{\bE \bR}^2)} \le \kappa \Delta. 
\]
We are now in a position to show contraction in terms of two-to-infinity norm. By \citet[Theorem~3.7]{CTP18}, 
\begin{align}
\label{eq:40}
\cT(\widehat \bR, \bR) & \leq \frac{2\twotoinf{\bR_{\perp}\bR_{\perp}^{\top}\bE^\top \bL\bL^\top}}{\sigma_K(\bGamma)} +  \frac{ 2\twotoinf{\bR_{\perp}\bR_{\perp}^\top\bE^\top \bL_{\perp}\bL^\top_{\perp}} }{\sigma_K(\bGamma)}  \opnorm{\sin\Theta(\widehat \bL, \bL)}\nonumber\\
&\hspace{4cm} + \opnorm{\sin\Theta(\widehat\bR, \bR)}^2\twotoinf{\bR} =: T_1 + T_2 + T_3, 
\end{align}
say.  Note that 	
\begin{align*}
	\|\bR_{\perp} \bR^\top_{\perp}\|_{\infty \rightarrow \infty} & \le \|\bI_d\|_{\infty \rightarrow \infty} + \|\bR\bR^\top\|_{\infty \to \infty} = 1 + \sup_{\|\bv\|_{\infty} \le 1} \|\bR\bR^\top \bv\|_{\infty}\\
	& \le 1 + \sup_{\|\bv\|_{2} \le \sqrt{d}} \|\bR\|_{2 \to \infty}\|\bR^\top \bv\|_2 \le 1 + \sqrt K\mu_2. 
\end{align*}
Hence, 
\begin{align*}
T_1 & \leq \frac{2(1 + \sqrt K\mu_2)\twotoinf{\bE^\top \bL\bL^\top}}{\sigma_K(\bGamma)} \leq \frac{2(1 + \sqrt K\mu_2)\twotoinf{\bE^\top\bL}}{\sigma_K(\bGamma)} \\
& \leq \frac{2(1 + \sqrt K \mu_2)\mu_1\sqrt{K} m \|{\bOmega^{\mathrm{c}}}\|_{1 \to 1}}{\sqrt{n} \sigma_K(\bGamma)}  {\lesssim}_{\mu_1, \mu_2} \frac{K^{2} \sigma_1(\bGamma)\|\bOmega^{\mathrm{c}}\|_{1 \to 1} \Delta}{n\sigma_K(\bGamma)}. 
\end{align*}
Moreover, 
\begin{align*}
T_2 & \leq \frac{2(1 + \sqrt K\mu_2)\twotoinf{\bE^\top} \kappa\Delta}{\sigma_K(\bGamma)} \leq \frac{2(1 + \sqrt K\mu_2) m \|\bOmega^{\mathrm{c}}\|_{1 \to 1}^{1 / 2} \kappa \Delta}{\sigma_K(\bGamma)} \\
& \lesssim_{\mu_1, \mu_2} \frac{K^{3 / 2}\sigma_1(\bGamma)\|\bOmega^{\mathrm{c}}\|_{1 \to 1} ^ {1 / 2} \kappa \Delta^2 }{\sqrt{n}\sigma_K(\bGamma)}. 
\end{align*}
Finally, 
\[
	T_3 \le \mu_2\kappa^2\Delta^2\biggl(\frac{K}{d}\biggr)^{1 / 2}. 
\]
Write
\[
	\eta := \frac{K^2\sigma_1(\bGamma)\|\bOmega^{\mathrm{c}}\|_{1 \to 1}^{1 / 2}}{\sqrt{n}\sigma_K(\bGamma)}
\]
for simplicity, so that $\kappa \lesssim_{\mu_1, \mu_2, \sigma_*} (d / K)^{1 / 2}\eta$. Given that $\cT(\widehat\bV_K^{(\mathrm{out})}, \bV_K) = \cT(\widehat \bR, \bR)$, substituting the bounds for $T_1, T_2, T_3$ into \eqref{eq:40} yields that 
\begin{align*}
\cT (\widehat\bV_K^{(\mathrm{out})}, \bV_K)  & \lesssim_{\mu_1, \mu_2, \sigma_*} \biggl\{ \eta\biggl(\frac{\|\bOmega^{\mathrm{c}}\|_{1 \to 1}}{n}\biggr)^{1 / 2} + \frac{\sqrt{d}}{K}\eta^2\Delta + \biggl(\frac{d}{K}\biggr)^{1/2}\eta^2\Delta \biggr\} \Delta\\
& \leq \eta^2 \biggl\{ \frac{\sigma_K(\bGamma)}{K^{2} \sigma_1(\bGamma)} + 2 \biggl(\frac{d}{K}\biggr)^{1 / 2} \Delta\biggr\}\Delta \lesssim_{\mu_1, \mu_2}  \frac{\eta^2 \sigma_K(\bGamma)}{K^{2 }\sigma_1(\bGamma) } \Delta = \frac{K^{2} \sigma_1(\bGamma) \|\bOmega^{\mathrm{c}}\|_{1 \to 1}}{\sigma_K(\bGamma) n}\Delta, 
\end{align*}
as desired. 
	\end{proof}
	
	\begin{proof}[Proof of Theorem~\ref{thm:primepca}]
		We prove this result by induction on $t$. The case $t = 0$ is true by definition of $\Delta$, so
		suppose that the conclusion holds for some $t\in \{0\} \cup [n_{\mathrm{iter}}-1]$. We make the following two claims:  
		\begin{enumerate}
			\item[(a)] $\mathcal{I}^{(t)} = \mathcal{I}$; 
			\item[(b)] The error is further contracted by refinement, i.e., $\cT(\widehat\bV^{(t + 1)}_K, \bV_K) \le \rho\cT(\widehat\bV^{(t)}_K, \bV_K)$. 
		\end{enumerate}
		To prove claim (a), notice that for each $i \in [n]$, by Weyl's inequality and the inductive hypothesis,
		\[
			\begin{aligned}
			\bigl|\sigma_K\bigl((\widehat \bV_K^{(t)})_{\cJ_i}\bigr) - \sigma_K((\bV_K)_{\cJ_i})\bigr| & = \bigl|\sigma_K\bigl((\widehat \bV_K^{(t)})_{\cJ_i}\bigr) - \sigma_K\bigl((\bV_K)_{\cJ_i} \bW_{\widehat\bV^{(t)}_{K}, \bV_K}\bigr)\bigr| \\
			& \leq \bigl\|(\widehat \bV_K^{(t)})_{\cJ_i} - (\bV_K)_{\cJ_i}\bW_{\widehat\bV^{(t)}_{K}, \bV_K}\bigr\|_{\mathrm{op}} \\
			& \leq |\cJ_i|^{1/2} \mathcal{T}\bigl(\widehat \bV_K^{(t)},\bV_K\bigr) \leq |\cJ_i|^{1/2}\rho^t \Delta. 
			\end{aligned}
		\]
		Now, for $i \in \cI$, 
		\[
			\begin{aligned}
			\sigma_K\bigl((\widehat \bV_K^{(t)})_{\cJ_i}\bigr) & \geq \sigma_K\bigl((\bV_K)_{\cJ_i}\bigr) - |\sigma_K\bigl((\widehat \bV_K^{(t)})_{\cJ_i}\bigr) - \sigma_K\bigl((\bV_K)_{\cJ_i}\bigr)| \\
			& \geq \bigl(\sigma_*^{-1} + \epsilon - \sqrt d \Delta\bigr)(|\cJ_i| / d)^{1 / 2}.
			\end{aligned}
		\]
		On the other hand, if $i \in \cI^{\mathrm{c}}$ and $\|\bomega_i\|_1 > K$, then
		\[
			\begin{aligned}
	    			\sigma_K\bigl((\widehat \bV_K^{(t)})_{\cJ_i}\bigr) & \le \sigma_K\bigl((\bV_K)_{\cJ_i}\bigr) + |\sigma_K\bigl((\widehat \bV_K^{(t)})_{\cJ_i}\bigr) - \sigma_K\bigl((\bV_K)_{\cJ_i}\bigr)| \\
    				& \le \bigl(\sigma_*^{-1} - \epsilon + \sqrt d \Delta\bigr)(|\cJ_i| / d)^{1 / 2}.
			\end{aligned}
		\]		
		Hence, if we choose $c_1 \le \epsilon$, then $\sqrt d\Delta < \epsilon$, so for $i \in \cI$, 
		\[
			\sigma_K\bigl((\widehat \bV_K^{(t)})_{\cJ_i}\bigr) > \Bigl(\frac{|\cJ_i|}{d\sigma_*}\Bigr)^{1 / 2};
		\]
		moreover, for $i \in \cI^{\mathrm{c}}$, 
		\[
			\sigma_K\bigl((\widehat \bV_K^{(t)})_{\cJ_i}\bigr) < \Bigl(\frac{|\cJ_i|}{d\sigma_*}\Bigr)^{1 / 2}. 
		\]
		Claim~(a) follows. As for claim~(b), note that $\widehat\bV_K^{(t + 1)} = \texttt{refine}(K, \widehat\bV_K^{(t)}, \bOmega_{\cI^{(t)}}, (\bY_{\bOmega})_{\cI^{(t)}})$.  Taking $c_1,C > 0$ from Proposition~\ref{thm:4}, and reducing $c_1$ if necessary so that $c_1 \le \epsilon$, we may apply this proposition to deduce that whenever
		\begin{enumerate}
			\item[(i)] $\cT(\widehat\bV^{(t)}_K, \bV_K) \le \frac{c_1 \sigma_K(\bY_{\cI})}{K^2 \sigma_1(\bY_{\cI})\sqrt{d}}$; 
			\item[(ii)] $\rho := \frac{CK^{2} \sigma_1(\bY_{\cI}) \|{\bOmega_{\mathcal{I}}^{\mathrm{c}}}\|_{1 \to 1}}{\sigma_K(\bY_{\cI}) |\cI|} < 1$, 
                        \end{enumerate}
                        we have $\cT(\widehat\bV_K^{(t + 1)}, \bV_K) \le \rho \cT(\widehat\bV_K^{(t)}, \bV_K)$.  But the conditions~(i) and~(ii) are ensured by the inductive hypothesis and our assumptions, so the conclusion follows. 
\end{proof}

\begin{proof}[Proof of Theorem \ref{thm:init_2toinf}]
Let $\bE := \widetilde \bG - \E^{\bOmega} \widetilde \bG = \widetilde{\bG} - \bSigma_{\by}$. By \citet[Theorem~3.7]{CTP18}, when $\lambda_K \geq 2\opnorm{\bE}$, we have that 
\begin{align*}
	\cT(\widetilde \bV_K, \bV_K) & \le 2\lambda_K^{-1}\twotoinf{\bV_{-K} \bV_{-K}^\top\bE \bV_K\bV_K^\top} \\
	& \qquad\qquad + 2\lambda_K^{-1}\twotoinf{\bV_{-K} \bV^{\top}_{-K} \bE \bV_{-K} \bV_{-K}^\top} \opnorm{\sin\Theta(\widetilde \bV_K, \bV_K)}\\
	& \qquad\qquad + 2\lambda_K^{-1}\twotoinf{\bV_{-K} \bV^{\top}_{-K} \bSigma_{\by} \bV_{-K} \bV_{-K}^\top} \opnorm{\sin\Theta(\widetilde \bV_K, \bV_K)} \\ 			
	& \qquad\qquad + \opnorm{\sin\Theta(\widetilde \bV_K, \bV_K)}^2 \twotoinf{\bV_K} \\
	& =: T_1 + T_2 + T_3 + T_4. 
\end{align*}
Note that if $\lambda_K$ satisfies \eqref{eq:lambda_K_lowerbound} for some $\xi > 1$, then $\mathbb{P}^\bOmega(\opnorm{\bE} \geq \lambda_K/4) \leq 4d^{-(\xi-1)}$. In fact, since $\opnorm{\widetilde \bW} \le \fnorm{\widetilde \bW}$, there exists $c_{M, \tau, \mu} > 0$ such that \eqref{eq:lambda_k} implies \eqref{eq:lambda_K_lowerbound}, which, together with~\eqref{eq:E_op} ensures that $\mathbb{P}^\bOmega(\opnorm{\bE} \geq \lambda_K/2) \leq 4d^{-(\xi-1)}$.  
		
To bound $T_1$, we have
\begin{align}
 \label{Eq:Kmusquared}
\twotoinf{\bV_{-K}\bV_{-K}^\top \bE \bV_K\bV_K^\top} & \le \inftoinf{\bV_{-K} \bV^\top_{-K}} \twotoinf{\bE \bV_{K}\bV^\top_{K}} \nonumber \\
& \le (1 + K \mu^2) \max_{j \in [d]}  \sup_{\bu \in \cS^{K-1}} \be_j^\top \bE\bV_K\bu, 
\end{align}
where the second inequality is due to the fact that 
\[
	\inftoinf{\bV_{-K}\bV_{-K}^\top} \le \inftoinf{\bI_d} + \inftoinf{\bV_K\bV_K^\top} \le 1 + K\mu^2. 
\]
We use a covering argument to bound the supremum term. Let $\cN_K(1 / 2)$ be a $1 / 2$-net of the Euclidean sphere $\cS^{K - 1}$, i.e., for any $\bu \in \cS^{K - 1}$, there exists a point $\pi(\bu) \in \cN_K(1 / 2)$ such that $\ltwonorm{\bu - \pi(\bu)} \le 1 / 2$. Note that for any $\bu \in \cS^{K - 1}$,
\[
	\be_j^\top \bE\bV_K\bu =  \be_j^\top \bE\bV_K\pi(\bu) + \be_j^\top \bE\bV_K(\bu - \pi(\bu))\le \max_{\bv \in \cN_K(1 / 2)} \be_j^\top \bE\bV_K\bv + \frac{1}{2} \sup_{\bv \in \mathcal{S}^{K-1}} \be_j^\top \bE\bV_K\bv, 
\]
which further implies that 
\begin{equation}
	\label{eq:covering}
	\sup_{\bu \in \mathcal{S}^{K-1}} \be_j^\top \bE\bV_K\bu \le 2 \max_{\bu \in \cN_K(1 / 2)} \be_j^\top \bE\bV_K\bu. 
\end{equation}
We then argue similarly as in~\eqref{eq:28}, with $\bV_K\bu$ taking the role of $\bv_k$ there (since $\|\bV_K\bu\|_{\infty} \le \mu( K / d)^{1 / 2}$, we correspondingly have $\sqrt{K}\mu$ taking the role of the incoherence parameter $\mu$ there), to obtain that for any $\xi > 0$, 
\[
		\PP^{\bOmega}\biggl\{ \left |\be_j^\top \bE \bV_K\bu \right | \ge 2e\tau\mu\biggl(\frac{KMR}{d}\biggr)^{1 / 2}\biggl( \frac{\xi^{1/2}\lonenorm{\widetilde{\bW}_j}^{1/2} }{n^{1/2}} + \frac{\xi \ltwonorm{\widetilde{\bW}_j} }{n}\biggr )  \biggr\} \le 2e^{-\xi}.
\]
By \citet[Lemma~5.2]{Ver10}, $|\cN_K(1/2)| \le 5^K$.  Hence, by~\eqref{Eq:Kmusquared},~\eqref{eq:covering} and a union bound, we have for any $\xi>\log 5$ that
\begin{align*}
		\PP^{\bOmega}\biggl\{ T_1 \ge \frac{8\tau\mu(1 + K\mu^2)}{\lambda_K}\biggl(\frac{KMR}{d}\biggr)^{1 / 2} \biggl(\frac{\xi^{1/2}\inftoinf{\widetilde \bW}^{1/2}  }{n^{1/2}} + \frac{\xi \twotoinf{\widetilde \bW} }{n}\biggr) \biggr\} \le 2d e^{K\log 5- \xi}. 
\end{align*}
Next we bound $T_2$. Note that 
\begin{align*}
		\twotoinf{\bV_{-K}\bV_{-K}^\top \bE\bV_{-K}\bV_{-K}^\top} \le \inftoinf{\bV_{-K} \bV^{\top}_{-K}} \twotoinf{\bE} \le (1 + K\mu^2)\twotoinf{\bE}. 
\end{align*}	
For $j, k \in [d]$, let $\cI_{jk} := \{i: \omega_{ij}\omega_{ik} = 1\}$ and $n_{jk} := |\cI_{jk}| = n/\widetilde{W}_{jk}$.  Then 
\[
	E_{jk} = \frac{1}{n}\sum_{i = 1}^n \widetilde y_{ij}\widetilde y_{ik} \widetilde W_{jk} - [\E^{\bOmega} \widetilde \bG]_{jk} = \frac{1}{n_{jk}} \sum_{i \in \cI_{jk}} y_{ij} y_{ik} - [\E^{\bOmega} \widetilde \bG]_{jk}. 
\]
By applying both parts of Lemma \ref{lem:psi1_to_psi2}, for any $i \in [n]$ and $j,k \in [d]$, we have $\|y_{ij}y_{ik}\|_{\psi_1} \le 2 \|y_{ij}\|_{\psi_2} \|y_{ik}\|_{\psi_2} \le 2M$. Applying Bernstein's inequality \citep[Theorem~2.10]{BLM13} yields that for any $\xi > 0$, 
\[
	\PP^{\bOmega}\biggl\{ |E_{jk}| \ge 2eM\biggl(\biggl(\frac{2\xi\widetilde W_{jk}}{n}\biggr)^{1 / 2} + \frac{\xi \widetilde W_{jk}}{n}\biggr) \biggr\} \le 2e^{-\xi}. 
\]
Therefore, a union bound with $(j, k) \in [d] \times [d]$ yields that 
\begin{align*}
	\PP^{\bOmega}\biggl\{ T_2 \ge \frac{4\sqrt 2 eM(1 + K\mu^2)}{\lambda_K} \biggl( \biggl(\frac{2\xi\inftoinf{\widetilde \bW} }{n}\biggr)^{1 / 2} + \frac{\xi \twotoinf{\widetilde \bW}}{n}\biggr) \opnorm{\sin\Theta(\widetilde \bV_K, & \bV_K)}\biggr\} \\
	& \le 2d^2e^{-\xi}. 
\end{align*}
Now we bound $T_3$. We have that 
\[
	T_3 = \frac{2\|\bV_{-K} \bV_{-K}^\top\|_{2 \to \infty}}{\lambda_K} \opnorm{\sin\Theta(\widetilde \bV_K, \bV_K)} \le \frac{2\bigl\{1 + \mu(K / d)^{1  / 2}\bigr\}}{\lambda_K}\opnorm{\sin\Theta(\widetilde \bV_K, \bV_K)}. 
\]
Finally, $T_4$ satisfies  
\[
	T_4 \le \frac{\mu K^{1/2}}{d^{1/2}} \opnorm{\sin\Theta(\widetilde \bV_K, \bV_K)}^2. 
\]
Since $\opnorm{\sin\Theta(\widetilde \bV_K, \bV_K)} \leq \min\bigl\{L(\widetilde \bV_K, \bV_K),1\bigr\}$, combining our bounds for $\{T_j\}_{j = 1}^4$ yields that there exists $C_{M, \tau, \mu} > 0$ such that  for any $\xi > 2$,
\begin{align*}
\PP^{\bOmega}  \biggl\{ \cT(\widetilde \bV_K, \!\bV_K)  &\geq \frac{KC_{M, \tau, \mu}}{\lambda_K} \biggl \{L(\widetilde \bV_K, \bV_K) \!+ \!\biggl(\frac{KR}{d}\biggr)^{\hspace{-.15cm}1 / 2} \biggr\}  \biggl( \frac{\xi^{1/2} \inftoinf{\widetilde \bW}^{1/2}}{n^{1/2}} + \frac{\xi \twotoinf{\widetilde \bW}}{n}\biggr)\nonumber\\
&+ \mu \biggl( \frac{K^{1 / 2}}{d^{1 / 2}} + \frac{4}{\lambda_K} \biggr)L(\widetilde \bV_K, \bV_K) \biggr\} \leq 2de^{K\log 5 - \xi} + 2d^2e^{- \xi} + 4d^{- (\xi - 1)}.
\end{align*}
It therefore follows from Proposition~\ref{prop:init_f}, which applies because condition~\eqref{eq:lambda_k} for a suitable $c_{M,\tau,\mu}$ implies~\eqref{eq:lambda_k_f}, together with the facts that $\|\widetilde{\bW}\|_1 \leq d\|\widetilde{\bW}\|_{\infty\rightarrow \infty}$ and $\|\widetilde{\bW}\|_{\mathrm{F}} \leq d^{1/2}\|\widetilde{\bW}\|_{2\rightarrow \infty}$ that the first conclusion of the theorem holds. The second conclusion then follows immediately. 
\end{proof}
%

\begin{proof}[Proof of Proposition \ref{prop:init_f}]
In this proof, we use the shorthand $\bD_{\bu} := \mathrm{diag}(\bu)$ for $u \in \mathbb{R}^d$.  We represent $\widetilde \bG$ under the orthonormal basis $(\bV_K, \bV_{-K})$ as follows: 
\[
	\widetilde \bG = (\bV_K, \bV_{-K}) \biggl(
		\begin{array}{cc}
			\bV_K^{\top} \widetilde \bG\bV_K & \bV_{K}^\top \widetilde \bG \bV_{-K} \\
			\bV_{-K}^{\top} \widetilde \bG \bV_K & \bV_{-K}^{\top} \widetilde \bG \bV_{-K}
		\end{array} \biggr)
		\biggl(
		\begin{array}{c}
	        		\bV_K^{\top} \\
			\bV_{-K}^{\top}
		\end{array}
		\biggr). 
\]
Define 
\[
	\bG^{*} := (\bV_K, \bV_{-K}) \biggl(
		\begin{array}{cc}
			\bV_K^{\top}\widetilde \bG\bV_K & \bzero \\
			\bzero & \bV_{-K}^{\top} \widetilde \bG \bV_{-K}
		\end{array} \biggr)
		\biggl(
		\begin{array}{c}
	        		\bV_K^{\top} \\
			\bV_{-K}^{\top}
		\end{array}
		\biggr). 
\]
In the sequel, we regard $\widetilde \bG$ as a corrupted version of $\bG^*$ with the off-diagonal blocks $\bV_K^\top\widetilde \bG\bV_{-K}$ and $\bV^\top_{-K}\widetilde \bG\bV_K$ as perturbations. We have
\[
	\fnorm{\bV_K^\top \widetilde \bG\bV_{-K}} = \fnorm{\bV_K^\top (\widetilde \bG - \bSigma_{\by}) \bV_{-K}} \leq \fnorm{\bV_K^\top(\widetilde \bG - \E^{\bOmega} \widetilde\bG)} 
\]
We control the right-hand side through a concentration inequality, and for $k \in [K]$ let $\bv_k$ denote the $k$th column of $\bV_K$. For any $j \in [d]$ and $k \in [K]$, 
\begin{align}
\label{Eq:T1Conc}
\bv_k^\top (\widetilde \bG - \E^{\bOmega} \widetilde\bG)\be_j & = \frac{1}{n}\sum_{i = 1}^n \bv_k^\top \bigl\{\widetilde\by_i \widetilde \by_i^\top \circ \widetilde \bW - \E^{\bOmega}\bigl(\widetilde\by_i \widetilde \by_i^\top \circ \widetilde \bW\bigr)\bigr \}\be_j   \nonumber\\
& = \frac{1}{n}\sum_{i = 1}^n \bigl\{\widetilde \by_i^\top \bD_{\bv_k} \widetilde \bW \bD_{\be_j} \widetilde \by_i -  \E^{\bOmega} \bigl(\widetilde \by_i^\top \bD_{\bv_k} \widetilde \bW \bD_{\be_j} \widetilde \by_i \bigr)\bigr\} \nonumber\\
& = \frac{1}{n} \sum_{i = 1}^n \bigl\{\widetilde y_{ij}\widetilde \by_i^\top \bD_{\bv_k} \widetilde \bW_j - \E^{\bOmega} \bigl(\widetilde y_{ij}\widetilde \by_i^\top \bD_{\bv_k} \widetilde \bW_j \bigr) \bigr\}, 
\end{align}
where $\widetilde \bW_j$ denotes the $j$th column of $\widetilde \bW$. 
Note that 
\[
	\|\by_i\|_{\psi^*_2} \le \sup_{\bv \in \cS^{d - 1}}\frac{\|\bv^\top \bV_K\bu_i\|_{\psi_2} + \|\bv^\top \bz_i\|_{\psi_2}}{\sqrt{\bv^\top \bV_K\bSigma_{\bu}\bV_K^\top\bv + 1}} \le 2\tau. 
\]
Thus for any vector $\ba \in \RR^d$, we have by Lemma~\ref{lem:psi1_to_psi2} that 
\[
	\|y_{ij} (\ba^\top \by_i)\|_{\psi_1} \leq 2\|y_{ij}\|_{\psi_2} \|\ba^\top \by_i\|_{\psi_2} \leq 4 \tau (M\ba^\top \bSigma_{\by}\ba)^{1 / 2}. 
\]
For $i \in [n]$, let $\ba_i := \omega_{ij} \widetilde \bW_j \circ \bv_k \circ \bomega_i$. Now for any $q \geq 2$, 
\begin{align*}
	\E^{\bOmega}  \bigl|\widetilde y_{ij}(\widetilde \bW_j^\top \bD_{\bv_k}\widetilde \by_i )\bigr|^q &= \E^{\bOmega} | y_{ij}\ba_i^\top \by_i |^q \le  \Bigl( 4q \tau \sqrt{M \ba_i^\top \bSigma_{\by} \ba_i} \Bigr )^q \\
	& \leq \frac{16q^q\tau^2\mu^2{MR}}{d} \Bigl(4\tau \mu \sqrt{MR\ltwonorm{\widetilde \bW_j}^2 / d}\Bigr)^{q - 2} \sum_{t = 1}^d \widetilde W^2_{tj} \omega_{it}\omega_{ij} \\
	& \leq  \frac{8e^2q!\tau^2\mu^2 MR}{d} \Bigl(4e\tau \mu \sqrt{MR\ltwonorm{\widetilde \bW_j}^2 / d} \Bigr)^{q - 2} \sum_{t = 1}^d \widetilde W^2_{tj} \omega_{it}\omega_{ij}, 
\end{align*}
where the penultimate inequality uses the fact that $\ltwonorm{\ba_i}^2 \le \mu^2 \ltwonorm{\widetilde \bW_j}^2 / d$, and the last inequality is due to Stirling's approximation.  
Hence,
\begin{align*}
\frac{1}{n}\sum_{i = 1}^n \E^{\bOmega} |\widetilde y_{ij}\widetilde \bW_j^\top \bD_{\bv_k}\widetilde \by_i|^q & \le  \frac{8e^2q!\tau^2\mu^2 MR}{d} \Bigl(4e\tau \mu \sqrt{MR\ltwonorm{\widetilde \bW_j}^2 / d} \Bigr)^{q - 2} \sum_{t = 1}^d \sum_{i = 1}^n \frac{\widetilde W^2_{jt} \omega_{it}\omega_{ij}}{n} \\
& = \frac{8e^2q!\tau^2\mu^2 MR}{d} \Bigl(4e\tau \mu \sqrt{MR\ltwonorm{\widetilde \bW_j}^2 / d} \Bigr)^{q - 2} \lonenorm{\widetilde \bW_j}. 
\end{align*}
Thus by~\eqref{Eq:T1Conc} and Bernstein's inequality \cite[Theorem~2.10]{BLM13}, we have that for any $\xi > 0$, 
\begin{equation}
	\label{eq:28}
	\PP^{\bOmega}\biggl \{ \bigl|\bv_k^\top (\widetilde \bG - \E^{\bOmega} \widetilde\bG)\be_j\bigr|  \geq 2^{5/2}e\tau\mu\biggl(\frac{MR}{d}\biggr)^{1 / 2}\biggl (\biggl( \frac{\xi\lonenorm{\widetilde \bW_j} }{n}\biggr)^{1 / 2} + \frac{\xi \ltwonorm{\widetilde \bW_j} }{n}\biggr )  \biggr\} \leq 2e^{-\xi}. 
\end{equation}
By a union bound over $(j, k) \in [d] \times [K]$, for any $\xi > 1$, 
\begin{align}
\label{eq:perturb}
   	\PP^{\bOmega}\biggl\{ \fnorm{\bV_K^\top \widetilde \bG\bV_{-K}} \geq 8e\tau\mu\biggl(\frac{KMR}{d}\biggr)^{1 / 2}\biggl(\frac{\xi^{1/2} \|\widetilde \bW\|_1^{1/2}\log^{1/2} d  }{n^{1/2}}& +  \frac{\xi\fnorm{\widetilde \bW}\log d }{n}\biggr) \biggr\}\nonumber\\
       & \leq 2Kd^{- (\xi - 1)}. 
\end{align}
%
		
Now we provide a condition under which $\lambda_{\min} (\bV_K^\top \widetilde\bG \bV_K) > \opnorm{\bV_{-K}^\top \widetilde \bG \bV_{-K}}$, which ensures that $\bV_K$ is the top $K$ eigenspace of $\bG^*$. Note that 
\[
    \lambda_{\min}(\bV_K^\top \widetilde \bG \bV_K) \geq \lambda_K + 1 - \opnorm{\bV_K^\top (\widetilde\bG - \bSigma_{\by})\bV_K} \geq \lambda_K + 1 - \opnorm{\widetilde \bG - \bSigma_{\by}}
\]
and
\[
    \opnorm{\bV_{-K}^\top\widetilde \bG  \bV_{-K}} \leq 1 + \opnorm{\widetilde \bG - \bSigma_{\by}}. 
\]
This implies that if $\lambda_K > 4\opnorm{\widetilde \bG - \bSigma_{\by}}$, then 
\begin{equation}
    \label{Eq:lambdamin}
    \lambda_{\min} (\bV_K^\top \widetilde\bG \bV_K) - \opnorm{\bV_{-K}^\top \widetilde \bG \bV_{-K}} > \lambda_K / 2.
\end{equation}
In the following, we derive an exponential tail bound for $\opnorm{\widetilde \bG - \bSigma_{\by}} =\opnorm{\widetilde \bG - \E^{\bOmega} \widetilde \bG}$.
Let $\bA_i := \widetilde \by_i \widetilde \by_i^\top \circ \widetilde \bW$ and note that $\opnorm{\bA_i} \le \linfinity{\by_i}^2 \opnorm{\widetilde \bW}$.  Thus, for any $\bv = (v_1,\ldots,v_d)^\top \in \cS^{d - 1}$ and any integer $q \geq 2$, 
\begin{align*}
	\E^{\bOmega} \bigl(\bv^\top  |\bA_i|^ q \bv\bigr) & \le \E^{\bOmega}\bigl( \bigl\| \bA_i \bigr\|_{\text{op}}^{q - 2} \bv^\top \bA^2_i\bv\bigr) \leq \E^{\bOmega} \bigl\{\bigl(\opnorm{\widetilde \bW}\linfinity{\by_i} ^2 \bigr)^ {q - 2} \bv^\top \bigl(\widetilde \by_i\widetilde \by_i^\top \circ \widetilde \bW \bigr) ^ 2\bv\bigr\} \\
             			& = \opnorm{\widetilde \bW}^ {q - 2} \E^{\bOmega}\bigl\{\linfinity{\by_i} ^ {2(q - 2)} \bv^\top \bD_{\widetilde \by_i} \widetilde \bW \bD_{\widetilde \by_i}\bD_{\widetilde \by_i}\widetilde \bW \bD_{\widetilde \by_i}\bv\bigr\} \\
	& = \opnorm{\widetilde \bW}^ {q - 2} \E^{\bOmega} \bigl\{\linfinity{\by_i} ^ {2(q - 2)} \Tr( \bD^2_{\widetilde \by_i}\widetilde \bW \bD_{\widetilde \by_i}\bv\bv^\top \bD_{\widetilde \by_i}\widetilde \bW )\bigr\} \\
	&= \opnorm{\widetilde \bW}^ {q - 2} \sum_{j = 1}^d \omega_{ij}\E^{\bOmega} \bigl\{ y^2_{ij}\linfinity{\by_i} ^ {2(q - 2)} \bigl(\widetilde\bW_j^\top \bD_{\bv} \widetilde \by_i \bigr)^2\bigr \}. 
\end{align*}
	Now, for each $j \in [d]$, and $q \geq 2$, 
\begin{align*}
\E^{\bOmega} \bigl\{  y^2_{ij}\linfinity{\by_i} ^ {2(q - 2)} (\widetilde\bW_j^\top \bD_{\bv} \widetilde \by_i)^2\bigr \}  & =  \E^{\bOmega} \bigl[ y^2_{ij}\linfinity{\by_i} ^ {2(q - 2)} \{(\widetilde \bW_j \circ \bv \circ \bomega_i)^\top \by_i\}^2 \bigr]  \\
		& \leq \bigl(\E y^8_{ij} \bigr)^{1 / 4}\bigl\{\E \bigl(\linfinity{\by_i} ^ {8(q - 2)}\bigr) \bigr\}^{1 / 4} 8R\tau^2 \ltwonorm{\widetilde \bW_j \circ \bv \circ \bomega_i} ^ 2 \\
		& \lesssim MR\tau^2 \{8(q - 2)CM\log d\}^{q - 2}\sum_{t = 1} ^ d (v_t\widetilde W_{tj} \omega_{it})^ 2, 
\end{align*}
where the last inequality is due to the fact that $\bigl\|\linfinity{\by_i}\bigr\|_{\psi_2} \le (CM\log d)^{1 / 2}$ by \eqref{Eq:Psi2max}. Therefore, 
\begin{align*}
	\sum_{i = 1}^n \E^{\bOmega} \bigl(\bv^\top |\bA_i|^q \bv \bigr) & \lesssim MR\tau^2 \bigl\{8(q - 2)CM\opnorm{\widetilde \bW} \log d\bigr \}^{q - 2} \sum_{j, t = 1}^d \sum_{i = 1}^n \omega_{ij}\omega_{it}v_t^2\widetilde W^2_{tj}\\
	& = nMR\tau^2 \bigl\{8(q - 2)CM\opnorm{\widetilde \bW} \log d\bigr \}^{q - 2} \sum_{j, t = 1}^d v_t^2\widetilde W_{tj} \\
	& \lesssim q! nMR\tau^2 \|{\widetilde \bW^\top}\|_{1 \rightarrow 1} \bigl(8eCM\opnorm{\widetilde \bW}\log d\bigr)^{q - 2}, 
\end{align*}
where $\|{\widetilde \bW^\top}\|_{1 \rightarrow 1} = \sup_{\lonenorm{\bu} =1} \lonenorm{\widetilde \bW^\top \bu} = \onetoone{\widetilde \bW}$. Since the above inequality holds for all $\bv \in \cS^{d - 1}$,  we have
\[
	\biggl\|\sum_{i = 1}^n \E^{\bOmega}\bigl( |\bA_i|^q\bigr) \biggl\|_{\mathrm{op}} \lesssim q! nMR\tau ^ 2 \|{\widetilde \bW}\|_{1 \rightarrow 1} \bigl(8eCM\opnorm{\widetilde \bW}\log d\bigr)^{q - 2}.
\]
By a version of the Matrix Bernstein inequality for non-central absolute moments, which we give as Lemma~\ref{lem:nc_matrix_bernstein}, there exists a universal constant $C_1> 0$ such that for any $\xi  > 1$, 
\begin{equation}
	\label{eq:E_op}
	\PP^{\bOmega}\biggl\{ \bigl\|\widetilde \bG - \E^{\bOmega} \widetilde \bG\bigr \|_{\text{op}} \geq C_1\biggl(\biggl(\frac{MR\tau ^ 2 \onetoone{\widetilde \bW} \xi \log d}{n}\biggr)^{1 / 2} + \frac{M\opnorm{\widetilde \bW} \xi \log^2 d}{n} \biggr ) \biggl \}\leq 4d^{-(\xi - 1)}. 
\end{equation}
Now let 
\[
	\cA := \left \{\lambda_{\min} (\bV_K^\top \widetilde\bG \bV_K) - \opnorm{\bV_{-K}^\top \widetilde \bG \bV_{-K}} > \frac{\lambda_K}{2} \right \}. 
\]
From~\eqref{Eq:lambdamin} and~\eqref{eq:E_op}, we deduce that for any $\xi > 1$, if 
\begin{equation}
	\label{eq:lambda_K_lowerbound}
	\lambda_K \geq 4C_1\biggl\{\biggl(\frac{MR\tau ^ 2 \onetoone{\widetilde \bW} \xi \log d}{n}\biggr)^{1 / 2} + \frac{M\opnorm{\widetilde \bW} \xi \log^2 d}{n} \biggr \}, 
\end{equation}
then $\PP^{\bOmega}(\cA^{\mathrm{c}}) \leq \PP^{\bOmega}\bigl\{\opnorm{\widetilde \bG - \bSigma_{\by}} \geq \lambda_K/4\bigr\} \le 4d^{-(\xi - 1)}$. The desired result follows immediately by combining this with \eqref{eq:perturb} and applying \citet[Theorem~2]{YWS14}.
	\end{proof}


\section{Auxiliary results}
\label{Sec:Auxiliary}

\begin{lem}
	\label{lem:psi1_to_psi2}
	Let $X$ and $Y$ be two sub-Gaussian random variables. Then we have 		$\|X\|^2_{\psi_2} \le \|X^2\|_{\psi_1}$ and $\|XY\|_{\psi_1} \le 2\|X\|_{\psi_2} \|Y\|_{\psi_2}$.
\end{lem}
\begin{proof}
	For any $x \ge 0$, let $\lceil x\rceil := \inf\{z \in \NN: z \ge x\}$. According to the definitions of the $\psi_1$-norm and $\psi_2$-norm, we have that 
	\[
		\|X\|^2_{\psi_2} = \sup_{p \in \mathbb{N}} \frac{\E (|X|^p)^{2/p}}{p} \le \sup_{p \in \mathbb{N}} \frac{\bigl\{\E \bigl(X^{2\lceil{p / 2}\rceil}\bigr)\bigr\}^{\frac{1}{\lceil{p / 2}\rceil}}}{p} \le \|X^2\|_{\psi_1}, 
	\]
	where the penultimate inequality is due to Jensen's inequality and the last inequality is due to the fact that $p \ge \lceil p / 2\rceil$. For the second inequality, 
	\[
		\begin{aligned}
		\|XY\|_{\psi_1} & = \sup_{p \in \NN} \frac{(\E |XY|^p)^{1/p}}{p} \leq 2 \sup_{p \in \NN} \frac{(\E|X|^{2p})^{1/(2p)}}{\sqrt{2p}}\frac{(\E|Y|^{2p})^{1/(2p)}}{\sqrt{2p}}\\
		&  \le 2 \sup_{p \in \NN} \frac{(\E|X|^{2p})^{1/(2p)}}{\sqrt{2p}} \sup_{q \in \NN}\frac{(\E|Y|^{2q})^{1/(2q)}}{\sqrt{2q}} \leq 2\|X\|_{\psi_2}\|Y\|_{\psi_2},
		\end{aligned}
              \]
              as required.
\end{proof}
\begin{lem}
\label{lem:psi1bernstein}
If $X_1,\ldots,X_n$ are independent centred random variables with $\max_{i \in [n]} \|X_i\|_{\psi_1} < \infty$, then there exists a universal constant $C>0$ such that
\[
\biggl\|\sum_{i=1}^n X_i \biggr\|_{\psi_1} \leq C \biggl(\sum_{i=1}^n\|X_i\|_{\psi_1}^2\biggr)^{1/2}.
\]
\end{lem}
\begin{proof}
Write $K_i := \|X_i\|_{\psi_1}$ and $\bK:=(K_1,\ldots,K_n)^\top$. From \citet[Lemma~5.15]{Ver10}, there exist universal constants $c_1, C_1 >0$ such that for $|t| \leq c_1/\|\bK\|_\infty$,
\[
\mathbb{E}\exp\biggl\{t\sum_{i=1}^n X_i\biggr\} = \prod_{i=1}^n \mathbb{E}\exp\{t X_i\} \leq \exp\bigl\{C_1t^2\|\bK\|_2^2\bigr\}.
\]
Setting $t = \min\{C_1^{-1/2}\|\bK\|_2^{-1}, c_1\|\bK\|_\infty^{-1}\}$ in the above expression, the right-hand side is bounded above by $e$. The desired result follows from the fact that (5.15) and (5.16) in \citet{Ver10} are two definitions that yield equivalent $\psi_1$-norms.
\end{proof}
The following lemma provides a variant of the existing matrix Bernstein inequality \citep[][Theorem~6.2]{Tro12}.  The primary difference is that we impose non-central absolute moment inequalities, as opposed to central moment inequalities.  We believe that this inequality may be of independent interest, with applications beyond the scope of this paper.  Recall that if $\bA \in \mathbb{S}^{d \times d}$, with eigendecomposition $\bA = \bQ \, \mathrm{diag}(\mu_1,\ldots,\mu_d) \bQ^\top$ for some orthogonal $\bQ \in \mathbb{R}^{d \times d}$, then we write $|\bA| := \bQ \, \mathrm{diag}(|\mu_1|,\ldots,|\mu_d|) \bQ^\top$.
\begin{lem}[Matrix Bernstein inequality with non-central absolute moment conditions]
	\label{lem:nc_matrix_bernstein}
	Let $\{\bX_i\}_{i \in [n]}$ be independent symmetric $d\times d$ random matrices.  Assume that 
	\[
		\E \bigl(|\bX_i|^q\bigr) \preceq \frac{q!}{2} R^{q - 2} \bA^2_i \quad \text{for}~q = 2, 3, 4, \ldots
	\]
	for some $R > 0$ and deterministic $d$-dimensional symmetric matrices $\{\bA_i\}_{i \in [n]}$. Define the variance parameter 
	\[
		\sigma^2 := \biggl\| \sum_{i=1}^n \bA^2_i\biggr\|_{\mathrm{op}}. 
	\]
	Then for each $t > 0$,
	\[
		\PP \biggl[ \lambda_{\max} \biggl\{ \sum_{i=1}^n (\bX_i - \E \bX_i)\biggr\} \ge t\biggr] \le 4d \exp\biggl( \frac{-t^2 / 32}{\sigma^2 + Rt}\biggr). 
              \]
\end{lem}
\begin{proof}
  Let $\widetilde{\bX}_1,\ldots,\widetilde{\bX}_n,\epsilon_1,\ldots,\epsilon_n$ be independent random matrices and variables, independent of $(\bX_1,\ldots,\bX_n)$, satisfying $\widetilde{\bX}_i \stackrel{d}{=} \bX_i$ and $\epsilon_i\sim U(\{-1,1\})$ for $i\in [n]$. Write $\bS_n := \sum_{i=1}^n (\bX_i-\mathbb{E}\bX_i)$ and $\widetilde{\bS}_n := \sum_{i=1}^n (\widetilde{\bX}_i-\mathbb{E}\bX_i)$. Given $\bX_1,\ldots,\bX_n$, let $\bv_* = \bv_*(\bX_1,\ldots,\bX_n)$ be a leading unit-length eigenvector of $\bS_n$.  Let $\tilde{\bv}_1,\ldots,\tilde{\bv}_d$ denote orthonormal eigenvectors of $\widetilde{\bX}_1$ with corresponding eigenvalues $\tilde{\mu}_1,\ldots,\tilde{\mu}_d$; fix $\bv \in \mathcal{S}^{d-1}$, and let $w_j := (\tilde{\bv}_j^\top \bv)^2$ for $j \in [d]$.  Since $\sum_{j=1}^d w_j = 1$, we have by Jensen's inequality that for $q \in \{2,3,\ldots\}$,
  \[
    |\bv^\top \widetilde{\bX}_1 \bv|^q = \biggl|\sum_{j=1}^d w_j \tilde{\mu}_j\biggr|^q \leq \sum_{j=1}^d w_j |\tilde{\mu}_j|^q = \bv^\top |\widetilde{\bX}_1|^q \bv.\]
  We deduce that $\mathbb{E}\bigl\{(\bv^\top \widetilde{\bX}_i \bv)_+^q\bigr\} \leq \mathbb{E}\bigl\{|\bv^\top \widetilde{\bX}_i \bv|^q\bigr\} \leq \frac{q!}{2}R^{q-2}\bv^\top \bA_i^2\bv$ for $i \in [n]$, so by Bernstein's inequality \citep[][Corollary~2.11]{BLM13}, 
\[
\mathbb{P}\bigl(\bv_*^\top \widetilde{\bS}_n \bv_* > t/2 \bigm | \bX_1,\ldots,\bX_n\bigr) \leq \exp\biggl(\frac{-t^2/8}{\bv_*^\top \sum_{i=1}^n \bA_i^2 \bv_* + Rt}\biggr) \leq \exp\biggl(\frac{-t^2/8}{\sigma^2+Rt}\biggr).
\]
We may assume that the right-hand side of the above inequality is at most $1/2$, since otherwise the lemma is trivially true. Therefore,
\begin{align}
\mathbb{P}\{\lambda_{\max}(\bS_n) \geq t\} &= \mathbb{P}(\bv_*^\top \bS_n \bv_* \geq t) \leq 2\mathbb{E}\bigl\{\mathbb{P}\bigl(\bv_*^\top \widetilde{\bS}_n \bv_* \leq t/2 \bigm | \bX_1,\ldots,\bX_n\bigr)\mathds{1}_{\{\bv_*^\top \bS_n \bv_* \geq t\}}\bigr\}\nonumber\\
&= 2\mathbb{P}\bigl(\bv_*^\top \widetilde{\bS}_n \bv_* \leq t/2 \text{ and } \bv_*^\top \bS_n \bv_* \geq t\bigr) \leq 2\mathbb{P}(\bv_*^\top (\bS_n - \widetilde{\bS}_n) \bv_* \geq t/2)\nonumber\\
&\leq 2\mathbb{P} \biggl[\lambda_{\max} \biggl\{\sum_{i=1}^n \epsilon_i(\bX_i - \widetilde{\bX}_i)\biggr\} \geq t/2\biggr] \leq 4\mathbb{P}\biggl\{ \lambda_{\max}\biggl(\sum_{i=1}^n \epsilon_i\bX_i\biggr) \geq t/4\biggr\},
\label{eq:symmetrisation}
\end{align}
where we have used the fact that $\epsilon_i(\bX_i-\widetilde{\bX}_i) \stackrel{d}{=}\bX_i-\widetilde{\bX}_i$ for all $i$ in the penultimate inequality.

Since $\E (\epsilon_i\bX_i) = \bzero$ and $\mathbb{E}\bigl\{(\epsilon_i\bX_i)^q\bigr\} \preceq \mathbb{E}(|\bX_i|^q) \preceq \frac{q!}{2}R^{q-2}\bA_i^2$ for $q \in \{2,3,\ldots\}$, applying the matrix Bernstein inequality \citep[][Theorem~6.2]{Tro12} to the sequence $\{\epsilon_i\bX_i\}_{i \in [n]}$ yields
	\[
		\PP\biggl\{\lambda_{\max}\biggl(\sum_{i=1}^n \epsilon_i \bX_i \biggr) \geq t/4\biggr\} \le  d \exp\biggl( \frac{-t^2 / 32}{\sigma^2 + Rt}\biggr). 
	\]
	We attain the conclusion by combining the above inequality with \eqref{eq:symmetrisation}.
\end{proof}
\begin{lem}
\label{lem:bernsteinoulli_moments}
Let $X_1,\ldots,X_n$ be independent $\textup{Bin}(d, p)$ random variables and let $\widehat p_i := X_i / d$. When $dp \geq 1$ and $n \geq 2$, we have 
\[
	\E \max_{i \in [n]} \widehat p_i \leq  10p \log n.
\]
\end{lem}

\begin{proof}
	By Bernstein's inequality \citep[Lemma~2.2.9]{vanderVaartWellner1996} and a union bound, 
	\[
		\mathbb{P}\Bigl( \max_{i\in[n]} \widehat p_i \geq p+ t \Bigr) \leq  n\exp\biggl( - \frac{dt^2}{2(p + t/3)}\biggr). 
	\]
	Setting $t_0 := 2\sqrt{pd^{-1}\log n} + \frac{4}{3d}\log n$, we have
	\[
	\mathbb{E}\max_{i\in[n]} \widehat p_i = p+t_0+\int_{t_0}^\infty n\{e^{-dt^2/(4p)}+e^{-3dt/4}\}\,dt\leq  p+t_0+\sqrt\frac{\pi p}{d} + \frac{4}{3d} \leq  10p\log n,
	\]
	where we have used $\log n\geq \log 2$ and $1/d\leq  p$ in the final inequality.
\end{proof}

\begin{lem}
	\label{lem:7}
	Suppose that $ \bbeta, \betta \in \RR^d$ and $\ltwonorm{\betta} = \ltwonorm{\bbeta}$.  Let $\bSigma_1 := \bI_d + \bbeta\bbeta^\top$ and $\bSigma_2 := \bI_d + \betta\betta^\top$. Then 
	\[
		\mathrm{KL}\bigl(N_d(\bzero, \bSigma_1), N_d(\bzero, \bSigma_2)\bigr) = \frac{\|\betta\|_2^{4} - (\betta^\top\bbeta)^2}{2(1+ \|\betta\|_2^2)}.
	\]
\end{lem}
	
\begin{proof}
Since $\ltwonorm{\betta} = \ltwonorm{\bbeta}$, the matrices $\bSigma_1$ and $\bSigma_2$ share the same set of eigenvalues. Hence $|\bSigma_1| = |\bSigma_2|$ and we have
\[
	\KL(N_d(\bzero, \bSigma_1), N_d(\bzero, \bSigma_2)) = \frac{1}{2}\bigl\{\Tr\bigl(\bSigma_2^{-1}\bSigma_1\bigr) - d \bigr\}  = \frac{1}{2}\bigl\{ \Tr\bigl((\bI_d + \betta\betta^\top)^{-1}(\bI_d  + \bbeta\bbeta^\top)\bigr) - d\bigr\}.
\]
Now, by the Sherman--Morrison formula, 
\[
	(\bI_d + \betta\betta^\top)^{-1} = \bI_d - \frac{\betta\betta^\top}{1+ \|\betta\|_2^2}
\]
and thus we have
\[
	\begin{aligned}
		\KL(N(\bzero, \bSigma_1), N(\bzero, \bSigma_2)) & = \frac{1}{2} \biggl [\Tr\biggl(\biggl( \bI_d - \frac{\betta\betta^\top}{1+ \|\betta\|_2^2}\biggr) (\bI_d + \bbeta\bbeta^\top) \biggr) - d \biggr ] \\
		& = \frac{1}{2} \biggl(\|\bbeta\|_2^2 - \frac{\|\betta\|_2^2}{1 + \|\betta\|_2^2} - \frac{(\betta^\top\bbeta)^2}{1 + \|\betta\|_2^2} \biggr )= \frac{\|\betta\|_2^4 - (\betta^\top\bbeta)^2}{2(1 + \|\betta\|_2^2 )},
	\end{aligned}
\]
as required.
	\end{proof}

Theorem~\ref{thm:init_2toinf} and Proposition \ref{prop:init_f} exhibit bounds on $\mathcal{T}(\widetilde \bV_K, \bV_K)$ and $L(\widetilde \bV_K, \bV_K)$ given a deterministic observation scheme. To provide some intuition on the size of these bounds under the $p$-homogeneous missingness setting described in Section~\ref{Eq:HomogeneousTheory}, the following lemma derives probabilistic bounds for various norms of $\widetilde \bW$. 
\begin{lem}
\label{Lemma:LotsofNorms}
Assume (A5).  Then there exists an event of probability at least $1-d^2e^{-3np^2/32}$ on which each of the following inequalities hold:
\begin{enumerate}[noitemsep, label=\textrm{(\roman*)}]
\item  $\|\widetilde \bW\|_{\mathrm{op}} \leq 2dp^{-2}$;
\item $\|\widetilde \bW\|_{1\to 1} = \|\widetilde \bW\|_{\infty \to \infty} \leq 2dp^{-2}$;
\item $\|\widetilde \bW\|_{1} \leq 2d^2p^{-2}$;
\item $\|\widetilde \bW\|_{\mathrm{F}} \leq 2dp^{-2}$;
\item $\|\widetilde \bW\|_{2\to\infty} \leq 2d^{1/2}p^{-2}$.
\end{enumerate}
\end{lem}

\begin{proof}
Define an event
\[
\cA := \bigl\{\|\widetilde \bW - p^{-2}\bone_d \bone_d^\top\|_\infty \leq p^{-2}\bigr\}.
\]
For $j,k\in[d]$, write $\widehat P_{jk} := n^{-1} \sum_{i = 1}^n \omega_{ij}\omega_{ik}$. Then by a union bound and Bernstein's inequality \citep[][Proposition~2.14]{Wainwright2019},  we have
\[
\PP(\cA^{\mathrm{c}}) \leq \sum_{j,k\in [d]} \PP\bigl(\widehat P_{jk} < p^2/2\bigr) \leq d^2 e^{-3np^2/32}.
\]
Note that on $\cA$, we have $\|\widetilde \bW\|_\infty \leq 2p^{-2}$. The desired bounds then follow respectively from the following inequalities: $\|\widetilde \bW\|_{\mathrm{op}} \leq d\|\widetilde \bW\|_\infty$, $\|\widetilde \bW\|_{1\to 1} = \|\widetilde \bW\|_{\infty\to\infty} \leq d\|\widetilde \bW\|_\infty$, $\|\widetilde \bW\|_1 \leq d^2 \|\widetilde \bW\|_\infty$, $\|\widetilde \bW\|_{\mathrm{F}} \leq d\|\widetilde \bW\|_\infty$ and $\|\widetilde \bW\|_{2\to\infty} \leq d^{1/2}\|\widetilde \bW\|_\infty$. 
\end{proof}

\end{document}